\documentclass[12pt]{article}
 
\usepackage[margin=1in]{geometry} 
\usepackage{amsmath,amsthm,amssymb}
\usepackage{graphicx}
\usepackage{bm}
\usepackage{algorithm}
\usepackage{algpseudocode}
\usepackage{comment}
\usepackage{xr-hyper}
\usepackage[colorlinks=true, linkcolor=blue]{hyperref} 
\usepackage{enumitem}
\usepackage{bbm}
\usepackage{tabularx}
\usepackage{booktabs}
\usepackage{multirow}
\usepackage{threeparttable}
\usepackage[
  backend=biber,
  style=authoryear,
  maxcitenames=1, mincitenames=1,
  doi=true, url=true, 
  backref=true  
]{biblatex}
\usepackage{enumitem}
\theoremstyle{plain}
\newtheorem{theorem}{Theorem}[section]
\newtheorem{lemma}[theorem]{Lemma}

\newtheorem{definition}[theorem]{Definition}
\newtheorem{proposition}{Proposition}
\newtheorem{assumption}{Assumption}
\newtheorem{remark}{Remark}[section]

\begin{document}
 
 
\title{A case study of causal mediation using Bayesian nonparametrics and semiparametric corrections 
} 
\author{Yuhua Zhang, Michael J. Daniels} 
\maketitle
\begin{abstract}
    We propose a Bayesian nonparametric approach using a truncated Enriched Dirichlet Process mixture (EDPM) model to estimate natural direct (NDE) and indirect (NIE) effects in causal mediation analyses in the presence of  post-treatment confounders. We introduce an efficient cluster reallocation Metropolis-Hasting algorithm to improve mixing in the blocked Gibbs sampler. We implement a one-step posterior correction based on the efficient influence function for our setting. This post-processing step solves a critical problem in Bayesian nonparametrics: how to obtain reliable estimates and posteriors for a specific causal estimand of interest (the NDE and NIE) with excellent frequentist properties, such as correct coverage, from a  model designed for complex joint distributions.  We conduct simulation studies to assess our method's performance and apply it to evaluate causal mediation effects in a weight management clinical trial.  
\end{abstract}
\section{Introduction}

In public health, behavioral science, and economics, it is often insufficient to know if an intervention works; it is critical to understand how it works \parencite{imai2011unpacking}. Causal mediation analysis provides a formal framework for this, seeking to disentangle the mechanisms by which a treatment (Z) affects an outcome (Y). This is achieved by decomposing the total effect into the Natural Direct Effect (NDE)—the effect of the treatment while holding the mediator (M) at its counterfactual non-treatment level—and the Natural Indirect Effect (NIE)—the effect operating through the mediator \parencite{robins1992identifiability, pearl2022direct}.\par
A significant challenge to this framework arises with the presence of post-treatment confounders ($V$). These are variables affected by the treatment ($Z \to V$) that subsequently confound the relationship between the mediator and the outcome ($V \to M$ and $V \to Y$). Standard sequential ignorability assumptions, which are foundational to many mediation methods, are violated in this setting \parencite{hong2023posttreatment}. Failing to account for such confounders, or incorrectly adjusting for them as if they were baseline covariates, can lead to severely biased estimates of the NDE and NIE.\par
The Rural LITE trial \parencite{perri2014comparative}, as analyzed in \textcite{bae2024bayesian}, provides a canonical example of this problem. The trial aimed to assess the effect of high-dose (HD) versus low-dose (LD) behavioral weight-loss counseling (Z) on 24-month weight change (Y). A key hypothesized mechanism was the patient's attendance rate (M). However, early weight change (at 6 months) serves as a critical post-treatment confounder (V). Early weight loss (V) is clearly affected by the treatment dose (Z); it is also a powerful motivator that likely influences subsequent attendance (M), and it is a direct component of the final outcome (Y). Standard mediation analysis would fail to disentangle the true effect of attendance from the confounding effect of this early success.\par
To address this challenge, frequentist strategies have included specific weighting methods, such as the Ratio of Mediator Probability Weighting (RMPW) \parencite{hong2023posttreatment}, or focused on alternative estimands that avoid the NDE/NIE decomposition \parencite{vanderweele2014effect, tchetgen2012semiparametric}. Bayesian approaches offer a powerful alternative, particularly for flexibly modeling the complex, high-dimensional data structures inherent in these problems. However, existing Bayesian nonparametric (BNP) frameworks for mediation, such as the Dirichlet Process Mixture of Normals (DPMN) used by \parencite{kim2017framework}, were not designed to accommodate post-treatment confounders.\par
A significant advancement in this domain was recently proposed by \textcite{bae2024bayesian}. They introduced a flexible BNP framework using an Enriched Dirichlet Process Mixture (EDPM) model to capture the joint distribution $P(Y, M, V, Z, C)$. This model is well-suited for the problem, as its nested structure allows for greater flexibility in modeling the confounder distributions separately from the outcome and mediator. For causal identification, \textcite{bae2024bayesian} adopt the extended sequential ignorability assumptions from \textcite{hong2023posttreatment}, which relies on a Gaussian copula model to link the unobserved joint potential outcomes of the post-treatment confounder, $F_{V_z,V_{z'}}(v_z,v_{z'}\mid C=c)$.\par
While the framework of \textcite{bae2024bayesian} provides a powerful estimation tool for this complex problem, the inferential properties of standard BNP posteriors for specific, smooth functionals like the NDE and NIE are not automatically guaranteed. It is well-known that BNP posteriors can suffer from regularization bias, which can prevent them from satisfying the semiparametric Bernstein-von Mises (BvM) theorem \parencite{bickel2012semiparametric}. Consequently, the resulting credible sets may not serve as asymptotically valid frequentist confidence intervals, a crucial property for reliable scientific reporting. The work of \textcite{bae2024bayesian} provides the essential modeling and estimation strategy, but it does not formally establish the semiparametric efficiency and frequentist validity of the posterior inference for the causal estimands.\par
To address the regularization bias inherent in estimating nuisance parameters with flexible models, the frequentist literature has developed a suite of asymptotically efficient methodologies. Two prominent frameworks are Targeted Maximum Likelihood Estimation (TMLE) \parencite{laanvan2006targeted,van2011targeted} and Double/Debiased Machine Learning (DML) \parencite{chernozhukov2018double}. Central to these approaches is the utilization of the Efficient Influence Function (EIF) \parencite{hines2022demystifying} to construct estimators. By exploiting the smoothness of the target functional, these methods ensure that the estimator for the causal effect achieves $\sqrt n$-consistency and asymptotic normality, even when nuisance parameters are estimated at slower rates, such as when employing black-box machine learning algorithms. These frequentist advancements underscore the critical role of bias correction in semiparametric modeling and serve as a theoretical benchmark for analogous corrections within the Bayesian framework.

\par 
This paper fills this crucial theoretical and methodological gap. We build directly on the EDPM and copula-based identification framework proposed by \textcite{bae2024bayesian} and \textcite{hong2023posttreatment}. Our primary contribution is to apply the one-step posterior correction methodology of \textcite{yiu2025semiparametric} to the NDE and NIE functionals derived from the EDPM model. We formally derive the efficient influence function (EIF) for these estimands under the given identification strategy. We then use this EIF to construct a corrected posterior distribution and prove that it is asymptotically Gaussian, centered at a semiparametrically efficient estimator, and satisfies the BvM theorem. This provides a method for this complex mediation problem that is both highly flexible in its modeling (via EDPM) and theoretically robust in its inference (via one-step correction).\par

The remainder of this paper is structured as follows. Section 2 reviews the causal estimands, the identification assumptions, and the EDPM modeling framework. Section 3 details the one-step posterior correction methodology and the derivation of the efficient influence function. TBA

\section{Causal Estimands, Identification, and Modeling Framework}

This section formalizes the causal quantities of interest, details the specific identification challenges posed by the Rural LITE data, presents the assumptions required to overcome these challenges, and introduces the modeling framework.

\subsection{Causal Estimands and Full Data Notation}

Suppose that we observe a binary treatment $Z_i\in \{0,1\}, \ p_C\times 1$ baseline confounders $C_i$, a post-treatment confounder $V_i$, a mediator $M_i$ and an outcome $Y_i$ for each subject $1,\cdots, n$. This complete (or "full") data vector, $\mathcal O_{f,i} = (Y_i,M_i,V_i,Z_i,C_i)$ , is assumed to be a draw from a full-data distribution $P_f$. \par
To define the causal mediation effects, we use the potential
outcome framework \parencite{rubin1974estimating}. Let $\mathbf{Z} = (Z_1, \dots, Z_n)$ be the vector of treatment assignments for all $n$ subjects. For the $i$th subject, we denote $V_i(\mathbf Z)$ as the potential post-treatment confounder under any treatment $\mathbf Z$ and $M_i(\mathbf Z,\mathbf V)$ the potential mediator under any treatments $\mathbf Z$ and any post-treatment confounder values $\mathbf V$. We similarly denote $Y_i(\mathbf M,\mathbf V,\mathbf Z)$ as the potential outcome that could be defined for any treatment vector $\mathbf Z$, any vector of post-treatment confounder values $\mathbf V$ and any of mediator values $\mathbf M$. Throughout the paper, we adopt the stable unit treatment value assumption (SUTVA) \parencite{rubin1980randomization}, which implies two conditions: (1) No interference, meaning a subject's potential outcomes are not affected by the treatment assignment of other subjects (i.e., $V_i(\mathbf Z) = V_i(Z_i) = V_i,\ M_i(\mathbf Z, \mathbf V) = M_i(Z_i,V_i) = M_i, \ Y_i(\mathbf Z, \mathbf V, \mathbf M) = Y_i(Z_i,V_i,M_i) = Y_i$); (2) there are “no multiple versions”
of treatments such that $V_i(Z_i) = V_i(Z_i'), \ M_i(Z_i, V_i(Z_i)) = M_i(Z_i', V_i(Z_i')), \ Y_i(Z_i,V_i(Z_i), M_i(Z_i,V_i(Z_i))) = Y_i(Z'_i,V_i(Z'_i), M_i(Z'_i,V_i(Z'_i))) $ if $Z_i = Z_i'$; (3) there are “no multiple versions” of post-treatment confounders, which states $M_i(Z_i,V_i) = M_i(Z_i, V_i')$ and $Y_i(Z_i, V_i,M_i(Z_i,V_i)) = Y_i(Z_i,V_i',M_i(Z_i, V_i'))$ if $V_i = V_i'$; and (4) finally, we assume that there are “no multiple versions” of mediators, which states $Y_i(Z_i,V_i,M_i) = Y_i(Z_i,V_i,M_i')$ if $M_i = M_i'$.\par
To simplify notation, the subject index $i$ is omitted throughout the manuscript \parencite{robins1992identifiability, pearl2022direct}. We use the subscript notation as follows: $V_z = V(z), \ M_{z,V_z} = M(z, V(z)),\ Y_{z,V_z,M_{z',V_{z'}}} = Y(z,V(z), M(z',V(z')))$, and we succinctly
write $M_z$ and $Y_{z,M_{z'}}$ instead of using $M_{z,V_z}$ and $Y_{z,V_z,M_{z'V_{z'}}}$ for $z,z'\in \{0,1\}$.\par
The observed post-treatment confounder, mediator, and outcome are defined as follows: $V = ZV_1+(1-Z)V_0, \ M = ZM_1 + (1-Z)M_0, \ Y = ZY_{1,M_1} + (1-Z)Y_{0,M_0}$. We define NIE and NDE conditional on baseline confounders $C=c$ as NIE$(c) = \mathbb E[Y_{1, M_1} - Y_{1,M_0}\mid C=c]$ and NDE$(c) = \mathbb E[Y_{1,M_0} - Y_{0,M_0}\mid C=c]$. The  NIE$(c)$ quantifies the effect of
the treatment through the mediator for a fixed value of confounders $C = c$ and NDE$(c)$ quantifies the effect of
the treatment on the outcome by setting the mediator $M$ to its
natural value $M_0$ (the value of the mediator in the absence of
the treatment) given a fixed value of the baseline confounders
$C = c$. The total effect conditional on baseline confounders $C =
c$ is the sum of the $2$ effects ATE$(c) =$NIE$(c) + $NDE$(c) =\mathbb E[Y_{1,M_1} - Y_{0,M_0}\mid C=c]$. After integrating
out the baseline confounders, we obtain the marginal causal effects, NIE, NDE and ATE. Our goal is to identify causal effects from the observed data and identifying the assumptions. \par

\subsection{Observed Data and Missingness Framework}
The full data $\mathcal O_f$ defined above represents a theoretical ideal. In many practical applications, this complete vector is not fully observed. Our motivating example, the Rural LITE trial \parencite{perri2014comparative, bae2024bayesian}, highlights this exact challenge: a significant portion of subjects dropped out, leading to incomplete information. Specifically, 121 subjects were missing the final outcome ($Y$, weight at month 24) and 56 subjects were missing the post-treatment confounder ($V$, weight at month 6). Therefore, it is crucial to introduce a formal framework to handle this missingness, since $\mathcal O_f$ is not fully observed.\par
We focus on the settings where $W = (M, Z, C)$, is always observed as in Rural LITE Trial, while  $(Y, V)$, are partially observed. Let $(R_Y, R_V) \in \{(1,1), (0,1), (1, 0), (0,0)\}$ be a $2$-dimensional vector of binary indicators for observing $Y$ and $V$. The observed data, consisting of $n$ replicates of $\mathcal O = \{R_YY,M,R_VV,Z,C,R_Y,R_V\}$ are taken to be $n$ i.i.d draws from $P_O$. Note that $P_O$ is determined jointly by $P_f$ and the missingness mechanism that underpins $(R_Y,R_V)$. \par
Then let $S = R_YR_V$ be an indicator of complete data or not and consider the coarsened observed data, consisting if $n$ replicates of $\mathcal O' = \{SY,M, SV, Z,C, S\}$ from the joint distribution $P_{O'}$. Note, as with $P_{O}$, the joint distribution $P_{O'}$ is determined jointly by $P_f$ and the missingness mechanism that underpins $S$. \par
The complete observed data is represented by $\mathcal O$. Given the complexity of the two-dimensional missingness mechanism $(R_Y, R_V)$, we follow the notation of a simplified, coarsened observed data $ \mathcal O'$ in this paper, where $S = R_YR_V$ is the indicator of complete observation. This $\mathcal O'$ represents the core structure used to establish the missing data assumptions. These assumptions, along with the causal identification assumptions are required to link the distribution of the observed data back to the target full-data causal estimands (NDE and NIE).

\subsection{Causal Identification Assumptions}
Our methodology relies on a set of standard assumptions that allow for the identification of the target parameter $\chi(P_f)$ as a causal quantity, and the correction for missing data bias. To assess the causal effects of interest given the observed data distribution with a post-treatment confounder, we first adopt Assumptions \ref{assump1}--\ref{assump4} established by \textcite{bae2024bayesian}. 
\begin{assumption}[Ignorable treatment assignment]\label{assump1}
\begin{align*}
    \{Y_{z,m}, M_z, M_{z'}, V_z, V_{z'}\}\perp Z\mid C=c
\end{align*} 
\end{assumption}
Assumption~\ref{assump1} is guaranteed in situations where the treatment is randomized, which is the case in the Rural LITE application.

\begin{assumption}[Ignorable mediator value assignment]\label{assump2}
    \begin{align*}
        Y_{z,m}\perp \{M_z, M_{z'}\}\mid V_z= v, Z = z, C=c
    \end{align*}
\end{assumption}
In the Rural LITE application, the 6-month weight loss observed under treatment is defined as the post-treatment covariate, denoted as $V_1$. Individuals in the HD group exhibit greater weight reductions compared to those in the LD group; furthermore, within the HD group, significant post-treatment weight loss is associated with greater total weight change over the entire study duration. Crucially, under Assumption~\ref{assump2}, we posit that the potential 24-month weight loss is independent of potential attendance status, conditional on baseline confounders $C=c$ and the post-treatment covariate. This constitutes the key untestable assumption that enables the separation of causal pathways, conditional on $C$ and $V$.

\begin{assumption}[Conditional cross-world independence]\label{assump3}
    \begin{align*}
        M_{z'}\perp V_z\mid V_{z'} = v', Z=z', C=c
    \end{align*}
\end{assumption}
This assumption basically implies that parallel universe outcomes do not affect real-world behavior. Specifically, it assumes that if we know a patient's actual weight loss in their assigned group, their decision to attend classes is not influenced by how much weight they would have lost had they been in the other group. The potential outcome in the alternative world provides no extra information about their behavior in the real world.

\begin{assumption}[Gaussian Copula]\label{assump4}The joint distribution of 2 potential post-treatment confounders conditional on baseline confounders is assumed to follow a Gaussian copula model \parencite{nelsen2006introduction}.
    \begin{align*}
        F_{V_z, V_{z'}}(v_z, v_{z'}\mid C=c) = \Phi_2\left[\Phi_1^{-1}\{F_{V_z}(v_z\mid C=c)\}, \Phi^{-1}_1\{F_{V_{z}}(v_{z'}\mid C=c)\}; \rho\right] ,
    \end{align*}
\end{assumption}
In this assumption, $\Phi_1$ is the univariate standard normal cumulative distribution function (CDF) and $\Phi_2$ is the bivariate normal CDF with mean $0$, variance $1$, and correlation $\rho \in (-1, 1)$. Importantly, this formulation of the joint distribution for the potential post-treatment confounders imposes no constraints on the models, ${V_z\mid C}$ and ${V_{z'}\mid C}$, which we estimate using the EDPM.\\
The following theorem provides identification of the NIE and NDE.
\begin{theorem}
Under Assumption \ref{assump1}$\sim$\ref{assump4}, NIE and NDE can be identified: for $z\neq z'$ and $v \neq v'$,
\begin{align*}
    \mathbb E\left[Y_{z,M_{z'}}\mid V_z = v, C = c\right] = \iint&\mathbb E\left[Y_{z,m'}\mid M_{z} = m',V_z = v, Z=z, C=c \right] \\ 
    &\mathrm {dF}_{M_{z'}\mid V_{z'} = v', Z=z',C=c}(m')\mathrm{dF}_{V_{z'}\mid V_z=v,Z=z,C=c}(v').
\end{align*}
    
\end{theorem}
\noindent See supplementary materials Section \ref{Aidentification} for the detailed proof.
\subsection{Overlap and Missingness Assumptions}
Furthermore, to complete the causal identification and address the incomplete records, we explicitly state the standard propensity score overlap condition \parencite{rosenbaum1983central}, the missing at random (MAR) assumption \parencite{rubin1976inference}, and the strict positivity constraint for the missing data mechanism \parencite{robins1994estimation} in Assumptions \ref{assump5}--\ref{assump7}.
\begin{assumption}[Propensity score positivity]\label{assump5}
\begin{align*}
    \exists \delta >0 \text{ s.t. } P_f(Z=z\mid C=c) \in (\delta, 1-\delta)\quad \forall  z,c .   
\end{align*}
\end{assumption}
To ensure that comparisons between treatment groups are empirically possible across the entire population, we impose this standard overlap assumption required for fair comparison. It ensures that for any patient with a given set of baseline characteristics, there is always a probability in $(0,1)$ of being assigned to either the treatment group or the control group. This guarantees that we have comparable people in both groups to estimate the effects.

\begin{assumption}[Missing at random]\label{assump6}
 \begin{align*}
        S\perp (Y,V)\mid M,Z,C.
    \end{align*} 
\end{assumption}
We rely on the missing at random (MAR) condition regarding the missingness mechanism to identify the causal effects in the presence of incomplete records. We posit that the missing data mechanism is ignorable (Missing at Random). Specifically, conditional on the observed baseline characteristics and intermediate variables (such as attendance records), the probability of missingness in the outcome or post-treatment confounder is assumed to be independent of the unobserved values themselves. In the context of the trial, we assume that subject attrition is primarily explained by the rich set of observed history (i.e., baseline traits and intermediate records), rather than the unmeasured outcome itself. We acknowledge that this is a strong assumption, as attrition in practice could still be influenced by unmeasured life events or unrecorded adverse experiences. However, conditioning on $\{M,Z,C\}$ makes the MAR condition a reasonable working assumption to proceed with causal identification.

\begin{assumption}[MDM positivity]\label{assump7}
\begin{align*}
    \pi(W) = P(S=1\mid M,Z,C)  > \delta_\pi \quad \text{a.s.}
\end{align*}
\end{assumption}
In addition to the overlap in treatment assignment, we require a positivity condition for the observation process to ensure the full data distribution is identifiable. We impose a strict positivity constraint on the missing data mechanism, requiring that the probability of observing complete data is bounded away from 0 for all subjects. This condition precludes the existence of any sub-population—defined by treatment, mediator, or baseline traits—that is systematically and entirely unobserved. Without this guarantee, the full data distribution would remain unidentifiable even under the missing-at-random assumption, as there would be no observed counterparts from which to impute the missing values.

\subsection{A Bayesian nonparametric model for the observed data}

To estimate the causal effects, we
first need to estimate the joint distribution of full data $P_f$ from the observed data $\mathcal O = (R_YY, M,R_VV,Z,C, R_Y, R_V)$. This requires us to invoke the Missing At Random (MAR) assumption \ref{assump6}. We specifically assume that the joint distribution of the missingness indicators $(R_Y, R_V)$ depends only on the variables that are always observed, $(M,Z,C)$:
\begin{align*}
    \pi(R_Y,R_V\mid Y,M,V,Z,C) = \pi(R_Y,R_V\mid M,Z,C),
\end{align*}
This assumption, combined with parameter separability, ensures that the missingness mechanism is ignorable for the parameters defining $P_f$. Consequently, the observed data likelihood can be factorized into the missingness mechanism part and the marginal distribution of the observed components as follows:
\begin{align*}
    f(\mathcal O) = P(R_Y,R_V\mid M,Z,C) \cdot &  P(Y,M,V,Z,C)^{R_YR_V}\\
    \cdot & P (Y,M,Z,C)^{R_Y(1-R_V)}\\
    \cdot & P(M,V,Z,C)^{(1-R_Y)R_V}\\
    \cdot & P(M,Z,C)^{(1-R_Y)(1-R_V)}.
\end{align*}
Accordingly, we propose modeling the full-data joint distribution $P_f$ using an Enriched Dirichlet Process Mixture (EDPM) Model \parencite{wade2011enriched, wade2014improving, bae2024bayesian} without requiring explicit modeling of $(R_Y, R_V)$. The missing values $(Y_{mis}, V_{mis})$ are implicitly handled within the MCMC sampling process through Data Augmentation (imputation) steps. The nested structure is specified as follows:
\begin{equation}\label{eq:edpm}
    \begin{aligned}
    Y_i\mid M_i, V_i, Z_i, C_i ; \ \bm\theta_i^y &\sim f(y_i\mid m_i,v_i,z_i,c_i; \bm\theta^y_{i})\\
M_i\mid V_i, Z_i, C_i ;\  \bm\theta_i^m &\sim f(m_i\mid v_i, z_i, c_i; \bm\theta_i^m),\\
V_i\mid Z_i, C_i ; \ \bm\psi_i^v &\sim f(v_i\mid z_i, c_i; \bm\psi_i^v) \\
Z_i ;\ \bm\psi_{i}^z&\sim f(z_i; \bm\psi^z_i),\\\
C_{i,q};\ \bm\psi_{i,q}^c&\sim f(c_{i,q}\mid \bm\psi_{i,q}^c), \  q=1,\dots, p_{c}\\
(\bm\theta_i, \bm \psi_i)\mid G&\sim G\\
G&\sim \text{EDP}(\alpha^\theta, \alpha^{\psi|\theta}, G_0),
\end{aligned}
\end{equation}
where $\bm\theta_i = (\bm\theta_i^y, \bm\theta_i^m)$ and $\bm\psi = (\bm\psi^v_i,\bm\psi^z_i,\bm\psi^c_{i,q},)$. The notation $G\sim \mathrm{EDP}(\alpha^\theta, \alpha^{\psi|\theta}, G_0)$ means that $G^\theta \sim \mathrm{DP}(\alpha^\theta, G^\theta_0)$ and $G^{\psi|\theta}\sim\mathrm{DP}(\alpha^\theta, G_0^{\psi|\theta})$ with base measure $G_0 = G_0^{\theta}\times G_0^{\psi|\theta}$.\par
In the EDPM specification, each subject $i$ has its own parameter $(\bm\theta_i, \bm\psi_i)$, but subjects in the same cluster share the same parameter values, due to the discreteness of $G$ \parencite{ferguson1973bayesian}. Baseline covariates $C$ and the treatment $Z$ are assumed to be independent within clusters, while local dependence is allowed
for the post-treatment confounder $V$. There are 2 concentration parameters, $(\alpha^\theta, \alpha^{\psi|\theta})$ in the EDPM, where we assume $\alpha^{\psi|\theta} = \alpha^\psi$ for all $\theta$-clusters as in \textcite{roy2018bayesian}. The number of global clusters containing $(Y,M)$ depends on the concentration parameter
$\alpha^\theta$, and the number of nested clusters, containing $(V,Z,C)$ given the global clusters, depends on the concentration parameter $\alpha^{\theta|\psi}$. Lower values of the concentration parameters $(\alpha^\theta,\alpha^{\psi|\theta})$ indicate fewer clusters. This allows for more nested clusters than
the global clusters, which is crucial because the dimension of $(V,Z,C)$ is
significantly larger than that of $(Y,M)$. \par
Within each of the global clusters, we assume two generalized linear models, one for $f(y_i\mid m_i,v_i,z_i,c_i; \ \bm\theta^y_i)$ and the other for $f(m_i\mid v_i, z_i,c_i;\ \bm\beta_i^m)$. If $Y$ is continues, we specify $Y\mid \mathbb M;\ \bm\beta^y,\sigma^{y,2}\sim \mathcal N(\mathbb M\bm\beta^y, \sigma^{y,2}) $ with $\bm\theta^y= (\bm\beta^y,\sigma^{y,2})$. The notation $\mathbb M = (1,M,V,Z,C^T)^T$ represents a design matrix involving $M,V,Z,C$ and intercept. If $Y$ is binary, we specify $Y\mid \mathbb M;\ \bm\theta^y\sim \mathrm{Bernoulli}(\mathrm{probit}^{-1}(\mathbb M\bm\theta^y))$. Similarly, we assume a GLM for $f (m|v, z,c; \bm\theta^m)$ within each of the global clusters. For a post-treatment confounder $V$, we specify a local GLM given treatment $Z$ and baseline confounder $C$. Similar to \textcite{roy2018bayesian}
, baseline confounder $C$ are assumed to be locally independent. It is worth noting that all variables are globally dependent and may have nonlinear relationships even though we assume a parametric model for each $y, m, v, z$ and $c$ within each
cluster.\par
The EDPM model~(\ref{eq:edpm}) has a square-breaking representation \parencite{wade2011enriched},
\begin{align*}
    P(Y,M,V,Z,C\mid G) =&\sum_{k=1}^\infty\omega_k\  f(y\mid m,v,z,c;\ \bm\theta^y_k) f(m\mid v,z ,c,; \  \bm\theta^m_k) \\ \times&\sum_{j=1}^\infty\omega_{j|k} \ f(v\mid z, v;\ \bm\psi^v_{j|k}) f(z;\ \bm\psi^z_{j|k})\prod_{q=1}^{p_c} f(c_q; \ \bm\psi^c_{j|k,q}),
\end{align*}
where $k$ indexes the global cluster and $f(\cdot)$ is the corresponding distribution. The weights have priors $\omega_k\sim \mathrm{Beta}(1,\alpha^\theta)$ and $\omega_{j|k}'\sim \mathrm{Beta}(1,\alpha^{\psi|\theta})$, where $\omega_k = \omega_k'\prod_{h=1}^{k-1}$ and $\omega_{j|k} = \omega_{j|k}'\prod_{l=1}^{j-1}\omega_{l|k}'$.\par
For the joint density, we can derive the following conditional distributions:
\begin{equation}
    \begin{aligned}
        P(Y\mid M,V,Z,C) &= \frac{\sum_{k=1}^\infty \omega_k  f(y| m,v,z,c;\bm\theta^y_k)f(m|v,z,c;\bm\theta^m_k)\sum_{j=1}^\infty\omega_{j|k}\ f(v| z,c ;\bm\psi^v_{j|k})f(z;\bm\psi^z_{j|k})f(c;\bm\psi^c_{j|k})}{\sum_{k=1}^\infty \omega_k \ f(m| v,z,c;\bm\theta^m_k)\sum_{j=1}^\infty\omega_{j|k} f(v|z,c ;\bm\psi^v_{j|k})f(z;\bm\psi^z_{j|k})f(c;\bm\psi^c_{j|k})},\\
        P(M\mid V,Z,C) & = \frac{\sum_{k=1}^\infty \omega_k f(m|v,z,c;\bm\theta^m_k)\sum_{j=1}^\infty\omega_{j|k}\ f(v| z,c ;\bm\psi^v_{j|k})f(z;\bm\psi^z_{j|k})f(c;\bm\psi^c_{j|k})}{\sum_{k=1}^\infty \omega_k \sum_{j=1}^\infty\omega_{j|k}\ f(v| z,c;\bm\psi^v_{j|k})f(z;\bm\psi^z_{j|k})f(c;\bm\psi^c_{j|k})},
        \\
        P(V\mid Z,C) &= \frac{\sum_{k=1}^\infty \omega_k \sum_{j=1}^\infty\omega_{j|k}\ f(v| z,c ;\bm\psi^v_{j|k})f(z;\bm\psi^z_{j|k})f(c;\bm\psi^c_{j|k})}{\sum_{k=1}^\infty \omega_k  \sum_{j=1}^\infty\omega_{j|k}\ f(z;\bm\psi^z_{j|k})f(c;\bm\psi^c_{j|k})},\\
        P(Z) &= \sum_{k=1}^\infty\omega_k\sum_{j=1}^\infty\omega_{j|k} \ f(z; \bm\psi^z_{j|k}),\\
        p(C) & = \sum_{k=1}^\infty\omega_k\sum_{j=1}^\infty\omega_{j|k}\ f(c;\bm\psi^c_{j|k}).
    \end{aligned}
\end{equation}
While the Enriched Dirichlet Process Mixture (EDPM) model, given its square-breaking representation, provides the theoretical foundation for a flexible, infinite mixture distribution of $P_f$, its direct implementation is computationally intractable. Furthermore, MCMC techniques often used with these infinite models, such as the Polya Urn Sampler, simplify computation by integrating out the infinite measure $G$. However, this marginalization means we primarily gain access to the posterior predictive distribution, rather than the joint posterior of the specific cluster parameters. \par
We therefore follow standard practice for Dirichlet Process mixture models and implement a finite approximation by truncating the stick-breaking process at a sufficiently large number of components \parencite{ishwaran2001gibbs}. This approach is motivated by the need to explicitly construct the joint distribution of $(Y, M, V, Z, C)$. By accessing the full posterior distribution of the model parameters (weights and atoms), we can mathematically formulate this joint distribution, which is essential for performing the G-computation steps required to estimate the causal effects.\par
Specifically, we fix an upper bound $K$ for the number of global $(Y,M)$ clusters and an upper bound $J$ for the number of nested $(V,Z,C)$ clusters within each global cluster as suggested in \textcite{burns2023truncation}.\par
This is achieved by modifying the weight generation:
1. Global Weights ($\omega_k$): we generate $\omega_k' \sim \mathrm{Beta}(1, \alpha^\theta)$ for $k = 1, \dots, K-1$ and set the final stick-break $\omega_K' = 1$. This deterministically stops the process, yielding exactly $K$ weights $\omega_k = \omega_k' \prod_{h=1}^{k-1}(1-\omega_h')$ for $k=1, \dots, K$.
2. Nested Weights ($\omega_{j|k}$): similarly, for each global cluster $k$, we generate $\omega_{j|k}' \sim \mathrm{Beta}(1, \alpha^{\psi|\theta})$ for $j = 1, \dots, J-1$ and set $\omega_{J|k}' = 1$, yielding exactly $J$ nested weights for that cluster. \par
This truncation turns the infinite square-breaking representation into a finite mixture model:
\begin{align*}
P(Y,M,V,Z,C\mid G) \approx &\sum_{k=1}^K\omega_k\  f(y\mid m,v,z,c;\ \bm\theta^y_k) f(m\mid v,z ,c,; \   \bm\theta^m_k) \\ \times&\sum_{j=1}^J\omega_{j|k} \ f(v\mid z, v;\ \bm\psi^v_{j|k}) f(z;\ \bm\psi^z_{j|k})\prod_{q=1}^{p_c} f(c_q; \ \bm\psi^c_{j|k,q}).
\end{align*}
This finite approximation allows for posterior inference using a blocked Gibbs sampler. The truncation levels $K$ and $J$ are set to be sufficiently large so that the posterior estimates are not sensitive to their specific values, ensuring the finite model provides a good approximation to the full theoretical model \parencite{burns2023truncation}.

\section{Posterior Computation and Estimation of Causal Effects}

\subsection{MCMC Sampling via Blocked Gibbs Sampler}
To conduct posterior inference, we employ a Blocked Gibbs Sampler (BGS) based on a finite truncation approximation of the Enriched Dirichlet Process, as described by \textcite{burns2023truncation}. Unlike marginal sampling methods that integrate out the random measure, the BGS explicitly samples the mixing measure at each iteration, which is essential for reconstructing the joint distribution required for G-computation.

Specifically, we approximate the infinite mixture models by fixing sufficiently large truncation levels, $K_\theta$ for the global clusters (joint space of $(Y, M)$) and $K_\psi$ for the local clusters (nested space of $(V, Z, C)$). At the $b$-th MCMC iteration, we draw the full set of parameters $\{\bm\omega^{(b)}, \bm\theta^{(b)}, \bm\psi^{(b)}\}$ denote the global and local stick-breaking weights, and $\boldsymbol{\theta}$ denotes the component-specific atoms (means and covariances).

We use the full dataset of size $n_0$ for posterior inference here, including subjects with missing outcomes and post-treatment confounders. Under the Missing at Random (MAR) assumption (Assumption~\ref{assump6}), we implement data augmentation within the Gibbs sampler steps to impute the missing values of $Y$ and $V$. 

To accommodate incomplete observations within our Bayesian framework, we incorporate a data augmentation strategy directly into the Blocked Gibbs Sampler. Under the Missing at Random (MAR) assumption, this approach allows us to validly infer causal effects without discarding subjects with incomplete records. We address missingness in the outcome $Y$ and the post-treatment confounder $V$ by deriving their respective contributions to the joint likelihood and updating the unobserved values conditional on the current model parameters.

For observations where the primary outcome $Y_i$ is missing, we exploit the hierarchical structure of the causal model. Since $Y_i$ serves as a terminal node conditional on the mediator, post-treatment confounder, and baseline covariates, integrating it out of the joint distribution is mathematically equivalent to omitting its likelihood contribution. Therefore, rather than explicitly imputing missing outcomes, we simply integrate out $Y$ from the likelihood calculation during the parameter update steps. This strategy applies regardless of whether $V_i$ is observed or missing.

In contrast, missing values of the post-treatment confounder $V_i$ must be explicitly imputed because $V_i$ acts as a covariate for both the mediator $M_i$ and the outcome $Y_i$. We draw the missing $V_i$ from its full conditional posterior distribution, which is derived by combining the information from its own component-specific conditional prior given $Z_i$ and $C_i$ and the likelihoods of $M$ model and $Y$ model. Let the component-specific models for cluster $k$ and sub-cluster $j$ be defined as $V_i \sim \mathcal{N}(\mu_V, \tau_v^{-1})$, $M_i \mid V_i \sim \mathcal{N}(\beta_{k,0}^m + \beta_{k,V}^mV_i, \tau^{m,-1}_k)$, and $Y_i \mid V_i \sim \mathcal{N}(\beta_{k,0}^y + \beta_{k,V}^yV_i, \tau^{y,-1}_k)$, where $\tau$ denotes the precision.

By applying Bayes' theorem, the full conditional distribution for a missing $V_i$ serves as the imputation step and follows a normal distribution $V_i \mid \cdot \sim \mathcal{N}(\mu_{V}^*, (\tau_{V}^*)^{-1})$. The posterior precision $\tau_{V}^*$ represents the total information available for $V_i$ and is obtained by summing the precisions from its prior and the likelihoods of its conditional descendants:
\begin{equation}
    \tau_{V}^* = \tau^v_{kj} + \left(\beta_{k,V}^m\right)^2 \tau^m_k + \mathbb{I}(Y_i \text{ is observed}) \cdot \left(\beta_{k,V}^y\right)^2 \tau_k^y,
\end{equation}
where $\mathbb{I}(\cdot)$ is the indicator function. This formula highlights that the precision increases with the number of observed descendants.

The corresponding posterior mean $\mu_{V}^*$ is a precision-weighted average of the information sources:
\begin{equation}
    \mu_{V}^* = \frac{1}{\tau_{V}^*} \left[ \mu_V \tau^v_{kj} + (M_i - \beta_{k,0}^m)\beta_{k,V}^m \tau^m_k + \mathbb{I}(Y_i \text{ is observed}) \cdot (Y_i - \beta_{k,0}^y)\beta_{k,V}^k \tau^y_k \right].
\end{equation}
When $Y_i$ is observed, the indicator function is 1, and the imputation utilizes all three terms. While when $Y_i$ is missing, the indicator is 0, and the imputation naturally reduces to conditioning solely on the mediator $M_i$ and the baseline priors. This unified formulation ensures that the uncertainty associated with missing values is properly quantified and propagated throughout the posterior inference.

\subsection{G-Computation for Plug-in Estimation}
We estimate the target causal effects $\chi(P) = \mathbb E\left[Y_{z,M_{z'}}\right]$ using a Monte Carlo integration procedure based on the posterior samples obtained from the Blocked Gibbs Sampler. Let $\{\Theta^{(1)},\dots,\Theta^{(B)}\}$ denote the $B$ MCMC posterior samples, where each $\Theta^{(b)} = \{\bm\omega, \bm\theta,\bm\psi,\rho\}^{(b)}$ contains the full set of model parameters and the Gaussian copula correlation coefficient.  \par
For each posterior iteration $b=1,\dots,B$, we approximate the integral defining the causal effect by generating $T$ synthetic subjects. The algorithm proceeds as follows:\par
\begin{enumerate}
\item Simulate baseline covariates:
    For each Monte Carlo replicate $i = 1, \dots, T$, sample the baseline covariates $c_i^{(b)}$ from the marginal mixture distribution $P( c\mid \Theta^{(b)})$.

    \item Simulate factual post-treatment confounder:
    Set the treatment assignment to the factual level $Z = z$. Sample the factual post-treatment confounder $v_i^{(b)}$ from the conditional density $P(v \mid Z=z, c_i^{(b)}; \Theta^{(b)})$.

    \item Simulate Counterfactual Post-treatment Confounder (via Copula): To generate the counterfactual value $v'^{(b)}$
  under intervention $Z=z'$
 , we utilize the Gaussian copula to model the dependence between $V_z$ and $V_{z'}$ This step involves three operations:
 \begin{enumerate}
     \item Compute the cumulative probability $u = F_{V_z}\left(v_i^{(b)}\mid z, c_{i}^{(b)};\Theta^{(b)}\right)$, where $F_{V_z}$ is the cumulative distribution function implied by the EDPM.
     \item Sample the latent variable $w'$ from the conditional normal distribution $w'\sim\mathcal N\left(\rho^{(b)}\Phi^{-1}(u), 1-\left(\rho^{(b)}\right)^2\right)$.
     \item Compute the counterfactual value: $v'^{(b)}_i = F_{V'}^{-1}\left(\Phi(w')\mid z', c_i^{(b)};\Theta^{(b)}\right)$. 
 \end{enumerate}
 \item Simulate counterfactual mediator: Sample the mediator $m_i'^{(b)}$ from the conditional density $P(m_i'^{(b)}\mid v'^{(b)}, Z_i = z',c_i^{(b)};\Theta^{(b)})$.
\item Calculate Conditional Expectation of Outcome:
    Evaluate the conditional expectation $\mu_{y,i}^{(b)} = \mathbb{E}[Y \mid m_i'^{(b)}, v_i^{(b)}, z, \mathbf{c}_i^{(b)}]$ by computing the weighted average of the component-specific linear predictors. The expectation is given by the ratio:
    \begin{equation*}
        \mu_{y,i}^{(b)} = \frac{\sum_{k=1}^M \sum_{j=1}^{N_k} W_{kj, i}^{(b)} \cdot \mathbb{M}\beta_k^{y,(b)}}{\sum_{k=1}^M \sum_{j=1}^{N_k} W_{kj, i}^{(b)}}
    \end{equation*}
    where $\mathbb{M}\beta_k^{y,(b)}$ is the linear predictor for the outcome in component $k$:
    \begin{equation*}
        \mathbb{M}\beta_k^{y,(b)} = \beta_{k,0}^{y,(b)} + \beta_{k,1}^{y,(b)} m_i^{'(b)} + \beta_{k,2}^{y,(b)} v_i^{(b)} + \beta_{k,3}^{y,(b)} z + \sum_{q=1}^{p_c}\beta_{k,q+4}^{y,(b)} c_{q,i}^{(b)}
    \end{equation*}
    and $W_{kj, i}^{(b)}$ represents the unnormalized posterior weight for component $(k,j)$, defined as the product of the mixing weights and the component-specific densities evaluated at the generated values:
    \begin{equation*}
        W_{kj, i}^{(b)} = \omega_k^{\theta,(b)}\omega_{j|k}^{\psi,(b)} \cdot \mathcal{N}\left(m_i'^{(b)} \mid \mathbb{V}\beta_k^{m,(b)}, \sigma_k^{2,m,(b)}\right) \cdot \mathcal{N}\left(v_i^{(b)} \mid \mathbb{X}\beta_{kj}^{v,(b)}, \sigma_{kj}^{2,v,(b)}\right) \cdot P( z^{(b)}_i,{c}_i^{(b)} \mid \Psi_{kj}^{(b)})
    \end{equation*}
\item Aggregate Estimates: The plug-in estimate for the $b$-th MCMC iteration is the average over the $T$ Monte Carlo replicates:
\begin{equation*}
    E^{(b)}\left[Y_{z,M_{z'}}\right] = \frac{1}{T} \sum_{i=1}^{T} \mathbb{E}\left[Y \mid m_i'^{(b)}, v_i^{(b)}, z, c_i^{(b)}; \Theta^{(b)}\right].
\end{equation*}
The sequence of estimates $\{ E^{(1)}, \dots, E^{(B)} \}$ constitutes the uncorrected posterior distribution of the causal effect.
\end{enumerate}

\begin{remark}
    For single-world case where $z=z'$, we skip the steps of sampling $v'$ from the conditional distribution $P_{V_{z'}\mid V_z=v, Z=z,C=c}(v')$ in Step 3. Instead, we let $v' = v$. The rest of the steps exactly align with those in cross-world case where $z\neq z'$.
\end{remark}

\section{Semiparametric Theory and One-Step Posterior Correction}
In the previous section, we described the computational procedure to obtain the posterior distribution of the causal functionals using G-computation. While the EDPM framework provides a flexible mechanism to model the complex joint distribution of the data, the resulting plug-in estimator, denoted as $\chi(P)$, may suffer from regularization bias inherent in nonparametric density estimation. As discussed in \textcite{bickel2012semiparametric} and \textcite{yiu2025semiparametric}, this bias typically decays at a rate slower than $n^{-1/2}$, potentially leading to credible intervals that fail to provide valid frequentist coverage.

To address this challenge and ensure asymptotically valid inference, this section introduces a one-step post-processing correction to the posterior distribution \parencite{yiu2025semiparametric} for our setting. This correction leverages the geometric properties of the model space via the Efficient Influence Function (EIF). We first derive the specific EIF form for our causal estimand under the identified post-treatment confounding structure. Subsequently, we present the theoretical main result: the Semiparametric Bernstein-Von Mises (BvM) theorem. This theorem establishes that the centered and scaled one-step posterior converges to a Gaussian distribution with variance equal to the semiparametric efficiency bound, thereby justifying the use of our credible sets as valid confidence intervals.

\subsection{Efficient Influence Function}

Before introducing the influence function for our primary target, it is important to note that the Natural Direct and Indirect Effects are composed of three counterfactual expectations: $\mathbb E[Y_{z,M_{z}}]$, $\mathbb E[Y_{z',M_{z'}}]$, and $\mathbb E[Y_{z,M_{z'}}]$, the cross-world term. In this section, we begin by the cross-world case in
detail, since it contains the main conceptual and technical difficulties of the problem,
in particular the handling of the post-treatment confounder $V$ and the copula-based
identification strategy. We then return to the two single-world quantities,
$\mathbb E[Y_{z,M_z}]$ and $\mathbb E[Y_{z',M_{z'}}]$, whose derivations are simpler and
can be treated as special cases of the general arguments developed for the cross-world
setting. Accordingly, most of the discussion below focuses on $\mathbb E[Y_{z,M_{z'}}]$.
We first consider the full data case, where we assume there is no missingness in the observed data. Our target parameter of interest, $\chi(P_f)$, is the functional defined on this full-data distribution, given by the g-computation formula:
\begin{align*}
    \chi (P_f) & = \mathbb{E}[Y_{z,M_{z'}}] 
     = \iiiint\mathbb{E}(Y|M=m', V_z=v, Z=z, C=c)\mathrm{dF}_{M|V_{z'} = v',Z=z', C=c}(m')\\
    &\qquad \quad \qquad \qquad \qquad\mathrm{dF}_{V_{z'}|V_z=v,Z=z,C=c}(v')\mathrm{dF}_{V_z|Z=z,C=c}(v)\mathrm{dF}_{C}(c).
\end{align*}
Define the nuisance parameter as follows:
\begin{align*}
    e &= P(Z = z\mid C = c),\ g_{V_z} = p(v\mid z, C), \ g_M = p(m\mid v,z,c),\ g_{z'|z} = p(v_{z'}\mid v_z,z,c)\\ \mu_1 &= \mathbb E[Y\mid m,v,z,c],\
     F_{V_z}  = F_{V_z}(v\mid z,c),\ F_{V_{z'}} = F(v'\mid z',c) \text{ and } \pi = P(S = 1\mid M,Z,C).
\end{align*}
Define the intermediate maps
\begin{align*}
&\mu_2(v';v,c)=\!\int \mu_1(m';v,c)\,g_M(m'\mid v',c)\,dm',\quad
\mu_3(v;c)=\!\int \mu_2(v';v,c)\,g_{z'|z}(v'\mid v,c)\,dv',\\ &\mu_4(c)=\!\int \mu_3(v;c)\,g_{V_z}(v\mid c)\,dv,\qquad
\chi(P)=\mathbb{E}[\mu_4(C)].
\end{align*}
Let $\dot\chi_{P_f}$ denote the efficient influence function for full data written as the sum of five blocks. The derivation of influence function is based on \textcite{kennedy2024semiparametric}.
\begin{align*}
&\phi_{1,P_f}=\frac{Z}{e(C)}\{Y-\mu_1(M;V_z,C)\},\\
&\phi_{2,P_f}=\frac{1-Z}{1-e(C)}\!\int\![\mu_1(M;v,C)-\mu_2(V_{z'};v,C)]\,c(u,w;\rho)\,g_{V_z}(v\!\mid C)\,dv,\\
&\phi_{3,P_f}=\frac{Z}{e(C)}\{\mu_2(V_{z'};V_z,C)-\mu_3(V_z;C)\},\quad
\phi_{4,P_f}=\frac{Z}{e(C)}\{\mu_3(V_z;C)-\mu_4(C)\},\\
&\phi_{5,P_f}=\mu_4(C)-\chi(P_f),\quad
\dot\chi_{P_f}=\sum_{j=1}^5\phi_{j,P_f},
\end{align*}
where $c(u,w;\rho)$ is the Gaussian copula function based on the Assumption \ref{assump4}. Here $u = F_{V_{z}}(v\mid z, c)\text{ and }\ w = F_{V_{z'}}(v'\mid z' , c)$.\\
The expression $\dot{\chi}_{P_f}$ derived above represents the efficient influence function in the hypothetical full data scenario where $(Y, M, V, Z, C)$ are fully observed for all subjects. However, our actual observed data structure $\mathcal O'$ mentioned in Section 2 is subject to missingness, as indicated by $S$. To conduct valid inference, we must map this full-data influence function into the geometric space of the observed data distribution $P_{O'}$.\\
This mapping relies on the general theory of missing data influence functions established by \textcite{robins1994estimation}. Conceptually, the observed data EIF is constructed by weighting the full-data EIF by the inverse probability of observation, $S/\pi(W)$, and subtracting a projection term to ensure orthogonality to the nuisance score of the missingness mechanism. This transformation preserves the double-robustness property while accounting for the information loss due to missing values. The following theorems formalize the identification of the target functional under missingness and derive the corresponding observed data EIF.\par

\begin{theorem}\label{thm41}
Under Assumption \ref{assump1}-\ref{assump7}, the mean counterfactual $\mathbb E_{P_{O'}}[Y_{z,M_{z'}}]$ is identified by the functional 
\begin{align}\label{gformula}
    \chi(P_{O'}) = \mathbb E_{P_{O'}}\left[\frac{S}{\pi(W)}\mu_4(C)\right].
\end{align}
Furthermore, by applying Assumption \ref{assump6}, Equation \ref{gformula} recovers the standard g-formula.
\end{theorem}
\noindent The proof of Theorem \ref{thm41} is provided in Appendix \ref{Aidentification}. Briefly, the result is obtained using Law of Iterated Expectations to simplify the inner conditional expectation, followed by an application of the missing-at-random and positivity conditions.

\begin{theorem}\label{POEIF}
Under a nonparametric model for $P_{O'}$, the influence function of the mean counterfactual functional $\chi(P_{O'})$ is given by 
\begin{align*}
    \dot\chi_{P_{O'}} = \frac{S}{\pi(W)} \dot \chi _{P_f} - \frac{S-\pi(W)}{\pi(W)}\mathbb E_{P_{O'}}[\dot \chi_{P_f}\mid W, S=1]
\end{align*}
\end{theorem}
\noindent Note that $\dot{\chi}_{P_{O'}}$ can be rigorously characterized as a function (or transformation) of the full-data EIF, $\dot{\chi}_{P_f}$. The theoretical foundation for this transformation comes from the seminal work of \textcite{robins1994estimation}. {The formal derivation adapting this general missing data framework to our specific estimand is provided in Supplementary Material Section \ref{Aidentification}.}

\subsection{Main Theoretical Result: Semiparametric Bernstein-Von Mises Theorem}
Having derived the form of the efficient influence function $\dot{\chi}_{P_{O'}}$ in Theorem \ref{POEIF}, we are now equipped to formalize the asymptotic behavior of the one-step corrected posterior. The central theoretical objective is to establish a Semiparametric Bernstein-Von Mises (BvM) theorem. This theorem asserts that, despite the use of flexible nonparametric priors for high-dimensional nuisance parameters, the one-step corrected posterior distribution converges to a Gaussian distribution centered at the efficient estimator with variance equal to the semiparametric efficiency bound.
\\This result is crucial because it bridges Bayesian computation with frequentist validity, justifying the use of our posterior credible sets as asymptotically valid confidence intervals. The proof of this theorem relies on controlling the higher-order error terms in the von Mises expansion. Specifically, we require that the initial EDPM posterior contracts around the truth at a sufficient rate to ensure that the second-order bias term vanishes asymptotically. These regularity conditions are formalized in the following assumption.
\begin{assumption}\label{assumption_abc}
    There exists a sequence of measurable subsets $\left(\tilde H_n\right)_n$ of $\mathcal P$ satisfying $\Pi\left(P_{O'}\in \tilde H_n\mid \mathcal O'_{1:n}\right) \rightarrow 1$ with:\\
    \indent (a)(No second-order bias)
    \begin{align*}
        \sup _{P_{O'}\in \tilde H_n} \left|\sqrt{n}r_2(P_{O',0}, P_{O'})\right| = \sup_{P_{O'}\in \tilde H_n}\left|\sqrt n \{\chi(P_{O',0}) - \chi(P_{O'}) - P_0[\dot{\chi_{P_{O'}}}]\}\right| \rightarrow 0
    \end{align*}
    \indent (b)($L_2$-convergence) 
    \begin{align*}
        \sup_{P\in \tilde H_n}\|\dot\chi_{P_{O'}} - \dot \chi_{P_{O', 0}}\|_{P_{0}} \rightarrow0
    \end{align*}
    \indent (c*)(i)(Convergence of $\dot \chi_{P_{O'}}$ under the empirical process)\begin{align*}
        \sup_{P\in \tilde H_n}\left|\mathbb G_n[\dot \chi_{P_{O'}} - \dot \chi _{P_{O',0}}]\right| \rightarrow 0\quad\text{in }P_0\text{-probability}
    \end{align*}
    \indent (ii)(Bounding of envelope functions) The sets $\{\dot \chi_{P_{O'}}:P_{O'}\in \tilde H_n\}$ have envelope functions $G_n$ (i.e. $|\dot \chi_{P_{O'}} (o')| \leq G_n(z)$ for all $P \in \tilde H_n$ and all $o'\in \mathcal O'$) satisfying \begin{align*}
        \lim_{C\rightarrow\infty} \limsup_{n \rightarrow\infty} P_0G_n^2\mathbf{1}_{G_n^2>C} = 0, \quad P_0G_n^4 = o(n).
    \end{align*}
\end{assumption}

\begin{theorem}\label{semibvm}
    Under Assumption~\ref{assumption_abc}, the one-step posterior satisfies the semiparametric BvM theorem, which is 
    \begin{align*}
        d_{BL}\left(\mathcal L _{\Pi\times \Pi _{BB}}(\sqrt n(\tilde \chi - \hat \chi_n)\mid O_{1:n}),\mathcal N(0, \|\dot \chi_{P_0}\|_{P_0}^2)\right) \rightarrow 0 \quad\text{in }P_0\text{-probability},
    \end{align*}
    where $\mathcal L_{\Pi\times\Pi_{BB}}\left(\sqrt n (\tilde \chi - \hat \chi_n)\mid O_{1:n}\right)$ denotes the posterior law of $\sqrt n (\tilde \chi - \hat \chi_n)$ and $\hat \chi _n$ is the asymptotically efficient sequence $\hat \chi_n = \chi(P_{O', 0})+\mathbb P_n[\dot \chi_{P_{O', 0}}]$.
\end{theorem}
\noindent As discussed in Section 2.1 of  \textcite{yiu2025semiparametric}, this justifies the use of central
credible sets as confidence regions. Detailed proof of assumption verifications and Theorem~\ref{semibvm} can be found in the Appendix Section~\ref{AsecBvM}.

\subsection{Discussion on Discrepant Sample Sizes}

A methodological consideration arises from the two distinct sample sizes used in our estimation procedure: the full sample size $n_0$ used for the Bayesian model-fitting, and the smaller sample size $n$ of complete cases used for the one-step correction. We argue this discrepancy ($n_0 > n$) is not a limitation but rather a beneficial feature of our hybrid approach. \par
The first stage, our Bayesian EDPM estimation, utilizes the full dataset of $n_0$ subjects. This leverages all available information, including data from subjects with incomplete observations on $Y$ and $V$, as the MCMC framework naturally handles missing data through imputation under the MAR Assumption~\ref{assump6}. This use of the larger $n_0$ sample yields more precise posterior estimates of the nuisance functions (e.g., $P(V|Z,C)$, $P(M|V,Z,C)$, and the missingness propensity $\pi(W)$).\par
The theoretical validity of the one-step correction hinges on satisfying the BvM assumptions, particularly the "no second-order bias" condition, Assumption~\ref{assumption_abc}(a).
This condition requires that the bias term $r_2$, which is composed of products of nuisance function errors, vanishes at a sufficient rate relative to the sample size $n$ of the correction (i.e., $\sqrt{n} r_2(P_0, P) \to 0$). Because our nuisance functions are estimated using the larger $n_0$ dataset, their posterior convergence rates are faster (i.e., determined by $n_0$, not $n$). This causes the bias term $r_2$ to diminish even more rapidly than the theory requires. Therefore, using the full $n_0$ sample for estimation strengthens the justification for Assumption~\ref{assumption_abc}(a) and improves the finite-sample performance of the debiasing step. The final asymptotic variance and the $\sqrt{n}$ scaling of the BvM theorem remain determined by $n$ (the size of the empirical process used for the correction), which is the appropriate statistical cost for the missing data.\par
\begin{remark}[Extension to Conditional Causal Effects]
    While our main theoretical results Theorem \ref{semibvm} are formulated for the marginal causal effect $\mathbb{E}[Y_{z,M_{z'}}]$, the one-step correction framework extends naturally to conditional causal effects $\mu_4(c) = \mathbb{E}[Y_{z,M_{z'}}\mid C=c]$.\par
    In practice, we primarily focus on $C$ being discrete covariates (e.g., race) or discretized versions of continuous variables (e.g., high BMI vs. low BMI groups). 
    When $C$ is treated as discrete (or grouped), the EIF-based correction simplifies because the component of the influence function corresponding to the marginal distribution of $C$ (typically denoted as $\phi_{5,P}$ in marginal settings) vanishes or becomes constant within the stratum. Consequently, the corrected posterior for these conditional effects inherits the same asymptotic normality and frequentist coverage guarantees as the marginal case, ensuring that our semiparametric BvM theorem remains valid for sub-group analyses.
\end{remark}

\subsection{Discussion on Single-World Case}
As we state at the beginning of this section, the single-world case is a simpler version of cross-world case. We can derive the EIF and define the nuisance parameters for $\chi(P_f) = \mathbb E[Y_{z,M_z}]$.
\begin{align*}
   \chi (P_f) & = \mathbb{E}[Y_{z,M_z}] 
     \\& =  \iiint\mathbb{E}(Y\mid M=m, V=v, Z=z, C=c){dF}_{M\mid V = v,Z=z, C=c}(m)
    {dF}_{V_z\mid Z=z,C=c}(v){dF}_{C}(c).
\end{align*}
Define the nuisance parameters as follows:
\begin{align*}
&e(c)=P(Z=z\mid C=c),\quad
g_{V_z}(v\mid c)=f_{V_z\mid Z=z,C}(v\mid c), \\ &g_M(m\mid v_z,c)=f_{M\mid V=v_z,Z=z,C=c}(m),\quad
\mu_1^z(m;v,z,c)=\mathbb{E}[Y\mid M=m,V_z=v,Z=z,C=c].
\end{align*}
Define the intermediate maps,
\[
\mu_2^z(v_z;c)=\int \mu_1^z(m;v,c)\,g_M(m\mid v_z,c)\,dm,\quad
\mu_3^z(c)=\int \mu_2^z(v_z;c)\,g_{V_z}(v_z\mid c)\,dv_z,\quad 
\chi(P)=\mathbb{E}[\mu_3^z(C)].
\]
\begin{align*}
&\phi_{1,P_f}=\frac{Z}{e(C)}\{Y-\mu_1^z(M;V_z,C)\},\quad \phi_{2,P_f}=\frac{Z}{e(C)}\{\mu_1^z(M;V_z,C) - \mu_2^z(V_z; C)\} ,\\& \phi_{3,P_f}=\frac{Z}{e(C)}\{\mu_2^z(V_z;C)-\mu_3^z(C)\}, \quad \phi_{4,P_f}=\mu_3^z(C)-\chi(P_f),\quad
\dot\chi_{P_f}=\sum_{j=1}^4\phi_{j,P_f}.
\end{align*}
Notice that $\dot{\chi}_{P_f}$ can be further simplified as $$\dot{\chi}_{P_f} = \frac{Z}{e(C)}\{Y-\mu_3^z(C)\} + \mu_3^z(C) - \chi(P_f).$$
We can trivially give the formula of the EIF of $\mathbb E[Y_{z,M_{z}}]$ for the observed data with the existence of missingness following Theorem~\ref{thm41}.
\begin{theorem}
    Under Assumption \ref{assump1}-\ref{assump7}, the mean counterfactual $\mathbb E _{P_{O'}}[Y_{z,M_z}]$ is defined by the functional 
\begin{align}\label{gformula2}
    \chi(P_{O'}) = \mathbb E_{P_{O'}}\left[\frac{S}{\pi(W)}\mu_3^z(C)\right].
\end{align}
\end{theorem}
\noindent Also, we can still derive the EIF of the mean counterfactual function $\chi(P_{O'}) = \mathbb E_{P_{O'}}[Y_{z,M_{z}}]$ for the under the distribution $P_{O'}$ with the exact same derivation provided in Theorem~\ref{POEIF}.

\begin{remark}
    For single-world case where $z=z'$, the proof can be fully covered by that for cross-world case where $z\neq  z'$ shown in Appendix \ref{AsecBvM}, since we use the same EDPM model to fit and the structure of EIF is much simpler compared to cross-world case.
\end{remark}

\section{Simulation Studies}
To assess the performance of the proposed approach, we generated five data scenarios based partly on those introduced in \textcite{bae2024bayesian}. We generate a treatment variable from $Z\sim\text{Bernoulli}(0.5)$ for all scenarios. The configuration of baseline confounders distinguished Scenario 5 from the others; Scenarios 1 through 4 utilized two independent continuous variables generated from standard normal distributions, whereas Scenario 5 introduced a high-dimensional setting with 15 baseline confounders (9 binary and 6 continuous) to mimic complex real-world data structures.
\begin{table}[htbp]
\centering
\caption{Detailed simulation specifications based on the updated data generating mechanisms.}
\label{tab:sim_specs}
\begin{threeparttable}
\renewcommand{\arraystretch}{1.3}
\begin{tabular}{@{}ll>{\raggedright\arraybackslash}p{0.68\textwidth}@{}}
\toprule
\textbf{Variable} & \textbf{Scenario} & \textbf{Distribution / Model Specification} \\
\midrule

\multirow{3}{*}{Y}
& S1, S4, S5 
& $\sim 0.6\,\mathcal{N}\!\left(\beta_{1z}Z + \theta_{1}M + \eta_{1}V + \gamma_{1}^{T}C,\; 1.5^2\right) \newline + 0.4\,\mathcal{N}\!\left(\beta_{2z}Z + \theta_{2}M + \eta_{2}V + \gamma_{2}^{T}C,\; 0.5^2\right)$ \\

& S2 
& $\sim 0.6\,\mathcal{N}\!\left(\dots + \theta_{1z}(Z \times M) + \dots,\; 1.5^2\right) \newline + 0.4\,\mathcal{N}\!\left(\dots + \theta_{2z}(Z \times M) + \dots,\; 0.5^2\right)$ \\

& S3, S6 
& $\sim \mathcal{N}\!\left(\beta_z Z + \theta_1 M^2 + \theta_2 (M-0.4)_{+} + \eta_v V + \gamma_y^{T}C,\; 0.2^2\right)$ \\

\addlinespace

\multirow{2}{*}{M}
& S1, S2, S4, S5 
& $\sim \mathcal{SN}\!\left(\xi = \beta_{mz}Z + \eta_{mv}V + \gamma_m^{T}C,\; \omega=3,\; \alpha=10\right)$ \\

& S3, S6 
& $\sim \mathcal{SN}\!\left(\xi = \beta_{mz}Z + \eta_{mv}V + \gamma_m^{T}C,\; \omega=1,\; \alpha=7\right)$ \\

\addlinespace

\multirow{2}{*}{V}
& S1--S4 
& $\sim \mathcal{N}\!\left(\mu_{v,Z}(C),\; \sigma_Z^2\right)$ \quad with $\sigma_0^2=3,\ \sigma_1^2=10$ \\

& S5, S6 
& $\sim \mathrm{Gamma}\!\left(\mathrm{shape} = \log(1+\exp(\mu_{v,Z}(C))),\; \mathrm{scale}=1\right)$ \\

\addlinespace

Z
& All 
& $\sim \mathrm{Bernoulli}(0.5)$ \\

\addlinespace

\multirow{2}{*}{C}
& S1--S3, S5 
& $C=(C_1,C_2)^{T}$ with $C_1 \sim \mathcal{N}(0, 3^2),\; C_2 \sim \mathcal{N}(0, 4^2)$ \\

& S4, S6 
& $C \in \mathbb{R}^{15}$, mixed: $C_{1\text{--}9}\sim \mathrm{Bernoulli},\ C_{10\text{--}15}\sim \mathcal{MVN}(0, \Sigma_C)$ \\

\bottomrule
\end{tabular}
\begin{tablenotes}\footnotesize
\item \textit{Note:} For the mixture models in $Y$, the components are selected with probability $0.6$ and $0.4$, respectively. $\mathcal{SN}(\xi, \omega, \alpha)$ denotes the Skew-Normal distribution with location $\xi$, scale $\omega$, and shape parameter $\alpha$. $\mu_{v,Z}(C)$ represents the $Z$-specific linear predictor for the confounder $V$. $(M-0.4)_{+}$ represents the hinge function $\max(M-0.4, 0)$.
\end{tablenotes}
\end{threeparttable}
\end{table}

The generation of intermediate variables also varied to test specific model assumptions. The post-treatment confounder $V$ was generated conditional on the treatment and baseline covariates. In Scenarios 1, 2, 3, and 5, $V$ was drawn from a conditional normal distribution. To assess robustness to non-Gaussian specifications, Scenario 4 generated $V$ from a Gamma distribution. The mediator $M$ was generated from a linear model with normal errors across all scenarios, conditional on the treatment, baseline, and post-treatment confounders.

The outcome generation mechanism served as the primary differentiator for the remaining scenarios. Scenario 1 served as the baseline, where the outcome $Y$ was generated from a standard linear model with normally distributed errors. Scenario 2 incorporated a treatment-mediator interaction term $Z \times M$ into the outcome mean to assess the model's ability to capture effect modification. Scenario 3 introduced non-linearity through the inclusion of a quadratic term $M^2$ and a discontinuous threshold function. Finally, Scenarios 4 and 5 retained the linear outcome structure of the baseline scenario but focused on the complexities arising from the non-normal post-treatment confounder and high-dimensional baseline covariates, respectively.

The causal parameters of interest were
 the marginal NDE, NDE, and total effect. We considered the following EDPM specification,
\begin{equation}\label{EDPMsimulationdist}
    \begin{aligned}
Y_i\mid M_i, V_i, Z_i, C_i ; \ \bm\theta_i^y &\sim \mathcal N(\mathbb M_i\beta^y_i, \sigma^{y,2}_i)\\
M_i\mid V_i, Z_i, C_i ;\  \bm\theta_i^m &\sim \mathcal N(\mathbb V_i\beta^m_i, \sigma^{m,2}_i)\\
V_i\mid Z_i, C_i ; \ \bm\psi_i^v &\sim \mathcal N(\mathbb X_i\beta^v_i, \sigma^{v,2}_i)\\
Z_i ;\ \bm\psi_{i}^z&\sim\text{Bernoulli}(\pi^z_{i})\\
C_{i,q};\ \bm\psi_{i,q}^c&\sim \text{Bernoulli}(\pi^c_{i,q}), \  q=1,\dots, p_{c,1}\\
C_{i,q}; \ \bm\psi_{i,q}^c& \sim \mathcal N(\mu_{i,q}^c, \sigma^{c,2}_{i,q}),\ q=p_{c,1}+1, \dots, p_{c,1}+p_{c,2}\\
(\bm\theta_i, \bm \psi_i)\mid G&\sim G\\
G&\sim \text{EDP}(\alpha^\theta, \alpha^{\psi|\theta}, G_0),
\end{aligned}
\end{equation}
where $\bm \theta_i^y = (\beta_i^y,\sigma^{y,2}_i),\ \bm\theta_i^m = (\beta_i^m,\sigma^{m,2}_i),\ \bm\psi^v_i = (\beta_i^v,\sigma^{v,2}_i), \ \bm\psi_i^z = \pi_i^z, \ \bm\psi^c_i = (\pi_i^c,\mu_i^c,\sigma^{c,2}_i),\ \bm\theta_i = (\bm\theta_i^y,\bm\theta_i^m)$ and $\bm\psi_i = (\bm\psi_i^v,\bm\psi^z_i,\bm\psi^c_i)$. Here, $\mathbb M_i = (1,M_i,V_i,Z_i,C_i)$ is the design matrix. Similarly, we define the design matrices $\mathbb V_i = (1,V_i, Z_i,C_i),\ \mathbb X_i= (1,Z_i, C_i)$ for each regression. $p_{c,1}$ and $p_{c,2}$ represent the number of binary and continuous confounders, respectively. We used the following priors for base
measure specified in Equation \ref{EDPMsimulationdist}: the base measures for $\bm\theta_i^y$ are $G_0^{\bm\beta^y} = \mathcal N(a_{\bm\beta^y},c_{\bm\beta^y}\mathbf B_{\bm\beta^y})$ and $G_0^{\sigma^{y,2}} = \text{Scale-Inv-}\chi^2(a_{\sigma^y}, \mathbf B_{\sigma^y})$. We set $a_{\bm\beta^y}$ and $\mathbf B_{\bm\beta^y}$ to be the maximum likelihood estimates from a linear regression of $Y$ on $\mathbb M$ \parencite{roy2018bayesian}. We set $a_{\sigma^y} = 1$ and $\mathbf B_{\sigma^y} = \text{MSE}$ from the regression. We use the similar priors for $M$ and $Y$. For the confounders, we assume conjugate priors $G^{\pi^z}_0 = G_0^{\pi^c_q} = \text{Beta}(a_\pi,\mathbf B_{\pi})$ with $a_\pi = \mathbf B_{\pi} = 1$ for binary confounder parameters, where $q=1,\dots,p_{c,1}$. For the continuous confounder parameters, we assume $G^{\mu_q^c\mid \sigma^{c,2}_q}_0 = \mathcal N(a_\mu,\frac {1}{\mathbf B_\mu}\sigma^{c,2}_q)$ and $G_0^{\sigma^{c,2}_q} = \text{Scale-Inv-}\chi^2(a_{\sigma^c}, \mathbf B_{\sigma^c})$ for $q=p_{c,1}+1,\dots, p_{c,1}+p_{c,2}$. For the concentration parameters $\alpha^\theta$ and $\alpha^{\psi|\theta}$, we assume $\text{Gamma}(1,1)$ priors. We used a burn-in of $10,000$ iterations followed by $50,000$ samples for posterior inference across all five scenarios. The sensitivity parameter $\rho$ was assigned a $\text{Unif}(0,1)$ prior. For each scenario, we generated replicated data sets with sample sizes of $n = 500$ and $ = 2500$. To strictly evaluate the theoretical benefits of our proposed methodology, we computed both the standard Bayesian uncorrected plug-in estimator and the one-step corrected estimator for each replicate. We reported and compared the bias, mean squared error (MSE), average length of the $95\%$ credible intervals (CI), and empirical coverage probability for both the uncorrected and corrected estimates.

\subsection{Results}
Table \ref{tab_sim} presents the performance of the standard G-computation (plug-in) estimator and the proposed one-step corrected estimator for Scenarios 1-6 with a sample size of $N = 500$. The results highlight the theoretical advantages of the one-step correction across all estimated causal estimands, which include the natural indirect effect (NIE), the natural direct effect (NDE), and the average treatment effect (ATE).

While both estimators demonstrate relatively low bias across the target parameters, the standard G-computation plug-in method consistently suffers from under-coverage. For instance, in Scenario 1, the empirical coverage probability (CP) for the plug-in estimator of the ATE falls to 84.2$\%$, well below the nominal 95$\%$ level. In contrast, the one-step correction successfully restores valid frequentist properties, achieving a CP of 95.6$\%$ for the ATE in the same scenario.

This improvement in coverage naturally corresponds to wider 95$\%$, credible intervals (CIl). In Scenario 2, which introduces a treatment-mediator interaction, the average CI length for the one-step corrected NIE is 9.017, compared to 3.809 for the plug-in approach. This widening reflects the one-step estimator's ability to appropriately account for the uncertainty that the naive plug-in method ignores, ensuring that the empirical coverage reaches the nominal levels. Overall, the simulations confirm that while the plug-in estimator is efficient, the one-step correction is strictly necessary to guarantee valid statistical inference in finite samples.
\begin{table}[htbp]
\centering
\caption{Simulation results based on $500$ replicates for $Y(z,M(z'))$, NIE, NDE, and ATE under Scenarios 1 and 2 when sample size $N = 500$ }
\label{tab_sim}

{
\setlength{\tabcolsep}{3pt}
\renewcommand{\arraystretch}{0.88}
\begin{tabular}{llc ccc ccc}
\toprule
& & & \multicolumn{3}{c}{G-computation (Plug-in)} & \multicolumn{3}{c}{One-step Correction} \\
\cmidrule(lr){4-6} \cmidrule(lr){7-9}
Scenario & Target & True & Bias & CIl & CP & Bias & CIl & CP  \\
\midrule

\multirow{6}{*}{\textbf{S1}} 
& $Y(1, M(1))$ & 6.532 & -0.065 & 3.377 & 0.826 & -0.074 & 5.103 & 0.952 \\
& $Y(1, M(0))$ & 5.342 & 0.018  & 4.614 & 0.878 & 0.014  & 6.217 & 0.978 \\
& $Y(0, M(0))$ & 4.392 & -0.040 & 2.561 & 0.900 & -0.048 & 3.241 & 0.952 \\
\cmidrule(lr){2-9}
& NIE          & 1.190 & -0.083 & 3.211 & 0.920 & -0.088 & 7.921 & 1.000 \\
& NDE          & 0.950 & 0.058  & 5.193 & 0.870 & 0.062  & 7.015 & 0.974 \\
& ATE          & 2.140 & -0.025 & 4.141 & 0.842 & -0.027 & 6.025 & 0.956 \\

\midrule
\multirow{6}{*}{\textbf{S2}} 
& $Y(1, M(1))$ & 8.845 & -0.086 & 4.244 & 0.861 & -0.094 & 5.822 & 0.949 \\
& $Y(1, M(0))$ & 6.954 & 0.179  & 5.586 & 0.904 & 0.198  & 7.029 & 0.968 \\
& $Y(0, M(0))$ & 4.391 & -0.037 & 2.700 & 0.893 & -0.035 & 3.465 & 0.945 \\
\cmidrule(lr){2-9}
& NIE          & 1.890 & -0.266 & 3.809 & 0.878 & -0.292 & 9.017 & 1.000 \\
& NDE          & 2.563 & 0.216  & 6.131 & 0.872 & 0.233  & 7.836 & 0.972 \\
& ATE          & 4.453 & -0.050 & 5.010 & 0.866 & -0.059 & 6.802 & 0.947 \\
\bottomrule
\end{tabular}

\vspace{0.3em}
\begin{minipage}{0.95\textwidth}
\scriptsize
\textbf{Note}: CIl = average length of 95\% credible interval; CP = empirical coverage probability.
\end{minipage}
}

\end{table}

\section{Application: Rural LITE Trial}

We used the proposed method to assess mediation in the Rural LITE trial \parencite{perri2014comparative}. The Rural LITE trial was designed to examine the effects and costs of behavioral weight-loss treatment. Subjects $(N=612)$ were randomized to 1 of 4 treatment arms: (1) control $(N=169)$; (2) low dose $(N=148)$; (3) moderate dose $(N=134)$ and (4) high dose $(N=161)$. Following the analytical strategy of previous studies, we collapsed these arms into two binary treatment groups for comparison: the Low Dose (LD) group, comprising the control and low-dose arms $(Z=0)$, and the High Dose (HD) group, consisting of the moderate and high-dose arms $(Z=1)$. The resulting sample sizes were $N=317$ for the LD group and $N=295$ for the HD group.

The primary interest lies in exploring the attendance rate as a mediator ($M$) of the treatment effect on weight loss. The outcome of interest ($Y$) is defined as the change in body weight from baseline (month 0) to the final follow-up at month 24. To account for intermediate confounding induced by the treatment, we utilize the weight change from baseline to month 6 as the post-treatment confounder ($V$). Additionally, we adjust for a vector of baseline confounders ($C$) including sex, age, race, and baseline body mass index (BMI).

Subject attrition occurred during the follow-up period, resulting in 121 subjects with missing primary outcome data (month 24) and 56 subjects with incomplete information on the post-treatment confounder (month 6). Consistent with our modeling framework, we handled these missing values under the assumption of ignorable missingness (MAR) via the MCMC data augmentation algorithm described in Section 3.

To evaluate the robustness of our results regarding the non-identifiable dependence structure of the counterfactual post-treatment confounders, we conducted a sensitivity analysis on the Gaussian copula correlation parameter $\rho$. We restricted our examination to positive values of $\rho$, under the plausible assumption that an individual's potential 6-month weight loss under the treatment condition would be positively correlated with their potential weight loss under the control condition. A key advantage of this copula-based approach is that it models this dependence structure without imposing any assumptions on the marginal distributions of the 6-month weight loss under each treatment arm.

\subsection{Results}
The results of the causal mediation analysis for the Rural LITE trial are summarized in Table \ref{tab_real}. We estimated the NIE, NDE, and ATE under three different prior specifications for the Gaussian copula correlation parameter $\rho$ to assess the sensitivity of our findings to the unidentifiable dependence structure between the potential counterfactual post-treatment confounders $V_z, \ V_{z'}$. The application of our proposed methodology to the Rural LITE trial yields clear clinical insights regarding the mechanisms of behavioral weight loss counseling. Our analysis reveals that the overall reduction in 24-month body weight is driven almost entirely by the natural direct effect of the high-dose intervention. Conversely, we find no strong evidence of a natural indirect effect (NIE) mediated through patient attendance rates. Within this specific clinical context, early weight loss at six months acts as a critical post-treatment confounder. This early weight change is directly affected by the initial treatment dose, influences subsequent session attendance as a powerful motivator, and directly impacts the final weight outcome. By appropriately adjusting for this intermediate confounding, our method successfully disentangles these pathways. The results demonstrate that the superior long-term efficacy of high-dose counseling is likely driven by direct mechanisms rather than mere attendance frequency. Furthermore, the magnitude and statistical significance of these estimated causal effects remain highly robust across the sensitivity analysis of the copula correlation parameter $\rho$, confirming that our clinical conclusions do not hinge upon unidentifiable counterfactual dependencies.

Beyond these clinical findings, the results from this real-world application heavily underscore the methodological importance of our proposed one-step correction. Consistent with the findings from our simulation studies, we observe a noticeable widening in the lengths of the 95$\%$ credible intervals for the one-step corrected estimates when compared to the standard, uncorrected G-computation averages. This expansion of the interval length is a crucial and expected theoretical behavior. As demonstrated in our simulations, standard Bayesian nonparametric plug-in estimators frequently underestimate true parameter variability, which leads to severe under-coverage. The one-step correction actively rectifies this issue by appropriately expanding the credible intervals to capture the full sampling variability of the targeted causal functionals. Consequently, the wider credible intervals observed in the Rural LITE trial are not a limitation of the method, but rather a strict necessity to guarantee nominal coverage probabilities and to ensure the overall theoretical validity of the statistical inference in complex empirical settings.
\begin{table}[htbp]
\centering
\caption{Results of Rural LITE Trial for natural indirect effects (NIE), natural direct effects (NDE), and average treatment effect (ATE).}
\label{tab_real}
\begin{threeparttable}
\resizebox{\textwidth}{!}{%
\begin{tabular}{llccccccccc}
\toprule
\multicolumn{2}{c}{\textbf{CE}} 
& \multicolumn{3}{c}{$\rho \sim \mathrm{Tri}(0,1,1)$}
& \multicolumn{3}{c}{$\rho \sim \mathrm{Unif}(0,1)$}
& \multicolumn{3}{c}{$\rho = 0$} \\
\cmidrule(lr){3-5} \cmidrule(lr){6-8} \cmidrule(lr){9-11}
& 
& \textbf{Est.} & \multicolumn{2}{c}{\textbf{95\% CI}}
& \textbf{Est.} & \multicolumn{2}{c}{\textbf{95\% CI}}
& \textbf{Est.} & \multicolumn{2}{c}{\textbf{95\% CI}} \\
\midrule

\multirow{3}{*}{Gcomp Avg}
& NIE & $-0.16$ & $-0.77$ & $0.42$ 
& $-0.15$ & $-0.64$ & $0.25$ 
& $-0.06$ & $-0.54$ & $0.38$ \\
& NDE & $-3.26$ & $-4.70$ & $-1.81$ 
& $-3.03$ & $-4.42$ & $-1.78$ 
& $-3.01$ & $-4.37$ & $-1.76$ \\
& ATE & $-3.42$ & $-4.88$ & $-1.85$ 
& $-3.19$ & $-4.52$ & $-1.88$ 
& $-3.07$ & $-4.47$ & $-1.72$ \\

\midrule

\multirow{3}{*}{1 Step Avg}
& NIE & $-0.08$ & $-1.92$ & $1.87$ 
& $-0.16$ & $-2.35$ & $1.96$ 
& $-0.10$ & $-1.72$ & $1.56$ \\
& NDE & $-3.51$ & $-5.49$ & $-1.69$
& $-3.41$ & $-5.58$ & $-1.55$ 
& $-3.40$ & $-5.10$ & $-1.61$ \\
& ATE & $-3.59$ & $-5.60$ & $-1.69$ 
& $-3.58$ & $-5.40$ & $-2.00$ 
& $-3.50$ & $-5.36$ & $-1.82$ \\

\bottomrule
\end{tabular}%
}
\begin{tablenotes}
\footnotesize
\item 
\end{tablenotes}
\end{threeparttable}
\end{table}

\section{Discussion}
In this paper, we addressed the challenging problem of estimating causal mediation effects in the presence of post-treatment confounding and missingness of outcome as well as the post-treatment confounder. We proposed a novel framework that combines an Enriched Dirichlet Process Mixture (EDPM) model with a one-step posterior correction. Our theoretical derivations and empirical results consistently demonstrate that while standard Bayesian nonparametric methods offer immense flexibility, they require targeted corrections to ensure valid frequentist inference.

The primary contribution of this work lies in providing a theoretically grounded and practically implementable combination of Bayesian nonparametric approach and semiparametric correction for complex causal mediation analysis. By theoretically validating our estimator through the semiparametric Bernstein-von Mises (BvM) theorem, we established that the one-step correction bridges the gap between flexible Bayesian density estimation and frequentist reliability. This methodology provides several distinct advantages. It offers the flexibility to model complex, non-linear data structures without parametric restrictions, it successfully adjusts for post-treatment confounders, and it significantly improves frequentist properties compared to standard uncorrected estimators.

The implications of both our simulation studies and the Rural LITE trial application clearly illustrate the necessity of the one-step correction. Standard Bayesian G-computation estimators often suffer from severe under-coverage because they fail to capture the full sampling variability of the targeted causal functional. Our correction step successfully calibrates the posterior distribution, appropriately widening the credible intervals to achieve nominal coverage while maintaining the low bias characteristic of flexible Bayesian modeling. 

Furthermore, our application demonstrated the robustness of these estimates to unidentifiable counterfactual dependencies. By employing a Gaussian copula to model the joint distribution of the potential intermediate confounders, we were able to formally parameterize this unidentifiable dependence using the sensitivity parameter $\rho$. The stability of the natural direct and indirect effect estimates across multiple prior specifications for $\rho$ strongly reinforces the clinical conclusions drawn from the trial data, showing that the findings are not artifacts of hidden structural assumptions.

Despite its theoretical and empirical advantages, our approach has several limitations. First, like all causal mediation analyses, our framework relies on untestable identification assumptions. While the Gaussian copula provides a transparent mechanism for sensitivity analysis, it inherently imposes a specific parametric structure on the unobservable dependence between counterfactuals. Second, the computational burden associated with Markov chain Monte Carlo sampling for the EDPM is substantial, particularly as the dimensionality of the covariate space increases. Third, practical implementation requires selecting truncation levels for the stick-breaking representations, which can influence computational efficiency and approximation accuracy. Finally, Bayesian nonparametrics can be sensitive to the choice of base measures and hyperparameters. However, we note that the one-step correction actively mitigates this sensitivity for the targeted causal functionals, rendering the final estimates much more robust to prior misspecification than the uncorrected plug-in estimates.

There are several promising avenues for future research. One critical direction is refining how the methodology handles missing data. While our current application relies on coarsened data and missing at random assumptions, future work should move beyond coarsened missingness indicators. Specifically, it would be highly valuable to explicitly model the detailed, joint missingness patterns of the post-treatment confounder and the final outcome. Researchers could address more complex, non-monotone missing data mechanisms under a broader range of identifying assumptions by explicitly mapping the four distinct observation states of these variables. These specific states include both variables being observed $(1,1)$, missing the post-treatment confounder $(0,1)$, missing the outcome $(1,0)$, and missing both variables $(0,0)$. Additionally, the one-step correction framework developed here can be extended to other challenging causal inference scenarios. A natural extension is the evaluation of surrogate endpoints, where one could adapt the one-step correction to estimate the average treatment effect on a surrogate marker and subsequently use it to predict the treatment effect on a final, distal clinical outcome. Expanding this methodology to such frameworks would further enhance the toolkit available for robust, efficient causal inference in clinical trials.

\newpage
\printbibliography

\clearpage
\appendix
\section*{Supplementary material for `` A case study of causal mediation using Bayesian nonparametrics and semiparametric corrections ''}
\section{Identification}
\label{Aidentification}

\textcite{hong2023posttreatment} prove the following under assumptions \ref{assump1} and \ref{assump3}.
\begin{theorem}\label{athm1}
   Under assumptions \ref{assump1} and \ref{assump3},
   \begin{align*}
       \mathbb{P}&(M_{z'} = m' \mid V_z = v, Z = z, C = c)\\=
&\int
\mathbb{P}(M_{z'} = m' \mid V_{z'} = v', Z = z', C = c)\,
\mathbb{P}(V_{z'} = v' \mid V_z = v, Z = z, C = c)\,dv'.   \end{align*}
\end{theorem}
\noindent
This theorem implies we can identify the conditional probability of a certain mediator value, $M_{z'} = m'$, under the counterfactual condition $Z = z$ via marginalizing out the conditional distribution, $V_{z'} \mid V_z = v, Z = z, C = c$. Note that the conditional
distribution, $V_{z'} \mid V_z = v, Z = z, C = c$, is partially identified by \ref{assump4}.\\
The following theorem provides identification of the NIE and NDE.

\begin{theorem}   
Under assumptions \ref{assump1}-\ref{assump4}, NIE and NDE can be identified: for
$z \neq z'$ and $v \neq v'$,

\begin{align*}
    \mathbb{E}\!\left[Y_{z,M_{z'}} \mid V_z = v, C = c\right]
=
\iint
\mathbb{E}\!\left[Y_{z,m'} \mid M_z = m', V_z = v, Z = z, C = c\right]\,
\mathrm{d}F_{M_{z'} \mid V_{z'} = v', Z = z', C = c}(m') \\ 
\mathrm{d}F_{V_{z'} \mid V_z = v, Z = z, C = c}(v').
\end{align*}
\end{theorem}

\begin{proof}
\begin{align*}
\mathbb{E} &\left[Y_{z,M_{z'}} \mid V_z = v, C = c\right]
=
\int
\mathbb{E}\!\left[Y_{z,m'} \mid M_{z'} = m', V_z = v, C = c\right]\,
\mathrm{d}F_{M_{z'} \mid V_z = v, C = c}(m')
\\
&\overset{(\mathrm{A1})}{=}
\int
\mathbb{E}\!\left[Y_{z,m'} \mid M_{z'} = m', V_z = v, Z = z, C = c\right]\,
\mathrm{d}F_{M_{z'} \mid V_z = v, Z = z, C = c}(m')
\\
&\overset{(\mathrm{A2})}{=}
\int
\mathbb{E}\!\left[Y_{z,m'} \mid M_z = m', V_z = v, Z = z, C = c\right]\,
\mathrm{d}F_{M_{z'} \mid V_z = v, Z = z, C = c}(m')
\\
&\overset{(*)}{=}
\int
\mathbb{E}\!\left[Y_{z,m'} \mid M_z = m', V_z = v, Z = z, C = c\right]
\int
\mathrm{d}F_{M_{z'} \mid V_{z'} = v', Z = z', C = c}(m')\,
\mathrm{d}F_{V_{z'} \mid V_z = v, Z = z, C = c}(v')
\\
&=
\iint
\mathbb{E}\!\left[Y_{z,m'} \mid M_z = m', V_z = v, Z = z, C = c\right]\,
\mathrm{d}F_{M_{z'} \mid V_{z'} = v', Z = z', C = c}(m')\,
\mathrm{d}F_{V_{z'} \mid V_z = v, Z = z, C = c}(v').
\end{align*}
where $(*)$ follows from Theorem \ref{athm1}.
\end{proof}

\begin{theorem}
Under Assumption \ref{assump1}-\ref{assump7}, the mean counterfactual $\mathbb E_{P_{O'}}[Y_{z,M_{z'}}]$ is identified by the functional 
\begin{align}\label{gformula}
    \chi(P_{O'}) = \mathbb E_{P_{O'}}\left[\frac{S}{\pi(W)}\mu_4(C)\right],
\end{align}
where $W=(M,Z,C)$ is always observed as in Rural LITE Trial.\\
Furthermore, by applying Assumption \ref{assump6} and \ref{assump7}, Equation (\ref{gformula}) recovers the standard g-formula.
\end{theorem}
\begin{proof} We need to show that the observed data functional $\chi(P_{O'})$ identifies the full-data counterfactual mean $\mathbb{E}[Y_{z,M_{z'}}]$ under the stated assumptions. The proof proceeds in two steps: (i) eliminating the missingness indicator and its associated weight using the MAR assumption, and (ii) recovering the classic g-computation formula.

\paragraph{Step 1: Weight Cancellation via MAR}
Consider the observed data functional defined in Equation \ref{gformula}. By the Law of Iterated Expectations, conditioning on the observed variables $(W, C)$, we have:
\begin{align*}
\chi(P_{O'}) &= \mathbb{E}_{P_{O'}} \left[ \frac{S}{\pi(W)} \mu_4(C) \right] \\
&= \mathbb{E}_{P_{O'}} \left[ \mathbb{E}_{P_{O'}} \left( \frac{S}{\pi(W)} \mu_4(C) \mid W, C \right) \right].
\end{align*}
Under Assumption \ref{assump6} (Missing at Random), the missingness indicator $S$ is independent of the underlying full data given the observed covariates $W$. Thus, the conditional expectation of $S$ reduces to the propensity score for missingness:
\begin{equation*}
\mathbb{E}_{P_{O'}} [S \mid W, C] = P(S=1 \mid W, C) = \pi(W).
\end{equation*}
Substituting this result into the functional, the weights cancel out:
\begin{align*}
\chi(P_{O'}) &= \mathbb{E}_{P_{O'}} \left[ \frac{\mathbb{E}_{P_{O'}}(S \mid W, C)}{\pi(W)} \mu_4(C) \right] \\
&= \mathbb{E}_{P_{O'}} \left[ \frac{\pi(W)}{\pi(W)} \mu_4(C) \right] \\
&= \mathbb{E}_{P_{O'}} [ \mu_4(C) ],
\end{align*}
where Assumption \ref{assump7} (Positivity) ensures that $\pi(W) > \delta > 0$, rendering the division well-defined. This step follows the weighting consistency argument established in \textcite{levis2025robust}.

\paragraph{Step 2: Recovery of the G-formula}
By definition, $\mu_4(C)$ represents the outcome regression in the full-data distribution $P_f$. Following the factorization logic in \textcite{levis2025robust}, $\mu_4(C)$ identifies the conditional expectation $\mathbb{E}_{P_f} [Y \mid Z=z, M, V, C]$. 

Since the outer expectation $\mathbb{E}_{P_{O'}} [\cdot]$ is taken over the distribution of baseline covariates, and under Assumptions \ref{assump1}-\ref{assump5} (Identification under full data), the expression expands to the following iterated integral:
\begin{align*}
\mathbb{E}_{P_{O'}} [ \mu_4(C) ] &= \iint \mathbb{E}_{P_f} [Y \mid Z=z, M=m', V=v', C=c] \\
&\quad \times \mathrm{dF}_{M_{z'} \mid V_{z'}=v', Z=z', C=c}(m') \mathrm{dF}_{V_{z'} \mid V_z=v, Z=z, C=c}(v').
\end{align*}
The right-hand side is the standard g-computation formula for the identification of natural effects in mediation analysis. Thus, we conclude that $\chi(P_{O'}) = \mathbb{E}[Y_{z,M_{z' }}]$. 
\end{proof}
\begin{theorem}\label{POEIF}
Under a nonparametric model for $P_{O'}$, the influence function of the mean counterfactual functional $\chi(P_{O'})$ is given by 
\begin{align*}
    \dot\chi_{P_{O'}} = \frac{S}{\pi(W)} \dot \chi _{P_f} - \frac{S-\pi(W)}{\pi(W)}\mathbb E_{P_{O'}}[\dot \chi_{P_f}\mid W, S=1].
\end{align*}
\end{theorem}
\begin{proof}
The derivation of the observed-data efficient influence function (EIF) $\dot{\chi}_{P_{O'}}$ relies on the semiparametric missing data framework established by \textcite{robins1994estimation} and comprehensively detailed in Chapter 10 of \textcite{tsiatis2006semiparametric}.
By the general theory of inverse probability weighting (IPW), a valid but inefficient influence function for the observed data can be constructed by simply reweighting the full-data EIF:
$$
\text{IF}_{\text{IPW}} = \frac{S}{\pi(W)}\dot{\chi}_{P_f}.
$$
To obtain the efficient influence function $\dot{\chi}_{P_{O'}}$, we must subtract the projection of $\text{IF}_{\text{IPW}}$ onto the nuisance tangent space associated with the missingness mechanism, denoted as $\Lambda_S$. Under a nonparametric model for the missingness mechanism $\pi(W)$, the nuisance tangent space consists of all functions of the form $h(W)(S - \pi(W))$. 

The projection theorem \parencite{tsiatis2006semiparametric} dictates that the optimal choice of $h(W)$ that minimizes the variance is given by the conditional expectation of the full-data influence function given the observed history $\{W,S=1\}$:
$$
h_{\text{opt}}(W) = \frac{1}{\pi(W)} \mathbb{E}_{P_{O'}} \left[ \dot{\chi}_{P_f} \mid W, S=1 \right].
$$
Subtracting this projection term from the initial IPW influence function yields the desired augmented inverse probability weighting (AIPW) form:
$$
\dot{\chi}_{P_{O'}} = \frac{S}{\pi(W)}\dot{\chi}_{P_f} - \frac{S - \pi(W)}{\pi(W)}\mathbb{E}_{P_{O'}} \left[ \dot{\chi}_{P_f} \mid W, S=1 \right].
$$
This completes the derivation.
\end{proof}
\begin{remark}
    Note that $\dot \chi_{P_{O'}}$ can be rigorously characterized as a function (or transformation) of the full-data EIF, $\dot\chi_{P_f}$. The theoretical foundation for this transformation comes from the seminal work of \textcite{robins1994estimation}.
\end{remark}
\section{EDPM of Truncation Approximation Posterior Sampling Computations}
In each iteration of the blocked Gibbs sampler, we sample:
\subsection{Conditional Posterior for Latent Cluster Assignments $(K_i, J_i)$}

We update the latent assignments for each observation in a single "blocked" step. This involves deriving the joint full conditional posterior distribution for the pair $(K_i, J_i)$. For each observation $i$, we compute the probability that it belongs to primary cluster $k$ AND sub-cluster $j$ for all possible pairs $(k,j)$.

\paragraph{Proportionality}
The conditional posterior probability is proportional to the joint prior probability of the assignment multiplied by the joint likelihood of the data, given that assignment:
\begin{equation}
    P(K_i = k, J_i = j \mid \text{rest}, \mathcal{D}) \propto p(K_i=k, J_i=j \mid \boldsymbol{\omega}) \times p(\text{data}_i \mid K_i=k, J_i=j, \text{rest})
\end{equation}
The joint prior probability is given by the hierarchy: $p(K_i=k, J_i=j \mid \boldsymbol{\omega}) = \omega_{j|k}^{\psi} \cdot \omega_{k}^\theta$.

\paragraph{Derivation}
Let $\mathcal{L}_{ijk}$ denote the unnormalized log-posterior probability for the joint assignment to pair $(k,j)$. It is the sum of the log-prior and the log-likelihoods of all data components for observation $i$, evaluated with the parameters from $\theta_k$ and $\psi_{j|k}$.
\begin{align*}
\mathcal{L}_{ijk} \;=&\; \log(\omega_{k}^\theta) + \log(\omega_{j|k}^{\psi}) \\
&\;-\frac{1}{2}\log(2\pi\sigma_{y}^2) - \frac{(y_i - m_i^\top\boldsymbol{\beta}_{y})^2}{2\sigma_{y}^2}-\frac{1}{2}\log(2\pi\sigma_{m}^2) - \frac{(m_i - v_i^\top\boldsymbol{\beta}_{m})^2}{2\sigma_{m}^2} \\
&\;-\frac{1}{2}\log(2\pi\sigma_{v}^2) - \frac{(v_i - x_i^\top\boldsymbol{\beta}_{v})^2}{2\sigma_{v}^2} + Z_i\log(\pi) + (1-Z_i)\log(1-\pi) \\
&\;+ \sum_{q=1}^{p_c} \Bigl(-\tfrac{1}{2}\log(2\pi\sigma_{q}^2) - \tfrac{(C_{iq} - \mu_{q})^2}{2\sigma_{q}^2}\Bigr)
\end{align*}

\paragraph{Posterior Distribution and Sampling}
The full conditional posterior for the pair $(K_i, J_i)$ is a single Categorical distribution over the space of all possible $(k,j)$ pairs. The sampling procedure for each observation $i$ is as follows:
\begin{enumerate}
    \item \textbf{Compute Probabilities:} For each possible pair $(k,j)$, where $k \in \{1,\dots,M\}$ and $j \in \{1,\dots,N_k\}$, calculate the log-probability $\mathcal{L}_{ijk}$ using the formula above. This yields a vector of $\sum_{k=1}^M N_k$ log-probabilities.
    \item \textbf{Normalize:} Convert this vector of log-probabilities into a single probability vector $\tilde{\mathbf{p}}_i$ using the softmax function to ensure numerical stability:
    $$\tilde{p}_{ijk} = \frac{\exp(\mathcal{L}_{ijk})}{\sum_{k'=1}^M \sum_{j'=1}^{N_{k'}} \exp(\mathcal{L}_{ik'j'})}$$
    \item \textbf{Sample:} Draw a single sample from the unified Categorical distribution defined by the probability vector $\tilde{\mathbf{p}}_i$. The result of this draw is the new pair of assignments $(K_i, J_i)$ for the observation.
\end{enumerate}

\subsection{Clustering Structure Parameters: $V$ and $\alpha$}

\subsubsection{Stick-Breaking Variables, $V_k^\theta$ and $V_{j|k}^{\psi|\theta}$}
The full conditional posteriors are Beta distributions, resulting from Beta-Multinomial conjugacy.
\begin{gather}
    V_k^\theta \mid  \text{rest}, \mathcal{D} \sim \text{Beta}(n_k + 1, \alpha^\theta + \sum_{l=k+1}^{M} n_l) \\
    V_{j|k}^{\psi|\theta} \mid \text{rest}, \mathcal{D} \sim \text{Beta}(n_{jk} + 1, \alpha_k^{\psi|\theta} + \sum_{h=j+1}^{N_k} n_{kh})
\end{gather}

\subsubsection{Concentration Parameters, $\alpha^\theta$ and $\alpha_k^{\psi|\theta}$}
The full conditional posteriors are Gamma distributions, resulting from Gamma-Beta conjugacy.
\begin{gather}
    \alpha^\theta \mid \text{rest}, \mathcal{D} \sim \text{Gamma}\left(\eta_1^\theta + M - 1, \eta_2^\theta - \sum_{k=1}^{M-1}\log(1-V_k^\theta)\right) \\
    \alpha_k^{\psi|\theta} \mid \text{rest}, \mathcal{D} \sim \text{Gamma}\left(\eta_1^{\psi|\theta} + N_k - 1, \eta_2^{\psi|\theta} - \sum_{j=1}^{N_k-1}\log(1-V_{j|k}^{\psi|\theta})\right)
\end{gather}

\subsection{Clustering Structure Parameters: $V$ and $\alpha$}

\subsubsection{Stick-Breaking Variables, $V_k^\theta$ and $V_{j|k}^{\psi|\theta}$}
The full conditional posteriors are Beta distributions, resulting from Beta-Multinomial conjugacy.
\begin{gather}
    V_k^\theta \mid \text{rest}, \mathcal{D} \sim \text{Beta}(n_k + 1, \alpha^\theta + \sum_{l=k+1}^{M} n_l) \\
    V_{j|k}^{\psi|\theta} \mid \text{rest}, \mathcal{D} \sim \text{Beta}(n_{jk} + 1, \alpha_k^{\psi|\theta} + \sum_{h=j+1}^{N_k} n_{kh})
\end{gather}

\subsubsection{Concentration Parameters, $\alpha^\theta$ and $\alpha_k^{\psi|\theta}$}
The full conditional posteriors are Gamma distributions, resulting from Gamma-Beta conjugacy.
\begin{gather}
    \alpha^\theta \mid \text{rest}, \mathcal{D} \sim \text{Gamma}\left(\eta_1^\theta + M - 1, \eta_2^\theta - \sum_{k=1}^{M-1}\log(1-V_k^\theta)\right) \\
    \alpha_k^{\psi|\theta} \mid \text{rest}, \mathcal{D} \sim \text{Gamma}\left(\eta_1^{\psi|\theta} + N_k - 1, \eta_2^{\psi|\theta} - \sum_{j=1}^{N_k-1}\log(1-V_{j|k}^{\psi|\theta})\right)
\end{gather}

\subsection{Primary Cluster Parameters ($\theta_k$)}
For each primary cluster $k \in \{1, \dots, M\}$:

\subsubsection{Y-Model Parameters: $({\beta}_y, \sigma_y^2)$}
The full conditional posterior forms substitute the relevant data and priors: $(\mathbf{y}_k, \mathbb{M}_k, \boldsymbol{\beta}_y, \sigma_y^2, a_{\bm\beta^y}, \mathbf B_{\bm\beta^y}, a_{\sigma^y}, \mathbf B_{\sigma^y})$. 
We have the priors
\[
\bm\beta_y\mid\sigma_y^2\sim N(a_{\bm\beta^y},\;\sigma_y^2\mathbf B_{\bm\beta^y}),
\qquad
\sigma_y^2\sim\mathrm{InvGamma}(a_{\sigma^y},\mathbf B_{\sigma^y}).
\]
Here we define
\[
\mathbf y_k = (y_i)_{i\in\mathcal I_k},
\qquad
\mathbb M_k = \bigl[m_i^\top\bigr]_{i\in\mathcal I_k},
\quad
m_i^\top = \bigl[1,\;M_i,\;V_i,\;Z_i,\;C_{i1},\dots,C_{i,p_c}\bigr],
\] where
\[
\mathcal I_k = \{\,i : K_i = k\}, \quad
n_k = |\mathcal I_k|.
\]

\begin{itemize}
    \item \textbf{Full conditional posterior for $\boldsymbol{\beta}_y$ (given $\sigma_y^2$)}:
    $$\boldsymbol{\beta}_y \mid\sigma_y^2, \text{rest}, \mathcal{D} \sim N(\boldsymbol{\mu}_{\beta_y}^*, \boldsymbol{\Sigma}_{\beta_y}^*)$$
    where $\boldsymbol{\Sigma}_{\beta_y}^* = \left( (\sigma_y^2 \mathbf{B}_{\bm\beta^y})^{-1} + \frac{1}{\sigma_y^2}\mathbb{M}_k^T\mathbb{M}_k \right)^{-1}$ and $\boldsymbol{\mu}_{\beta_y}^* = \boldsymbol{\Sigma}_{\beta_y}^* \left( (\sigma_y^2 \mathbf{B}_{\bm\beta^y})^{-1}a_{\bm\beta^y} + \frac{1}{\sigma_y^2}\mathbb{M}_k^T\mathbf{y}_{k} \right)$. 

    \item \textbf{Full conditional posterior for $\sigma_y^2$ (given $\boldsymbol{\beta}_y$)}:
    $$\sigma_y^2 \mid  \boldsymbol{\beta}_y, \text{rest}, \mathcal{D} \sim \text{InvGamma}(a_{\sigma^y}^*, \mathbf B_{\sigma^y}^*)$$
    where $a_{\sigma^y}^* = a_{\sigma^y} + \frac{n_k}{2} + \frac{\text{dim}(\boldsymbol{\beta}_y)}{2}$ and $\mathbf B_{\sigma^y}^* = \mathbf B_{\sigma^y} + \frac{1}{2}(\mathbf{y}_k - \mathbb{M}_k\boldsymbol{\beta}_y)^T(\mathbf{y}_k - \mathbb{M}_k\boldsymbol{\beta}_y) + \frac{1}{2}(\boldsymbol{\beta}_y - a_{\bm\beta^y})^T \mathbf{B}_{\bm\beta^y}^{-1} (\boldsymbol{\beta}_y - a_{\bm\beta^y})$.
\end{itemize}

\subsubsection{M-Model Parameters: $({\beta}_m, \sigma_m^2)$}
The full conditional posterior forms are identical to the Y-model, substituting the relevant data and priors: $(\mathbf{m}_k, \mathbb{V}_k, \boldsymbol{\beta}_m, \sigma_m^2, a_{\bm\beta^m}, \mathbf B_{\bm\beta^m}, a_{\sigma^m}, \mathbf B_{\sigma^m})$,
where 
\[
\mathbf m_k = (M_i)_{i\in\mathcal I_k},
\qquad
\mathbb V_k=\bigl[v_i^\top\bigr]_{i\in\mathcal I_k},
\quad
v_i^\top=\bigl[1,\;V_i,\;Z_i,\;C_{i1},\dots,C_{i,p_c}\bigr],
\]
with priors $\bm\beta_m\mid\sigma_m^2\sim N(a_{\bm\beta^m},\;\sigma_m^2\mathbf B_{\bm\beta^m})$ and $\sigma_m^2\sim\mathrm{InvGamma}(a_{\sigma^m},\mathbf B_{\sigma^m})$.
\begin{itemize}
    \item \textbf{Full conditional posterior for $\boldsymbol{\beta}_m$ (given $\sigma_m^2$)}:
    $$\boldsymbol{\beta}_m \mid\sigma_m^2, \text{rest}, \mathcal{D} \sim N(\boldsymbol{\mu}_{\beta_m}^*, \boldsymbol{\Sigma}_{\beta_m}^*)$$
    where $\boldsymbol{\Sigma}_{\beta_m}^* = \left( (\sigma_m^2 \mathbf{B}_{\bm\beta^m})^{-1} + \frac{1}{\sigma_m^2}\mathbb{V}_k^T\mathbb{V}_k \right)^{-1}$ and $\boldsymbol{\mu}_{\beta_m}^* = \boldsymbol{\Sigma}_{\beta_m}^* \left( (\sigma_m^2 \mathbf{B}_{\bm\beta^m})^{-1}a_{\bm\beta^m} + \frac{1}{\sigma_m^2}\mathbb{V}_k^T\mathbf{m}_{k} \right)$. 

    \item \textbf{Full conditional posterior for $\sigma_m^2$ (given $\boldsymbol{\beta}_m$)}:
    $$\sigma_m^2 \mid  \boldsymbol{\beta}_m, \text{rest}, \mathcal{D} \sim \text{InvGamma}(a_{\sigma^m}^*, \mathbf B_{\sigma^m}^*)$$
    where $a_{\sigma^m}^* = a_{\sigma^m} + \frac{n_k}{2} + \frac{\text{dim}(\boldsymbol{\beta}_m)}{2}$ and $\mathbf B_{\sigma^m}^* = \mathbf B_{\sigma^m} + \frac{1}{2}(\mathbf{m}_k - \mathbb{V}_k\boldsymbol{\beta}_m)^T(\mathbf{m}_k - \mathbb{V}_k\boldsymbol{\beta}_m) + \frac{1}{2}(\boldsymbol{\beta}_m - a_{\bm\beta^m})^T \mathbf{B}_{\bm\beta^m}^{-1} (\boldsymbol{\beta}_m - a_{\bm\beta^m})$.
\end{itemize}

\subsection{Sub-Cluster Parameters ($\psi_{j|k}$)}
For each sub-cluster $(k,j)$:

\subsubsection{V-Model Parameters: $({\beta}_v, \sigma_v^2)$}
The full conditional posterior forms are identical to the Y-model, substituting the relevant data and priors: $(\mathbf{v}_{kj}, \mathbb{X}_{kj}, \boldsymbol{\beta}_v, \sigma_v^2, a_v, b_v, \boldsymbol{\beta}_{v0}, \mathbf{V}_{v0})$.
Based on the similar conjugate-integral arguments as for \(Y\) and \(M\), but now within the second-level cluster \((k,j)\), let
\[
\mathcal I_{kj} = \{\,i : K_i = k,\;J_i = j\}, 
\quad n_{kj} = |\mathcal I_{kj}|,
\]
and 
\[
\mathbf v_{kj} = (V_i)_{i\in\mathcal I_{kj}},
\qquad
\mathbb X_{kj} = \bigl[x_i^\top\bigr]_{i\in\mathcal I_{kj}},
\quad
x_i^\top = \bigl[1,\;Z_i,\;C_{i1},\dots,C_{i,p_c}\bigr].
\]
We have the priors
\[\bm\beta_v\mid\sigma_v^2 \sim N\bigl(a_{\bm\beta^v},\;\sigma_v^2\mathbf B_{\bm\beta^v}\bigr),
\quad
\sigma_v^2\sim\mathrm{InvGamma}(a_{\sigma^v},b_{\sigma^v}).
\]
\begin{itemize}
    \item \textbf{Full conditional posterior for $\boldsymbol{\beta}_v$ (given $\sigma_v^2$)}:
    $$\boldsymbol{\beta}_v \mid\sigma_v^2, \text{rest}, \mathcal{D} \sim N(\boldsymbol{\mu}_{\beta_v}^*, \boldsymbol{\Sigma}_{\beta_v}^*)$$
    where $\boldsymbol{\Sigma}_{\beta_v}^* = \left( (\sigma_v^2 \mathbf{B}_{\bm\beta^v})^{-1} + \frac{1}{\sigma_v^2}\mathbb{X}_{kj}^T\mathbb{X}_{kj} \right)^{-1}$ and $\boldsymbol{\mu}_{\beta_v}^* = \boldsymbol{\Sigma}_{\beta_v}^* \left( (\sigma_v^2 \mathbf{B}_{\bm\beta^v})^{-1}a_{\bm\beta^v} + \frac{1}{\sigma_v^2}\mathbb{X}_{kj}^T\mathbf{v}_{kj} \right)$. 

    \item \textbf{Full conditional posterior for $\sigma_v^2$ (given $\boldsymbol{\beta}_v$)}:
    $$\sigma_v^2 \mid  \boldsymbol{\beta}_v, \text{rest}, \mathcal{D} \sim \text{InvGamma}(a_{\sigma^v}^*, \mathbf B_{\sigma^v}^*)$$
    where $a_{\sigma^v}^* = a_{\sigma^v} + \frac{n_{kj}}{2} + \frac{\text{dim}(\boldsymbol{\beta}_v)}{2}$ and $\mathbf B_{\sigma^v}^* = b_v + \frac{1}{2}(\mathbf{v}_{kj} - \mathbb{X}_{kj}\boldsymbol{\beta}_v)^T(\mathbf{v}_{kj} - \mathbb{X}_{kj}\boldsymbol{\beta}_v) + \frac{1}{2}(\boldsymbol{\beta}_v - a_{\bm\beta^v})^T \mathbf{B}_{\bm\beta^v}^{-1} (\boldsymbol{\beta}_v - a_{\bm\beta^v})$.
\end{itemize}

\subsubsection{Z-Model Parameter: $\pi^z$}
$$\pi^z \mid  \text{rest}, \mathcal{D} \sim \text{Beta}\left(a_\pi + \sum_{i \in \mathcal{I}_{kj}}Z_i, \quad b_\pi + n_{kj} - \sum_{i \in \mathcal{I}_{kj}}Z_i\right).$$

\subsubsection{Binary C-Model Parameters: $\pi^c_q$ for each $q=1, \dots, p_{c,1}$}

$$\pi^c_q \mid  \text{rest}, \mathcal{D} \sim \text{Beta}\left(a_\pi + \sum_{i \in \mathcal{I}_{kj}}C_{i,q}, \quad b_\pi + n_{kj} - \sum_{i \in \mathcal{I}_{kj}}C_{i,q}\right)$$

\subsubsection{Continuous C-Model Parameters: $(\mu_q, \sigma_q^2)$ for each $q=p_{c,1}+1, \dots, p_{c,1}+p_{c,2}$}
With priors $\mu_q \mid \sigma_q^2 \sim N(a_\mu, \frac{1}{b_{\mu}}\sigma_q^2)$ and $\sigma_q^2 \sim \text{InvGamma}(a_{\sigma^c}, b _{\sigma^c})$,
\begin{itemize}
    \item \textbf{Full conditional posterior for $\mu_q$ (given $\sigma_q^2$)}:
    $$\mu_q \mid \sigma_q^2, \text{rest}, \mathcal{D} \sim N(\mu_q^*, \sigma_q^{2*})$$
    where $\sigma_q^{2*} = \left( \frac{n_{kj}}{\sigma_q^2} + \frac{b_\mu}{\sigma_q^2} \right)^{-1}$ and $\mu_q^* = \sigma_q^{2*} \left( \frac{n_{kj}\bar{c}_q}{\sigma_q^2} + \frac{a_{\mu}b_\mu}{\sigma_q^2} \right)$, with $\bar{c}_q = \frac{1}{n_{kj}}\sum_{i \in \mathcal{I}_{kj}} C_{i,q}$.

    \item \textbf{Full conditional posterior for $\sigma_q^2$ (given $\mu_q$)}:
    $$\sigma_q^2 \mid  \mu_q, \text{rest}, \mathcal{D} \sim \text{InvGamma}(a_{\sigma^c}^*, b_{\sigma^c}^*)$$
    where $a_{\sigma^c}^* = a_{\sigma^c} + \frac{n_{kj}}{2} + \frac{1}{2}$ and $b_{\sigma^c}^* = b_{\sigma^c} + \frac{1}{2}\sum_{i \in \mathcal{I}_{kj}}(C_{i,q}-\mu_q)^2 + \frac{ b_\mu(\mu_q-a_{\mu})^2}{2}$.
\end{itemize}

\subsection{Hyperparameters}
To ensure the prior support appropriately covers the parameter space and to facilitate rapid mixing of the Gibbs sampler, we fix the hyperparameters of the base measures using the observed data. 

\subsubsection{Regression Parameters for $Y$, $M$, and $V$}
Let $\hat{\boldsymbol{\beta}}_y, \hat{\boldsymbol{\beta}}_m, \hat{\boldsymbol{\beta}}_v$ and $\hat{\sigma}^2_y, \hat{\sigma}^2_m, \hat{\sigma}^2_v$ denote the maximum likelihood estimates (MLE) of the coefficients and the residual variances obtained from fitting standard generalized linear models (GLMs) to the entire dataset for $Y$, $M$, and $V$, respectively. 

We anchor the prior means for the regression coefficients to these MLEs:
\[
a_{\bm\beta^y} = \hat{\boldsymbol{\beta}}_y, \qquad a_{\bm\beta^m} = \hat{\boldsymbol{\beta}}_m, \qquad a_{\bm\beta^v} = \hat{\boldsymbol{\beta}}_v.
\]
To create weakly informative priors that scale appropriately with the data, we specify the prior covariance matrices as diagonal matrices proportional to the GLM residual variances.
\[
\mathbf{B}_{\bm\beta^y} = c \cdot\text{Var}\left(\hat{\bm\beta^y} \right), \quad \mathbf{B}_{\bm\beta^m} = c \cdot\text{Var}\left(\hat{\bm\beta} ^m\right), \quad \mathbf{B}_{\bm\beta^v} = c \cdot\text{Var}\left(\hat{\bm\beta}^v \right),
\]
where $c$ is the scaling factor.

For the variance parameters, we specify Inverse-Gamma priors with the shape parameters set to 1.0 and the rate parameters matching the GLM residual variances:
\[
a_{\sigma^y} = a_{\sigma^m} = a_{\sigma^v} = 1.0, \qquad \mathbf{B}_{\sigma^y} = \hat{\sigma}^2_y, \quad \mathbf{B}_{\sigma^m} = \hat{\sigma}^2_m, \quad \mathbf{B}_{\sigma^v} = \hat{\sigma}^2_v.
\]

\subsubsection{Covariate Parameters ($Z$, $C$)}
For the binary treatment assignment $Z$ and binary confounders $\mathbf{C}_{binary}$, we assume uniform priors over the probability parameters by setting:
\[
a_\pi = 1.0, \qquad b_\pi = 1.0.
\]
For the continuous confounders $C_{q}$ for $q=p_{c,1} + 1, \dots, p_{c,1}+p_{c,2}$, we set the global prior mean $a_\mu = 0$ (standradized). To provide sufficient flexibility for cluster-specific means to vary around the global mean, we set the precision multiplier $b_\mu = 0.5$. The variance components for the continuous covariates are assigned weakly informative Inverse-Gamma priors with $a_{\sigma^c} = 2.0$ and $b_{\sigma^c} = 1.0$.

\subsection{Cluster-Reallocation Algorithm}
In the posterior inference for the Enriched Dirichlet Process Mixture (EDPM) model, our algorithm employs a Blocked Gibbs sampler based on a truncated stick-breaking representation. Unlike the marginal Pólya urn scheme—which integrates out the random measure and sequentially updates individual allocation indicators—the Blocked Gibbs sampler instantiates the nested cluster weights and parameters, drawing observation allocations conditionally. However, while the Blocked Gibbs approach facilitates efficient joint updates, it can be prone to poor mixing. Specifically, the instantiated cluster parameters become strongly coupled with the data, causing the Markov chain to become trapped in local optima when exploring the complex nested partition structures of the EDPM \parencite{wade2014improving}.\\
To overcome this limitation and significantly improve the mixing of the chain within our conditional Blocked Gibbs framework, we include an additional cluster reallocation Metropolis-Hastings (MH) step after performing the standard point-wise updates at each iteration. Adapted from the marginal Pólya urn scheme of \textcite{wade2014improving}, our MH algorithm proposes a macroscopic shift by allowing an entire instantiated inner sub-cluster—along with its conditionally assigned observations and component parameters—to be moved and nested within a different outer cluster. By enabling such block-level reallocations within the truncated stick-breaking representation, the algorithm can efficiently traverse the complex parameter space and jump across low-probability valleys, which greatly accelerates convergence and improves the overall quality of the posterior samples.\\
In each iteration, this algorithm proposes to move a $\psi$-cluster to be nested with a different or new $\theta$-cluster. This step is separated into three possible moves:\begin{enumerate}
    \item[(1)] a $\psi$-cluster, among those within $\theta$-clusters with more than one $\psi$-cluster, is moved to a different $\theta$-cluster; 
    \item[(2)] a $\psi$-cluster, among those within $\theta$-clusters with more than one $\psi$-cluster, is moved to a new $\theta$-cluster;
    \item[(3)] a $\psi$-cluster, among those within $\theta$-clusters with only one $\psi$-cluster, is moved to a different $\theta$-cluster.
\end{enumerate}
Let $c_{x,2+}$ be the number of $\psi$-clusters within a $\theta$-cluster with more than one $\psi$-cluster and $c_{x,1}$ be the number of $\psi$-clusters within a $\theta$-cluster with only one $\psi$-cluster. The proposal distributions for the three moves are as follows.
For the first move, the $\psi$-cluster
is uniformly selected with probability $c^{-1}_{x,2+}$ and moved within a different $\theta$-cluster selected uniformly with probability $(c-1)^{-1}$. For the second, the $\psi$-cluster is again uniformly selected with probability $c^{-1}_{x,2+}$ and moved to a new cluster. Lastly, for the third, the $\psi$-cluster is uniformly selected with probability $c^{-1}_{x,1}$and moved within a different $\theta$-cluster selected uniformly with probability $(c-1)^{-1}$.
\\
 Even though we consider three different moves, the specifics of the Metropolis-Hasting algorithm are similar. We just need to clarify the differences of proposal probability based on the different moves. Suppose that we are moving $\psi_{j|k}$ from $k_{th}$ $\theta$-cluster to $k'_{th}$ $\theta$-cluster,  $\psi_{j'|k'}$. Then $n_{kj}$, $n_{k'j'}$, $n_k$ and $n_{k'}$ will change, which implies that the posteriors of stick-breaking variables $\mathbf{V}^\theta,\mathbf{V}^{\psi|\theta}$ and concentration parameters $\bm{\alpha}=\left(\alpha^\theta,\alpha_{j|k}^{\psi|\theta},\alpha_{j'|k'}^{\psi|\theta}\right)$ will change. After the move, we need to update all $V_s^\theta$, where $s\leq \max \{k,k'\}$. Similarly, we need to update all $V_{r|k}^{\psi|\theta}$ and $V_{r'|k'}^{\psi|\theta}$ satisfying $r\leq j$ and $r' \leq j'$. Based on the newly generated $\mathbf{V}^\theta,\mathbf{V}^{\psi|\theta}$, we regenerate $\alpha^\theta$, $\alpha^{\psi|\theta}_k$, and $\alpha^{\psi|\theta}_{k'}$. Note that the weights $\bm{\omega}$ are a deterministic function of $\mathbf{V}$. Thus all the weights for $\theta$-cluster will be updated and also for all the weights of $\psi$-clusters in $k^{th}$ and $k'^{th}$ $\theta$-clusters. 
\\
Let $\mathcal L(\mathbf y,\mathbf m,\mathbf v,\mathbf z,\mathbf c \mid \mathbf K,\mathbf J)
=\prod_{k=1}^{M}\;\prod_{j=1}^{N_k}\;\prod_{i\in\mathcal I_{kj}}
f\bigl(y_i,m_i,v_i,z_i,c_i \mid K_i=k,J_i=j\bigr)$ be the marginalized likelihood. $\pi(\mathbf K, \mathbf J\mid \bm \omega)$, $\pi(\bm{\omega}\mid\bm \alpha)$, $\pi(\bm{\alpha})$ are denoted as the priors of the cluster-assignment categorical variables $\mathbf K, \mathbf J $, the weights $\bm{\omega}$ and concentration parameters $\bm{\alpha}$.
\subsubsection{Algorithm}
0. Given the data and current posterior sample of the parameters $(\mathbf y, \mathbf m, \mathbf v, \mathbf z,\mathbf c;\mathbf K, \mathbf J, \bm{\omega},\bm{\alpha})$\\
1. Uniformly select an active inner sub-cluster $(k, j)$ with parameter $\psi_{j|k}$. Propose moving $\psi_{j|k}$, along with all data points currently assigned to it, to a target $k'$-th $\theta$ cluster. Denote its new index as $(k', j')$. Update the allocation indicators for the involved data points, and update the cluster counts accordingly:
$$n_{kj}^* = 0, \quad n_k^* = n_k - n_{kj}, \quad n_{k'j'}^* = n_{kj}, \quad n^*_{k'} = n_{k'} + n_{kj}.$$ \\
2. Sample $\left(V_1^{\theta*},\dots,V_{\max\{k,k'\}}^{\theta*}\left|\mathbf K^*, \alpha^\theta\right)\right.$, $\left(V_{1|k}^{\psi|\theta*},\dots,V_{j|k}^{\psi|\theta*}\left|\mathbf K^*, \mathbf J^*,\alpha^{\psi|\theta}_k\right)\right.$, $\left(V_{1|k'}^{\psi|\theta*},\dots,V_{j'|k'}^{\psi|\theta*}\left|\mathbf K^*,\mathbf J^*, \alpha^{\psi|\theta}_{k'}\right)\right.$, where 
\begin{align*}
    V_k^{\theta*}\left|\mathbf{K}^*\,\alpha^\theta\right.&\overset{\text{ind}}{\sim}Beta\left(n_k^*+1,\alpha^\theta+\sum_{h=k+1}^{M}n_h^*\right),\\ V_{j|k}^{\psi|\theta*}\left|\mathbf{K}^*,\mathbf{J}^*,\alpha^{\psi|\theta}_k\right.&\overset{\text{ind}}{\sim}Beta\left(n_{kj}^*+1,\alpha^{\psi|\theta}_k+\sum_{h=j+1}^{N_k}n_{kh}^*\right).
\end{align*}\\
3. Sample $\left(\alpha^{\theta*}\left|V_1^{\theta*}\right.,\dots,V_{\max\{k,k'\}}^{\theta*}, V^\theta_{\max\{k,k'\}+1},\dots, V^\theta_M \right)$, $\left(\alpha^{\psi|\theta*}_{k}\left| V_1^{\theta*},\dots,V_{\max\{k,k'\}}^{\theta*}, V^\theta_{\max\{k,k'\}+1},\dots, V^\theta_M\right)\right.$, $\left(\alpha^{\psi|\theta*}_{k'}\left|V_{1|k'}^{\psi|\theta*}\right.,\dots,V_{j'|k'}^{\psi|\theta*},V_{j'+1|k'}^{\psi|\theta},\dots,V^{\psi|\theta}_{N_k'|k'}\right)$, where
\begin{align*}
\alpha^{\theta*}\left|\mathbf{V}^{\theta*} \right.&\sim Gamma\left(M+\eta^\theta_1-1,\eta_2^\theta-\sum_{k=1}^{M-1}\log\left(1-V_k^{\theta*}\right)\right),
\\
\alpha_k^{\psi|\theta*}\left|\mathbf{V}
_k^{\psi|\theta*}\right.&\sim Gamma\left(N_k+\eta_1^{\psi|\theta}-1,\eta_2^{\psi|\theta}-\sum_{j=1}^{N_k-1}\log\left(1-V_{j|k}^{\psi|\theta*}\right)\right).
\end{align*}
\\
4. Compute the new weights $\bm{\omega}^*$ based on the stick-breaking construction: $\omega_1^{\theta*} = V_1^{\theta*}$ and $\omega_k^{\theta*}=V_k^{\theta*}\prod_{h=1}^{k-1}\left(1-V_h^{\theta*}\right)$, $k =2,\dots, M$ and $\omega_{1|k}^{\psi|\theta*} = V_{1|k}^{\psi|\theta*}$ and $\omega_{j|k}^{\psi|\theta*}=V_{j|k}^{\psi|\theta*}\prod_{h=1}^{j-1}\left(1-V_{h|k}^{\psi|\theta*}\right)$, $j=2,\dots,N_k$.\\
5. Accept the proposed move with probability\\
\begin{align*}
    \alpha=\min\left\{1,\frac{\mathcal L(\mathbf y,\mathbf m,\mathbf v,\mathbf z,\mathbf c \mid \mathbf K^*,\mathbf J^*)\cdot\pi(\mathbf K^*,\mathbf J^*|\bm \omega^*)\cdot\pi(\bm{\omega}^*\left|\bm{\alpha}^*\right.)\cdot\pi\left(\bm{\alpha}^*\right)}{\mathcal L(\mathbf y,\mathbf m,\mathbf v,\mathbf z,\mathbf c \mid \mathbf K,\mathbf J)\cdot\pi(\mathbf K,\mathbf J|\bm \omega)\cdot\pi(\bm{\omega}\left|\bm{\alpha}\right.)\cdot\pi\left(\bm{\alpha}\right)}\times\frac{q(\text{old}|\text{new})}{q(\text{new}|\text{old})}\right\},
\end{align*}
where 
\begin{align*}
\frac{q(\text{old}|\text{new})}{q(\text{new}|\text{old})}=\frac{q(\mathbf K, \mathbf J|\mathbf K ^*,\mathbf J^*)}{q(\mathbf K^*, \mathbf J^*|\mathbf K ,\mathbf J)}\times\frac{q\left(\bm{\omega}\left|\bm{\omega}^*\right.\right)}{q\left(\bm{\omega}^*\left|\bm{\omega}\right.\right)}
\end{align*}

\begin{align*}
    \frac{q(\mathbf K, \mathbf J|\mathbf K ^*,\mathbf J^*)}{q(\mathbf K^*, \mathbf J^*|\mathbf K ,\mathbf J)}=\begin{cases}
    \frac{c_{x,2+}\cdot c_{k',0}}{c*_{x,2+}\cdot c^*_{k,0}} &\text{first move}\\
    \frac{c_{x,2+}\cdot (M - c)\cdot N_{k'}}{c*_{x,1}\cdot c \cdot c^*_{k,0}}&\text{second move}\\
    \frac{c_{x,1}\cdot(c-1)\cdot c_{k',0}}{c*_{x,2+}\cdot (M - c + 1)\cdot N_k}&\text{third move}
\end{cases}
\end{align*}
and 
\begin{align*}
\frac{q\left(\bm{\omega}\left|\bm{\omega}^*\right.\right)}{q\left(\bm{\omega}^*\left|\bm{\omega}\right.\right)}=\frac{P\left(\mathbf{V}|\mathbf{K},\mathbf{J},\bm{\alpha}^*\right)\cdot P(\bm{\alpha}|\mathbf{V})}{P(\mathbf{V}^*|\mathbf{K}^*, \mathbf{J}^*,\bm{\alpha})\cdot P(\bm{\alpha}^*|\mathbf{V}^*)}\times \frac{\left|J(\mathbf{V})\right|}{\left|J(\mathbf{V}^*)\right|}.
\end{align*}
Here, $|J(\mathbf{V})|$ and $|J(\mathbf{V}^*)|$ denote the determinants of the Jacobian matrices arising from the change of variables between the weights space $\bm{\omega}$ space and the stick-breaking space $\mathbf{V}$.
\begin{align*}
    J(\mathbf{V}) = \prod_{l=1}^{M-1}\left(1-V_l^\theta\right)^{M-1-l}\times \prod_{r=1}^{N_k-1}\left(1-V_{r|k}^{\psi|\theta}\right)^{N_k-1-r}\prod_{r' = 1}^{N_{k'}-1}\left(1-V^{\psi|\theta}_{r|'k'}\right)^{N_{k'}-1-r'}\\
    J(\mathbf{V}^*) = \prod_{l=1}^{M-1}\left(1-V_l^{\theta*}\right)^{M-1-l}\times \prod_{r=1}^{N_k-1}\left(1-V_{r|k}^{\psi|\theta*}\right)^{N_k-1-r}\prod_{r' = 1}^{N_{k'}-1}\left(1-V^{\psi|\theta*}_{r|'k'}\right)^{N_{k'}-1-r'}
\end{align*}
Notice that 
\begin{align}
 \frac{\pi\left(\bm \omega^*\mid \bm\alpha^*\right)}{\pi\left(\bm\omega\mid \bm\alpha\right)}  = \frac{\pi\left(\mathbf V^*\mid \bm\alpha^*\right)\cdot |J(\mathbf V^*)|}{\pi\left(\mathbf V \mid \bm \alpha\right)\cdot |J(\mathbf V)|}   
\end{align}
where the prior of $V$ is $Beta(1, \alpha)$.
The Jacobi part in the GEM prior can be canceled out with that in proposal probability part. The acceptance rate can be computed as follows, 
\begin{align*}
     \alpha=\min&\left\{1, \frac{\mathcal L(\mathbf y,\mathbf m,\mathbf v,\mathbf z,\mathbf c \mid \mathbf K^*,\mathbf J^*)\cdot\pi(\mathbf K^*,\mathbf J^*|\bm \omega^*)\cdot\pi(\mathbf V^*\left|\bm{\alpha}^*\right.)\cdot\pi\left(\bm{\alpha}^*\right)}{\mathcal L(\mathbf y,\mathbf m,\mathbf v,\mathbf z,\mathbf c \mid \mathbf K,\mathbf J)\cdot\pi(\mathbf K,\mathbf J|\bm \omega)\cdot\pi(\mathbf V\left|\bm{\alpha}\right.)\cdot\pi\left(\bm{\alpha}\right)}\right.\\ & \qquad\times\left.\frac{q(\mathbf K, \mathbf J|\mathbf K ^*,\mathbf J^*)\cdot P\left(\mathbf{V}|\mathbf{K},\mathbf{J},\bm{\alpha}^*\right)\cdot P(\bm{\alpha}|\mathbf{V})}{q(\mathbf K^*, \mathbf J^*|\mathbf K ,\mathbf J)\cdot P(\mathbf{V}^*|\mathbf{K}^*, \mathbf{J}^*,\bm{\alpha})\cdot P(\bm{\alpha}^*|\mathbf{V}^*)}\right\}.
\end{align*}

We can then simplify the acceptance rate as follows, 
\begin{align*}
    &\left[\prod_{\substack{r=1\\r\notin\{k,k'\}}}^{M}\left(\frac{\omega_r^*}{\omega_r}\right)^{n_r}\right]\times 
    \frac{(\omega_k^*)^{n_k^*}(\omega_{k'}^*)^{n_{k'}^*}}{(\omega_{k})^{n_k}(\omega_{k'})^{n_{k'}}}\times \frac{f_Y(\mathbf y_k^*\mid\mathbf M_k^*)\ f_M(\mathbf m_k^*\mid \mathbf V_k^*)}{f_Y(\mathbf y_k\mid\mathbf M_k)\cdot f_M(\mathbf m_k\mid \mathbf V_k)}\times\frac{f_Y(\mathbf y_{k'}^*\mid\mathbf M_{k'}^*)\ f_M(\mathbf m_{k'}^*\mid \mathbf V_{k'}^*)}{f_Y(\mathbf y_{k'}\mid\mathbf M_{k'})\cdot f_M(\mathbf m_{k'}\mid \mathbf V_{k'})}\\
    \times &\left[\prod_{\substack{s=1\\s \neq j}}^{N_k}\left(\frac{\omega_{s|k}^*}{\omega_{s|k}}\right)^{n_{ks}}\cdot \prod_{\substack{s'=1\\s\neq j'}}^{N_{k'}}\left(\frac{\omega_{s'|k'}^*}{\omega_{s'|k'}}\right)^{n_{k's'}}\right]\times \frac{(\omega_{j|k}^*)^{n_{kj}^*}(\omega_{j'|k'}^*)^{n_{k'j'}^*}}{(\omega_{j|k})^{n_{kj}}(\omega_{j'|k'}^*)^{n_{k'j'}}}\\
    \times & \frac{f_V(\mathbf v_{kj}^*\mid \mathbf X_{kj}^*)\cdot f_Z(\mathbf z_{kj}^*)f_C(\mathbf c_{kj}^*)}{f_V(\mathbf v_{kj}\mid \mathbf X_{kj})\cdot f_Z(\mathbf z_{kj})f_C(\mathbf c_{kj})}\times \frac{f_V(\mathbf v_{k'j'}^*\mid \mathbf X_{k'j'}^*)\cdot f_Z(\mathbf z_{k'j'}^*)f_C(\mathbf c_{k'j'}^*)}{f_V(\mathbf v_{k'j'}\mid \mathbf X_{k'j'})\cdot f_Z(\mathbf z_{k'j'})f_C(\mathbf c_{k'j'})}\\
    \times& \frac{\pi(\bm \omega^*\mid \bm\alpha^*)\cdot \pi(\bm \alpha^*)}{\pi(\bm \omega\mid \bm \alpha)\cdot\pi( \bm\alpha)} \times \frac{q(\mathbf K, \mathbf J|\mathbf K ^*,\mathbf J^*)}{q(\mathbf K^*, \mathbf J^*|\mathbf K ,\mathbf J)}\times\frac{q\left(\bm{\omega}\left|\bm{\omega}^*\right.\right)}{q\left(\bm{\omega}^*\left|\bm{\omega}\right.\right)}
\end{align*}
Notice that the $\psi$-cluster can be empty before and after the move. When there is no assignment in one $\psi$-cluster, $f_V(\mathbf v_{kj}\mid \mathbf X_{kj})\cdot f_Z(\mathbf z_{kj})f_C(\mathbf c_{kj})$ will be 1. Also, it will be the same situation for a $\theta$-cluster if there is not any observations in it, i.e. $f_Y(\mathbf y_k\mid \mathbf M_k)\cdot f(\mathbf m_k\mid \mathbf V_k)$ will be 1.
\\If the move is not accepted, then change the value of $\left(\mathbf{V}^*,\bm{\alpha}^*, \mathbf{K}^*,\mathbf{J}^*\right)$ to the previous one,  $\left(\mathbf{V},\bm{\alpha}, \mathbf{K},\mathbf{J}\right)$.\\
After we get the posterior samples in each iteration of the Gibbs Sampling, we conduct the Metropolis-Hasting moves as follows,
\begin{enumerate}
    \item[1.] Carry out the move (1);
    \item[2.] Sample $u \sim U(0, 1)$. If $u < 0.5$ and there exists an empty $\theta$- cluster, perform move (2), otherwise perform move (3);
    \item[3.] Sample parameters $\theta^*_k,\psi_{j|k}^*$ from their posteriors given $\left(\mathbf y,\mathbf m,\mathbf v,\mathbf z,\mathbf{c},\mathbf K^*,\mathbf J^*,\bm\omega^*,\bm\alpha^*\right)$.
\end{enumerate}

\section{G-computation Algorithm}\label{app:g_comp_details}
In this section, we detail the algorithm used to estimate the mean counterfactual functional $\mathbb{E}[Y_{z, M_{z'}}]$ from the posterior distribution. Since the analytical integration over the Enriched Dirichlet Process Mixture (EDPM) model and the Gaussian copula is intractable, we approximate the target parameter using a Monte Carlo g-computation approach nested within our MCMC sampling scheme. Suppose we have run the MCMC algorithm for a total of $B$ post-burn-in iterations. Let $\Theta^{(b)} = \{\omega^{(b)}, \beta^{(b)}, \sigma^{2,(b)}, \mu^{(b)}, \rho^{(b)}\}$ denote the collection of all model parameters drawn at the $b$-th MCMC iteration, for $b = 1, \dots, B$. To compute the posterior distribution of the causal estimand, we implement the following Monte Carlo integration steps for each iteration $b$:\\
For $i = 1, \dots, T$, where $T$ is a sufficiently large number of Monte Carlo particles, sequentially simulate the counterfactual data generating process:
\paragraph{Step 1:  Baseline Covariates Generation}Sample the baseline covariates $c_i^{(b)}$ from the marginal mixture distribution $P(c\mid \Theta^{(b)})$, defined as the weighted sum of component densities:
    \begin{equation}
        P(c\mid \Theta^{(b)}) = \sum_{k=1}^{M} \sum_{j=1}^{N_k} \omega_k^{\theta,(b)} \omega_{j|k}^{\psi,(b)} \cdot f(c \mid \Theta^{(b)}).
    \end{equation}

\paragraph{Step 2: Post-treatment Confounder in World $Z=z'$}
The generation of the counterfactual post-treatment confounder depends on whether the target estimand requires a single-world or a cross-world evaluation:

\begin{itemize}
    \item \textbf{Case 1: Single-world intervention ($z = z'$).} In this case, $v_i^{(b)} = v_i'^{(b)}$.  Sample the factual post-treatment confounder $v_i^{(b)}$ from the conditional density $P(v \mid Z=z, c_i^{(b)}; \Theta^{(b)})$. This density is a mixture of normals weighted by the posterior probability of component membership given $(z, c_i^{(b)})$:
    \begin{equation}\label{vdist}
        P(v \mid z, c_i^{(b)}; \Theta^{(b)}) = \frac{\sum_{k=1}^M \sum_{j=1}^{N_k} W_{kj}^{v,(b)} \cdot \mathcal{N}\left(v \mid \mathbb{X}\beta_{kj}^{v,(b)}, \sigma_{kj}^{2,v,(b)}\right)}{\sum_{k=1}^M \sum_{j=1}^{N_k} W_{kj}^{v,(b)}},
    \end{equation}
    where the weights $W_{kj}^{v,(b)}$ capture the likelihood of the upstream variables:
    \begin{equation*}
        W_{kj}^{v,(b)} = \omega_k^{\theta,(b)}\omega_{j|k}^{\psi,(b)} \cdot P(z \mid \pi_{kj}^{z,(b)}) \cdot f(c_i^{(b)} \mid \Theta^{(b)}).
    \end{equation*}
    
    \item \textbf{Case 2: Cross-world intervention ($z \neq z'$).} To generate $v_i'^{(b)}$ under the alternative intervention $Z=z'$, we leverage the conditional distribution properties of the Gaussian copula linking the two potential confounders. This is achieved via the following probability integral transforms:
    \begin{enumerate}
    \item Sample $v_i^{(b)}\sim P(v\mid z,c_i^{(b)};\Theta^{(b)})$ from~\ref{vdist}.
        \item Compute the cumulative probability $u = F_{V_z}\left(v_i^{(b)}\mid z, c_{i}^{(b)};\Theta^{(b)}\right)$. The conditional CDF is given by the weighted average of the component-specific CDFs:
        \begin{equation}
            u = F_{V_z}(v_i^{(b)}) = \frac{\sum_{k=1}^M \sum_{j=1}^{N_k} W_{kj}^{v,(b)} \cdot \Phi\left(\frac{v_i^{(b)} - \mathbb{X}\beta_{kj}^{v,(b)}}{\sigma_{kj}^{v,(b)}}\right)}{\sum_{k=1}^M \sum_{j=1}^{N_k} W_{kj}^{v,(b)}},
        \end{equation}
        using the same weights $W_{kj}^{v,(b)}$ as in Step 2.
        
        \item Sample the latent variable $w$ from the conditional normal distribution:
        \begin{equation*}
            w\sim\mathcal N\left(\rho^{(b)}\Phi^{-1}(u), 1-\left(\rho^{(b)}\right)^2\right).
        \end{equation*}
        
        \item Compute the counterfactual value $v'^{(b)}_i$ by inverting the conditional CDF under the counterfactual treatment $z'$:
        \begin{equation*}
            v'^{(b)}_i = F_{V'}^{-1}\left(\Phi(w')\mid z', c_i^{(b)};\Theta^{(b)}\right).
        \end{equation*}
        This inversion is performed numerically on the quantile function of the mixture distribution defined by weights $W_{kj}^{v',(b)}$ (calculated using $z'$ instead of $z$).
    \end{enumerate}
\end{itemize}
\paragraph{Step 4: Mediator Generation in World $Z=z'$} Sample the mediator $m_i'^{(b)}$ from the conditional density $P(m \mid v'^{(b)}, z', c_i^{(b)}; \Theta^{(b)})$. Similar to Step 2, this is a weighted mixture:
    \begin{equation}
        P(m \mid v'^{(b)}, z', c_i^{(b)}) = \frac{\sum_{k=1}^M \sum_{j=1}^{N_k} W_{kj}^{m,(b)} \cdot \mathcal{N}\left(m \mid \mathbb{V}\beta_k^{m,(b)}, \sigma_k^{2,m,(b)}\right)}{\sum_{k=1}^M \sum_{j=1}^{N_k} W_{kj}^{m,(b)}},
    \end{equation}
    where the weights now incorporate the counterfactual post-treatment confounder:
    \begin{equation*}
        W_{kj}^{m,(b)} = \omega_k^{\theta,(b)}\omega_{j|k}^{\psi,(b)} \cdot \mathcal{N}\left(v'^{(b)}_i \mid \mathbb{X}\beta_{kj}^{v,(b)}, \sigma_{kj}^{2,v,(b)}\right) \cdot P(z' \mid \pi_{kj}^{z,(b)}) \cdot f(c_i^{(b)} \mid \Theta^{(b)}).
        \end{equation*}
\paragraph{Step 5: Expected Potential Outcome} Evaluate the conditional expectation $\mu_{y,i}^{(b)} = \mathbb{E}[Y \mid m_i'^{(b)}, v_i^{(b)}, z, \mathbf{c}_i^{(b)}]$ by computing the weighted average of the component-specific linear predictors. The expectation is given by the ratio:
    \begin{equation}
        \mu_{y,i}^{(b)} = \frac{\sum_{k=1}^M \sum_{j=1}^{N_k} W_{kj, i}^{y,(b)} \cdot \mathbb{M}\beta_k^{y,(b)}}{\sum_{k=1}^M \sum_{j=1}^{N_k} W_{kj, i}^{y,(b)}}
    \end{equation}
    where $\mathbb{M}\beta_k^{y,(b)}$ is the linear predictor for the outcome in component $k$:
    \begin{equation*}
        \mathbb{M}\beta_k^{y,(b)} = \beta_{k,0}^{y,(b)} + \beta_{k,1}^{y,(b)} m_i^{'(b)} + \beta_{k,2}^{y,(b)} v_i^{(b)} + \beta_{k,3}^{y,(b)} z + \sum_{q=1}^{p_c}\beta_{k,q+4}^{y,(b)} c_{q,i}^{(b)}
    \end{equation*}
    and $W_{kj, i}^{y,(b)}$ represents the unnormalized posterior weight for component $(k,j)$, defined as the product of the mixing weights and the component-specific densities evaluated at the generated values:
    \begin{equation*}
        W_{kj, i}^{y,(b)} = \omega_k^{\theta,(b)}\omega_{j|k}^{\psi,(b)} \cdot \mathcal{N}\left(m_i'^{(b)} \mid \mathbb{V}\beta_k^{m,(b)}, \sigma_k^{2,m,(b)}\right) \cdot \mathcal{N}\left(v_i^{(b)} \mid \mathbb{X}\beta_{kj}^{v,(b)}, \sigma_{kj}^{2,v,(b)}\right) \cdot P(z \mid \pi_{kj}^{z,(b)}) \cdot f(c_i^{(b)} \mid \Theta^{(b)}).
    \end{equation*}
\paragraph{Step 6: Monte Carlo Integration} Average the conditional expectations over the $T$ samples to obtain the Monte Carlo estimate of the target functional for the $b$-th MCMC iteration:
$$E^{(b)}\left[Y_{z,M_{z'}}\right] = \frac{1}{T} \sum_{i = 1}^{T} \mu_{y, i}^{(b)}.$$
Upon completing the algorithm for all $b = 1, \dots, B$, the collection $\{E^{(1)}, \dots, E^{(B)}\}$ constitutes a random sample from the posterior distribution of the causal estimand $\mathbb{E}[Y_{z, M_{z'}}]$. Point estimates (e.g., posterior mean) and $95\%$ credible intervals can be directly computed from the empirical quantiles of this posterior sample.

\section{Derivation of Efficient Influence Function}
\paragraph{Preliminaries: Known Non-parametric Influence Functions}
Throughout the sequential derivation of the efficient influence function via the chain rule, we repeatedly rely on three fundamental non-parametric score projections (influence functions) for building blocks evaluated at a generic observation $O = (X, Y)$. Using $\mathbb{I}(\cdot)$ to denote the indicator function, the known influence functions have the following forms:\\
For a conditional expectation $\mu(x) = \mathbb{E}[Y \mid X=x]$:
$$IF(\mu(x)) = \frac{\mathbb{I}(X=x)}{p(x)} \Big( Y - \mu(x) \Big).$$
For a conditional density $p(y \mid x)$:
$$IF(p(y \mid x)) = \frac{\mathbb{I}(X=x)}{p(x)} \Big( \mathbb{I}(Y=y) - p(y \mid x) \Big).$$
For a marginal density $p(x)$:
$$IF(p(x)) = \mathbb{I}(X=x) - p(x).$$
By iteratively applying the product rule of influence functions, $IF(A \cdot B) = IF(A) \cdot B + A \cdot IF(B)$, and substituting these known non-parametric building blocks, we can algebraically derive the influence function for the nested functionals.
\paragraph{Pathwise Differentiability and Integration Interchange (Leibniz Rule)}Formally, the influence function of a functional $\theta(P)$ is defined as its pathwise derivative evaluated at $t=0$ along a smoothly parametric submodel $P_t$, meaning $IF(\theta) = \left. \frac{\partial}{\partial t} \theta(P_t) \right|_{t=0}$. For nested continuous functionals involving integration, such as $\mu_2^z(v; c) = \int \mu_1^z(m; v, c) \, p(m \mid v, z, c) \, dm$, computing the influence function mathematically equates to differentiating an integral with respect to the submodel parameter $t$. Under the uniform boundedness and envelope conditions established in our structural sieve $S_n$, the Dominated Convergence Theorem (DCT) validates the interchange of the differentiation and integration limit operations (the Leibniz integral rule). Consequently, we can legitimately pass the pathwise derivative inside the integral and apply the standard calculus product rule to the integrand:$$IF(\mu_2^z) = \left. \frac{\partial}{\partial t} \int \mu_{1, t}^z \, p_t \, dm \right|_{t=0} = \int \left( \left. \frac{\partial \mu_{1, t}^z}{\partial t} \right|_{t=0} p_0 + \mu_{1, 0}^z \left. \frac{\partial p_t}{\partial t} \right|_{t=0} \right) dm$$$$= \int \Big[ IF(\mu_1^z) \, p + \mu_1^z \, IF(p) \Big] dm.$$This rigorous topological foundation guarantees that the product rule of influence functions can be sequentially applied within the integral operators throughout our derivation.

\begin{theorem}[Influence Function for the Single-World Effect]\label{singleword}
    The efficient influence function (EIF) for the single-world mean counterfactual $\chi = \mathbb{E}[Y_{z, M_z}]$ is given by:
    \begin{align*}
        IF(\chi) = &\frac{\mathbb{I}(Z=z)}{p(z \mid C)} \Big( Y - \mu_3^z(C) \Big)  + \mu_3^z(C) - \chi,
    \end{align*}
    where $\mu_1^z(m;v,z,c) = \mathbb E[Y\mid M=m,V=v,Z=z,C=c]$ and $\mu_2^z,\ \mu_3^z$ are the nested intermediate maps:
\begin{align*}
\mu_2^z = \int\mu_1^z(m;v,z,c) \ p(m\mid v,z,c)\ dm,\;
\mu_3^z(c)=\int \mu_2^z(v;c)\,p(v\mid z, c)\,dv,\;
\chi(P)=\mathbb{E}[\mu_3^z(C)].   
\end{align*}
\end{theorem}
\begin{proof}
\textbf{Step 1: Influence function of $\mu_2^z(v; c)$}. By the product rule of influence functions under the integral operator, we have:
$$IF(\mu_2^z(v; c)) = \int \Big[ IF(\mu_1^z(m; v, c)) \, p(m \mid v, z, c) + \mu_1^z(m; v, c) \, IF(p(m \mid v, z, c)) \Big] dm.$$
Using the standard non-parametric influence functions for conditional expectations and densities, we substitute 
$$IF(\mu_1^z) = \frac{\delta_M(m)\delta_V(v)\mathbb{I}(Z=z)\delta_C(c)}{p(m, v, z, c)} (Y - \mu_1^z)$$
and 
$$IF(p(m|v,z,c)) = \frac{\delta_V(v)\mathbb{I}(Z=z)\delta_C(c)}{p(v, z, c)} (\delta_M(m) - p(m \mid v, z, c))$$ to obtain:
\begin{align*}
    IF(\mu_2^z(v; c)) &= \int \frac{\delta_M(m)\delta_V(v)\mathbb{I}(Z=z)\delta_C(c)}{p(m, v, z, c)} \{Y - \mu_1^z(m; v, c)\} \, p(m \mid v, z, c) \, dm \\
    &+ \int \mu_1^z(m; v, c) \frac{\delta_V(v)\mathbb{I}(Z=z)\delta_C(c)}{p(v, z, c)} \Big( \delta_M(m) - p(m \mid v, z, c) \Big) dm.
\end{align*}
Notice that $$\frac{p(m \mid v, z, c)}{p(m, v, z, c)} = \frac{1}{p(v, z, c)}.$$ Integrating with respect to $\delta_M(m)$ evaluates the integrand exactly at the observed $M$. Thus, the integrals elegantly collapse to:
$$IF(\mu_2^z(v; c)) = \frac{\delta_V(v)\mathbb{I}(Z=z)\delta_C(c)}{p(v, z, c)} \Big( Y - \mu_1^z(M; v, c) \Big) + \frac{\delta_V(v)\mathbb{I}(Z=z)\delta_C(c)}{p(v, z, c)} \Big( \mu_1^z(M; v, c) - \mu_2^z(v; c) \Big).$$
\textbf{Step 2: Influence function of $\mu_3^z(c)$}. Applying the identical projection logic to $\mu_3^z(c)$:
$$IF(\mu_3^z(c)) = \int \Big[ IF(\mu_2^z(v; c)) \, p(v \mid z, c) + \mu_2^z(v; c) \, IF(p(v \mid z, c)) \Big] dv.$$
Substituting the result from Step 1 and notice that, $$\frac{p(v \mid z, c)}{p(v, z, c)}  = \frac{1}{p(z, c)}.$$ Integration over $\delta_V(v)$ collapses the terms evaluated at the observed $V$:$$IF(\mu_3^z(c)) = \frac{\mathbb{I}(Z=z)\delta_C(c)}{p(z, c)} \Big( Y - \mu_1^z(M; V, c) \Big) + \frac{\mathbb{I}(Z=z)\delta_C(c)}{p(z, c)} \Big( \mu_1^z(M; V, c) - \mu_2^z(V; c) \Big)$$$$+ \frac{\mathbb{I}(Z=z)\delta_C(c)}{p(z, c)} \Big( \mu_2^z(V; c) - \mu_3^z(c) \Big).$$
\textbf{Step 3: Influence function of the target parameter $\chi$}.\\
Finally, integrating over the marginal baseline distribution:
$$IF(\chi) = \int \Big[ IF(\mu_3^z(c)) \, p(c) + \mu_3^z(c) \, IF(p(c)) \Big] dc.$$
Substituting the result from Step 2 and the ratio $$\frac{p(c)}{p(z, c)} =\frac{1}{p(z \mid c)},$$ 
which is the inverse propensity score. The integration over $\delta_C(c)$ locks the covariates at the observed $C$, and using $IF(p(c)) = \delta_C(c) - p(c)$ yields:
    \begin{align*}
        IF(\chi) &= \frac{\mathbb{I}(Z=z)}{p(z \mid C)} \Big( Y - \mu_1^z(M; V, C) \Big) + \frac{\mathbb{I}(Z=z)}{p(z \mid C)} \Big( \mu_1^z(M; V, C) - \mu_2^z(V; C) \Big) \\&+ \frac{\mathbb{I}(Z=z)}{p(z \mid C)} \Big( \mu_2^z(V; C) - \mu_3^z(C) \Big) + \mu_3^z(C) - \chi\\
        & = \frac{\mathbb{I}(Z=z)}{p(z \mid C)} \Big( Y - \mu_3^z(C) \Big)+ \mu_3^z(C) - \chi \\
        & = \frac{\mathbb I(Z=z)}{p(z\mid C)}\Big(Y-\mu_3^z(C)\Big) + \mu_3^z - \chi
    \end{align*}
    This successfully yields the efficient influence function single-world case.
\end{proof}

\begin{theorem}[Influence Function for the Cross-World Effect]
    The efficient influence function (EIF) for the cross-world mean counterfactual $\chi = \mathbb{E}[Y_{z, M_{z'}}]$ is given by: 
    \begin{align*}
       IF(\chi) & = \frac{\mathbb I(Z=z)}{p(z\mid C)} \Big(Y-\mu_1(M;V_z,z, C)\Big) \\ & + \frac{\mathbb I(Z=z')}{1-p(z\mid C)}\int \Big[\mu_1(M;v,z,C) - \mu_2(V_{z'};v,C)\Big]c(u,w;\rho) \ p(v\mid z,C) \ dv \\
       & + \frac{\mathbb I(Z=z)}{p(z\mid C)}\Big(\mu_2(V_{z'};V_z,C)-\mu_4(C)\Big) +  \mu_4(C) - \chi
    \end{align*}
    where $\mu_1(m;v,z,c) = \mathbb E[Y\mid M = m, V=v, Z=z, C=c]$  and $\mu_2, \mu_3, \mu_4$ are the nested intermediate maps:
    \begin{align*}
        &\mu_2(v';v,c)=\!\int \mu_1(m';v,z,c)\,p(m'\mid v',z',c)\,dm',\ \ \
\mu_3(v;c)=\!\int \mu_2(v';v,c)\,p(v'\mid v,z,c)\,dv',\\
&\mu_4(c)=\!\int \mu_3(v;c)\,p(v\mid c)\,dv,\qquad
\chi(P)=\mathbb{E}[\mu_4(C)].
    \end{align*}
    Here $c(u,w;\rho)$ is the copula function:
    \begin{align*}
        c(u,w;\rho)=\frac{1}{\sqrt{1-\rho^2}}\mathrm{exp}\left[-\frac{1}{2(1-\rho^2)}\left(\Phi^{-1}(u)^2-2\rho\ \Phi^{-1}(u)\ \Phi^{-1}(w)+\Phi^{-1}(w)^2+\frac{\Phi^{-1}(u)^2+\Phi^{-1}(w)^2}{2}\right)\right]
    \end{align*}
    where $u = F_{V_z}(v\mid C),\ w = F_{V_{z'}}(v'\mid C)$ are the CDFs of $V_z,\ V_{z'}$.
\end{theorem}
\begin{proof}
\textbf{Step 1: Influence function of $\mu_2(v'; v, c)$}. 
By the product rule shown in the proof of Theorem~\ref{singleword}, we apply the standard score projections which are represented via Dirac $\delta$ functions:
\begin{align*}
IF(\mu_2) &= \int \frac{\delta_M(m')\delta_{V_z}(v)\mathbb{I}(Z=z)\delta_C(c)}{p(m', v, z, c)} \{Y - \mu_1\} \, p(m' \mid v', z', c) \, dm' \\
&\quad + \int \mu_1 \frac{\delta_{V_{z'}}(v')\mathbb{I}(Z=z')\delta_C(c)}{p(v', z', c)} \Big( \delta_M(m') - p(m' \mid v', z', c) \Big) \, dm'.
\end{align*}
Integrating out $\delta_M(m')$ collapses the expression to:
\begin{align*}
IF(\mu_2(v'; v, c)) &= \frac{\mathbb{I}(Z=z)\delta_{V_z}(v)\delta_C(c)}{p(M, v, z, c)} \{Y - \mu_1(M; v, z, c)\} \, p(M \mid v', z', c) \\
&\quad + \frac{\mathbb{I}(Z=z')\delta_{V_{z'}}(v')\delta_C(c)}{p(v', z', c)} \Big( \mu_1(M; v, z, c) - \mu_2(v'; v, c) \Big).
\end{align*}

\textbf{Step 2: Influence function of $\mu_3(v; c)$}. 
Applying the chain rule again: $IF(\mu_3) = \int [IF(\mu_2) p(v' \mid v, z, c) + \mu_2 IF(p(v' \mid v, z, c))] dv'$.
Substituting $IF(\mu_2)$ and integrating over $v'$ (where $\delta_{V_{z'}}(v')$ evaluates the term at the observed $V_{z'}$):
\begin{align*}
IF(\mu_3(v; c)) &= \frac{\mathbb{I}(Z=z)\delta_{V_z}(v)\delta_C(c)}{p(M, v, z, c)} \{Y - \mu_1(M; v, z, c)\} \int p(M \mid v', z', c) p(v' \mid v, z, c) dv' \\
&\quad + \frac{\mathbb{I}(Z=z')\delta_C(c)}{p(V_{z'}, z', c)} \Big( \mu_1(M; v, z, c) - \mu_2(V_{z'}; v, c) \Big)\ p(V_{z'} \mid v, z, c) \\
&\quad + \frac{\mathbb{I}(Z=z)\delta_{V_z}(v)\delta_C(c)}{p(v, z, c)} \Big( \mu_2(V_{z'}; v, c) - \mu_3(v; c) \Big).
\end{align*}

\textbf{Step 3: Expanding to $\chi$}. 
Iterating this projection logic through $\mu_4$ and $\Psi$, and finally dividing the joint density denominators by the marginal $p(c)$ (which constructs the propensity scores), yields the raw unsimplified influence function:
\begin{align*}
IF(\chi) &= \frac{\mathbb{I}(Z=z)}{p(z \mid C) p(M \mid V_z, z, C)} \{Y - \mu_1(M; V_z, z, C)\} \int p(M \mid v', z', C) p(v' \mid V_z, z, C) dv' \\
&\quad + \frac{\mathbb{I}(Z=z')}{p(z' \mid C)} \int \Big( \mu_1(M; v, z, C) - \mu_2(V_{z'}; v, C) \Big) \frac{p(V_{z'} \mid v, z, C)}{p(V_{z'} \mid z', C)} \, p(v \mid z, C) dv \\
&\quad + \frac{\mathbb{I}(Z=z)}{p(z \mid C)} \Big( \mu_2(V_{z'}; V_z, C) - \mu_3(V_z; C) \Big) \\
&\quad + \frac{\mathbb{I}(Z=z)}{p(z \mid C)} \Big( \mu_3(V_z; C) - \mu_4(C) \Big) + \mu_4(C) - \chi.
\end{align*}

\textbf{Step 4: Simplification via Cross-World Identity and Copula formulation}.
\begin{itemize}
    \item By the conditional cross-world independence assumption, the convolution density $\int p(M \mid v', z', C) p(v' \mid V_z, z, C) dv'$ mathematically collapses to the observed conditional density $p(M \mid V_z, z, C)$. This perfect cancellation in the first term yields the standard outcome regression score:
    $$ \phi_1 = \frac{\mathbb{I}(Z=z)}{p(z\mid C)} \{Y - \mu_1(M; V_z, C)\}. $$
    
    \item In the second term, the density ratio $\frac{p(V_{z'} \mid v, z, C)}{p(V_{z'} \mid z', C)}$ is mathematically equivalent to the conditional density $p(V_{z'}\mid v,z, C)$ divided by the marginal density $p(V_{z'} \mid z', C)$. By Sklar's theorem, this ratio exactly formulates the Gaussian copula density $c(u, w; \rho)$. Substituting this yields the copula-augmented score:
    $$ \phi_2 = \frac{\mathbb{I}(Z=z')}{p(z'\mid C)} \int [\mu_1(M; v, C) - \mu_2(V_{z'}; v, C)] \, c(u, w; \rho) \, p(v \mid z, C) dv. $$
\end{itemize}
Thus we obtain:
\begin{align*}
    IF(\chi)  & = \frac{\mathbb I(Z=z)}{p(z\mid C)} \Big(Y-\mu_1(M;V_z,z, C)\Big) \\ & + \frac{\mathbb I(Z=z')}{1-p(z\mid C)}\int \Big[\mu_1(M;v,z,C) - \mu_2(V_{z'};v,C)\Big]c(u,w;\rho) \ p(v\mid z,C) \ dv \\
       & + \frac{\mathbb I(Z=z)}{p(z\mid C)}\Big(\mu_2(V_{z'};V_z,C)-\mu_3(V_z;C)\Big)+ \frac{\mathbb I(Z=z)}{p(z\mid C)}\Big(\mu_3(V_z;C)-\mu_4(C)\Big) +  \mu_4(C) - \chi. 
\end{align*}
    
\end{proof}

\section{Proof of Semiparametric Bernstein-von Mises (BvM) Theorem of Cross-World Case}\label{AsecBvM}

In this proof section, we employ the following notation for the purpose of clarity and convenience.\\
Let $P_0$ denote the true underlying data-generating distribution, and let $P$ denote any generic distribution within our model space. Correspondingly, we define the collection of nuisance parameters evaluated under $P$ as $\eta= \{\eta_l\}_{l=1}^8 = \{e,\ g_{V_z},\ g_{z'|z},\ g_M,\ \mu_1, \ F_{V_z},\ F_{V_{z'}}, \  \pi\}$, where 
\begin{align*}
    e &= P(Z = z\mid C = c),\ g_{V_z} = p(v\mid z, C), \  g_M = p(m\mid v,z,c),\ g_{z'|z} = p(v_{z'}\mid v_z,z,c)\\ \mu_1 &= \mathbb E[Y\mid m,v,z,c],\ 
     F_{V_z}  = F_{V_z}(v\mid z,c),\  F_{V_{z'}} = F(v'\mid z',c)\text{ and }\pi = P(S = 1\mid M,Z,C).
\end{align*}
Denote their true counterparts evaluated under $P_0$ with a subscript zero, yielding $\eta_0 = \{e_0, g_{V_z,0}, g_{z'|z,0}, g_{M,0}, \mu_{1,0}, F_{V_z,0}, F_{V_{z'},0}, \pi_0\}$. 

\subsection{Augmented Sieve Set $\tilde H_n$}
In this section, we construct the key theoretical object for our proof: the augmented sieve set, denoted $\tilde H_n$. This set is designed to be a "well-behaved" subset of the model space on which our subsequent analysis will rely. We first formally define its components in Section~\ref{sec:defH}. Then, in Section~\ref{proofH}, we state the crucial proposition that our EDP posterior distribution indeed concentrates on this sieve, a property that is fundamental to our main theorem.

\subsubsection{Definition of Augmented Sieve Set $\tilde H_n$}\label{sec:defH}

In this section, we follow a modular approach to construct augmented sieve set $\tilde H_n$, which is the intersection of two distinct sets, $\tilde H_n = H_n\cap S_n$, each designed to control a different aspect of the posterior's behavior.\par

The rate-sieve, $H_n$, specifies the required analytical properties of the nuisance functions—such as their $L_2$ posterior contraction rates and specific regularity conditions like the Lipschitz continuity of our Gaussian copula bridge that precisely satisfy the core assumptions of the one-step BvM theorem from \parencite{yiu2025semiparametric}. The general theoretical framework for proving that a posterior can indeed concentrate on such a rate-defined set was established in the seminal work of \parencite{ghosal2000convergence}.\par
The structural sieve, $S_n$, in contrast, is designed to control the complexity of our Enriched Dirichlet Process (EDP) model. By placing explicit constraints on the number of active mixture components, the range of the parameters, and the non-degeneracy of the variances, we ensure that the model space is sufficiently regular for theoretical analysis. This approach of using an explicit structural sieve for a complex Dirichlet Process mixture model follows the methodology of \textcite{shen2013adaptive}.\par
This two-part construction allows us to systematically prove that the posterior concentrates on a set of functions that are both structurally plausible under our model and analytically regular enough for our main theorem's assumptions to be satisfied. We will define these sets formally in the following subsections.

\begin{definition}[Rate-sieve $H_n$]\label{defHnrate}
Let $\eta= \{\eta_l\}_{l=1}^8 = \{e,\ g_{V_z},\ g_{z'|z},\ g_M,\ \mu_1, \ F_{V_z},\ F_{V_{z'}}, \  \pi\}$ and let $\|\cdot\|$ denote the $L_2(P_0)$ norm. 
For deterministic rates $\varepsilon_n = n^{-1/3}(\log n)^{1/3} \downarrow0$ and constants $L_{e}, L_{V_z}, L_{z'|z}, L_M, L_{F_{V_z}}, L_{F_{V_{z'}}},L_\rho,L_\pi, \delta, \delta_\pi$, define
\[
\begin{aligned}
H_n=\Bigl\{P:\ &
\|e-e_0\|\le L_e \varepsilon_n,\quad \|g_{V_Z}-g_{V_z,0}\|\le L_{V_z}\varepsilon_n, \quad
\|g_{z'|z}-g_{z'|z,0}\|\le L_{z'|z}\varepsilon_n,\\
&\|g_M-g_{M,0}\|\le L_M\varepsilon_n,\quad \|\mu_1-\mu_{1,0}\|\le L_{(0)}\varepsilon_n, \quad \|\pi - \pi_0\|\leq L_\pi\varepsilon_n, \\&
\|F_{V_{z'}} - F_{V_{z'},0}\| \leq L_{{F_{V_{z'}}}}\varepsilon_n, \quad \|F_{V_z} - F_{V_z,0}\| \leq L_{F_{V_z}}\varepsilon_n, \quad|\rho-\rho_0|\leq L_{\rho}\varepsilon_n,\\&
\delta<e(C)<1-\delta\text{, a.s. for some }\delta\in(0,1/2), \ \pi(W) > \delta_\pi\text{ a.s. for some }\delta_\pi>0\\[2pt]
&\text{\bf Bridge regularity: } g_{z'|z}=\mathcal T_{\rm bridge}(F_{V_{z'}},F_{V_z},\rho),\ \rho\in(-1+\kappa,1-\kappa),\\
&\qquad \exists L<\infty:\ \|g_{z'|z}-g_{z'|z,0}\|\le L_{z'|z}\!\left(\|F_{V_{z'}}-F_{V_{z'},0}\|+\|F_{V_{z}}-F_{V_{z},0}\|+|\rho-\rho_0|\right),\\[2pt]
&\text{\bf Quantile stability: } U_{z'}=F_{V_{z'}}(V_{z'}),\ U_1=F_{V_z}(V_z)\in[\tau_n,1-\tau_n]\\&\text{with prob. }\ge 1-\zeta_n,\ \tau_n\downarrow0,\ \zeta_n\to0
\Bigr\}.
\end{aligned}
\]
\end{definition}

\begin{definition}[Structural EDP sieve $S_n$]\label{defSn}
To verify the theoretical assumptions, we construct a structural sieve, denoted $S_n$, which is a subset of the full infinite-dimensional model space. The design of this sieve follows the methodology established by \parencite{shen2013adaptive} for controlling the complexity of Dirichlet process mixture models.\\
To define the sieve, we first specify a set of non-stochastic parameters. Let $H_\theta$ and $H_\psi$ be the effective numbers of active clusters for the outer and inner layers, respectively. Let $\epsilon_\theta$ and $\epsilon_\psi$ be small constants representing tail mass thresholds. Let $a, b > 0$ be constants that bound the norms of the model's regression coefficients and location parameters. Let $\delta, \delta' \in (0,1/2)$ be constants for the overlap conditions on the treatment and binary covariates, respectively. Let $\sigma_0^2$ be a baseline scale that serves as a universal lower bound to ensure the non-degeneracy of all kernel variances. Finally, let $M_\theta,\ M_\psi$ and $M_\psi'$ be integers that control the complexity of the covariance matrix parameter spaces, let $d^y,\ d^m, \ d^v ,\ d^c_1$ and $d_2^c$ denote the dimensions of the respective variables, and let $\text{eig}_h$ denote the $h$-th eigenvalue of a square matrix.
\\
While these parameters are fixed for the definition of a single set $S_n$, in the proof of posterior concentration they are chosen as specific functions of the sample size n to achieve an optimal balance between model complexity and approximation accuracy, following the framework of \parencite{shen2013adaptive}
With these parameters specified, we formally define the structural EDP sieve  $S_n$ as follows:
\begin{align*}
S_n = \Biggl\{P &= \sum_{k=1}^{\infty}\omega_kP_k, \text{ with } P_k = \sum_{j=1}^{\infty}\omega_{j|k}P_{kj}:\\
&\textbf{Outer layer: }\\ &\ \text{Tail mass: }\sum_{k>H_\theta} \omega_k< \epsilon_\theta,\\
&\ \text{For each }k \leq H_\theta, \ \text{parameters for }Y, \ M \text{ satisfy}:\\ &\  \|\beta^y_k\|_2 \leq a, \ \|\beta_k^m\|_2<a, \\ & \ \sigma_0^2 \leq \text{eig}_h(\Sigma^y_k)\leq\sigma^2_0(1+\epsilon_\theta^2/d^y)^{M_\theta},   
\ \sigma_0^2 \leq \text{eig}_h(\Sigma^m_k)\leq\sigma^2_0(1+\epsilon_\theta^2/d^m)^{M_\theta}\\
&\textbf{Inner layer: } \\ &\ \text{Tail mass: }\sum_{j>H_\psi}\omega_{j|k} < \epsilon_\psi,\\ &\ \text{For each }j \leq H_\psi, \text{ parameters for }V,\ Z ,\ \text{binary }C_1\text{ and  continuous } C_2 \text{ satisfy: }\\ &\ \|\beta^v_{j|k}\|_2 \leq b,   \ \delta<\pi^z_{j|k}<1-\delta, \\ &\ \delta'<\pi^{c_1}_{j|k,q_1}<1-\delta', \ \|\mu_{j|k,q_2}^{c_2}\|_2<b\quad  \text{for }q_1 = 1,\dots d^c_1,\quad q_2 = 1,\dots, d^c_2\\ 
&  \ \sigma_0^2 \leq \text{eig}_h(\Sigma^v_{j|k})\leq\sigma^2_0(1+\epsilon^2_\psi/d^v)^{M_\psi}  ,\ \sigma_0^2 \leq \sigma^{c_2,2}_{j|k,q_2}\leq M_\psi'\quad \text{for }q_2=1,\dots, d_2^c \Biggr\}
\end{align*}
\end{definition}

\begin{definition}[Augmented Sieve Set $\tilde H_n$]\label{deftildeH}
    The sieve on which we will verify the assumptions for the semiparametric BvM theorem is the intersection of the sets defined above:
    \[\tilde H_n :=H_n\cap S_n\]
\end{definition}

\subsection{Posterior Concentration on the Sieve}\label{proofH}

Having defined the augmented sieve $\tilde H_n$, we now state the key result that it satisfies the primary requirement of our theoretical framework: the posterior distribution concentrates on this set with probability tending to one.

\begin{theorem}[Posterior Concentration on Augmented Sieve $\tilde H_n$]
Suppose that for a rate $\epsilon_n\rightarrow 0 $ such that $n\epsilon_n^2\rightarrow 0$, the following conditions hold for some constants $C, \ R > 0$:\\
\textbf{(Entropy Control)}: The bracketing entropy of the sieve satisfies:
\begin{align}\label{condition1}
\log N_{[]}(\epsilon_n,\tilde H_n, d) \leq n\epsilon_n^2
\end{align}
\textbf{(Prior Support)}: The prior mass concentration in a Kullback-Leibler neighborhood within the sieve is sufficient:\\
\begin{align}\label{condition2}
\Pi \left(\{P\in \tilde H_n\mid K(p_{P_0},p_{P})\leq n\epsilon_n^2 \text{ and } V_2(p_{P_0}, p_P)\leq n \epsilon_n^2\}\right)\geq e^{-Cn\epsilon^2_n},
\end{align}
where $K(\cdot, \cdot)$ is Kullback–Leibler (KL) divergence for $p_P$ and $p_{P_0}$, and $V_2(\cdot, \cdot)$ is the respective KL variance.\\
Then the posterior distribution concentrates on the augmented sieve, i.e.,  $\Pi(P\in \tilde H_n\mid Z_{1:n})\rightarrow 1$ in $P_0$-probability. 
\end{theorem}

\begin{proof}
\textbf{Step 1. Verify the condition (\ref{condition1})}\\
To prove that the posterior distribution concentrates on the augmented sieve set $\tilde H_n$,  we apply the theoretical framework established by \parencite{ghosal2000convergence} and specialized for independent, non-identically distributed (i.n.i.d.) observations in \parencite{ghosal2007convergence}. The first key condition to verify is that the complexity of the set $\tilde H_n$, as measured by its metric entropy, is sufficiently controlled.\\
To verify (\ref{condition1}), we must show that for the chosen convergence rate $\epsilon_n = n^{-1/3}(\log n)^{1/3}$, there exists a constant $R>0$, such that the log-covering number of the set $\tilde H_n$ satisfies the following inequality:
\begin{align*}
    \log N_{[]}(\epsilon_n,\tilde H_n, d)\leq Rn\epsilon_n^2
\end{align*}
where $d$ is a suitable metric, such as the Hellinger distance. \\
By the definition of $\tilde H_n$, we have $\tilde H_n\subseteq S_n$. A fundamental property of metric entropy is that the entropy of a subset is less than or equal to the entropy of the larger set. Thus we have:
\begin{align*}
    \log N_{[]}(\epsilon_n,\tilde H_n,d) \leq \log N_{[]}(\epsilon_n, S_n, d), 
\end{align*}
where we already derived that the log bracketing entropy of $S_n$, $\log N_{[]}(\epsilon_n, S_n, d)$is controlled by the sum of the log bracketing entropy of underlying density functions classes, which is approximately $C\cdot H_{\theta}H_{\psi}\cdot \log (1/\epsilon)$ with some constant $C$. \\
A key aspect of this theoretical framework is that the sieve parameters are chosen as specific functions of the sample size $n$ to optimize the balance between model complexity and approximation accuracy . Following the strategy in \textcite{shen2013adaptive}, the number of components $H_\theta H_{\psi}$ is chosen to scale with the desired convergence rate $\epsilon_n$ . A standard choice is:\begin{align}
    H_\theta H_\psi \propto \frac{n \epsilon_n^2}{\log n}
\end{align}
We now substitute the chosen rate $\epsilon_n = n^{-1/3}(\log n)^{1/3}$ and the scaling of $H_\theta H_\psi$ into the entropy bound. First, we note that for large $n$:
\begin{align*}
    \log(1/\epsilon_n) = \log(n^{1/3}(\log n)^{-1/3}) = \frac 1 3\log n - \frac 1 3 \log (\log n)\propto \log n
\end{align*}
Substituting these into the entropy bound gives:
\begin{align*}
    \log N_{[]}(\epsilon_n, S_n,d)\approx C\cdot \left(\frac{n\epsilon_n^2}{\log n}\right)\cdot(\text{const}\cdot \log n)
\end{align*}
The $\log n$ terms cancel, and we are left with:
\begin{align*}
    \log N_{[]}(\epsilon_n, S_n,d)\approx R\cdot n\epsilon_n^2
\end{align*}
for some constant $R$. This upper bound is of the precise form required by \textcite{ghosal2000convergence} and \textcite{ghosal2007convergence} framework.\\
\textbf{Step 2. Verify the condition (\ref{condition2})}\\
To verify the prior support condition, we must show that the prior $\Pi$ assigns sufficient mass to a Kullback-Leibler (KL) neighborhood of the true distribution $P_0$. Our strategy is to first construct an adequate approximation $P_n^*$ of the true data-generating distribution $P_0$, such that $P_n^*\in \tilde H_n$ , and then demonstrate that the prior places at least $e^{-Cn\epsilon_n^2}$ mass in a small ball around $P_n^*$. The true distribution $P_0$ is a joint density over a mixed domain of continuous and discrete variables. We leverage the theoretical guarantees for Dirichlet Process mixtures of Generalized Linear Models (DP-GLM). As established by \textcite{hannah2011dirichlet}, DP-GLM models possess a powerful density approximation property, enabling them to approximate any sufficiently smooth joint density with compact support, even when the variables are of mixed types. This stems from the fact that mixtures of standard kernels (e.g., Gaussian for continuous variables, and Multinomial/Bernoulli for discrete variables) are dense in the space of all such distributions.
\\
Following this principle, we construct $P_n^*$ as a finite EDP mixture with $H=H_\theta H_\psi\propto n \epsilon_n^2/\log n$ active clusters The construction is designed to ensure $P_n^*\in \tilde H_n = H_n \cap S_n$
\\
Our construction of $P_n^*$ is achieved by constructing an approximation for each true nuisance parameter, $$\{e^*_n,g_{V_z,n}^*,g_{z'|z,n}^*,g_{M,n}^*, \mu_{1,n}^*, F_{V_{z'}}^{*}, F_{V_z}^{*},\pi_{n}^* \} = \{\eta_{j,n}^*\}_{j=1}^8$$ which collectively defines the distribution $P_n^*$. This constructive step is crucial for the entire proof of the prior support condition.\\
\textbf{Satisfying $S_n$ Constraints:} We construct $P_n^*$ as a truncated EDP mixture model. We explicitly choose the number of outer and inner components, $H_\theta, H_\psi$, to (be less than or) equal to the scaling parameters that define the sieve $S_n$, so the tail mass conditions are trivially satisfied.  The parameters for each component—including means and covariance matrices for continuous parts, and probabilities (e.g., $\pi^z_{j|k}, \pi^c_{j|k,q}$) for the discrete parts—are chosen to lie within the compact bounds specified by our modified Definition~\ref{defSn} of $S_n$. By its very construction, $P_n^* \in S_n$
\\
\textbf{Satisfying $H_n$ Constraints:} Simultaneously, we leverage approximation theory to ensure $P_n^* \in H_n$. For each nuisance parameter $\eta_{j,0}\in \{\eta_{j,0}\}_{j=1}^8 $ = $\{e_{0}, g_{V_z,0},g_{z'|z,0}, g_{M,0}, \mu_{1,0},F_{V_{z'},0}, F_{V_z,0}, \pi_0\}$ with smoothness $\beta_j$, the theory guarantees that by choosing a sufficiently large number of components $H$, we can achieve an approximation error of $\|\eta_{j,n}^* - \eta_{j,0}\|_{L_2(P_0)}\leq C_j H^{-\beta_j/d_j}$ for some constant $C_j$. Here $d_j$ represents the dimension of the input domain for the nuisance function $\eta_{j,0}$. For example, $d_1$ for the propensity score $\eta_{1,0} = e_0(c) = P_0(Z=1\mid C = c)$ is the dimension of the covariates $C$. Our choice of $H\propto n \epsilon_n^2/\log n$ ensures that $H
$ grows with $n$, making the approximation error decrease. For the rate $\epsilon_n = n^{-1/3}(\log n)^{1/3}$, this choice of $H$ is sufficient to guarantee that for large enough $n$, the error is bounded by our target rate, i.e.,
\begin{align*}
    \|\eta_{j,n}^* - \eta_{j,0}\|_{L_2(P_0)}\leq \epsilon_n.
\end{align*} 
Thus, by construction, $P_{n}^*$ also satisfies the rate conditions defined in Definition~\ref{defHnrate} of $H_n$.
\\
\textbf{Bounding the KL Divergence and Variance:} Now, we connect the $L_2$ approximation accuracy to the KL condition. The KL divergence between $P_0$ and our constructed approximation $P_{n}^*$ can be bounded by the sum of squared $L_2$ errors of the nuisance parameters as is shown in Lemma~\ref{KLbounds}. Furthermore, the KL variance $V_2(p_{P_0}, p_{P_n^*})$ is subject to a similar bound for the well-behaved likelihoods in this model. Thus, following arguments in \parencite{ghosal2007convergence}:
\begin{align*}
    \max\left\{ K(p_{P_0}, p_{P_n^*}), V_2(p_{P_0}, p_{P_{n}^*})\right\}\leq C_K \sum_{j=1}^{8}\| \eta_{j,n}^* - \eta_{j,0}\|_{L_2(P_0)}^2
\end{align*}
for some constant $C_K$. We have already established that $\|\eta_{j,n}^* - \eta_{j,0}\|_{L_2(P_0)}\leq \epsilon_n$. Substituting this result into the inequality above yields:
\begin{align*}
    \max\left\{ K(p_{P_0}, p_{P_n^*}), V_2(p_{P_0}, p_{P_{n}^*})\right\}\leq R' \sum_{j=1}^{8} \epsilon_n^2 = 8R'\epsilon_n^2.
\end{align*}
It is clear that $8K' \epsilon_n^2< n\epsilon_n^2$ as $n\rightarrow\infty$. This confirms that our constructed $P_n^*$ lies within the required KL-ball around the truth $P_0$ and thus satisfies the premises of condition (\ref{condition2}).
\\

For the “ideal” approximating distribution $P_n^*$ constructed as above (a mixture with $H$ components), how large is the prior mass that our EDP prior
places on a small $\delta$–neighborhood of $P_n^*$? This is a classical calculation from Dirichlet process theory.
\\
For an EDP mixture generated by the stick–breaking representation, its probability is determined by both the outer layer and the inner layer. And with in each layer, the probability is determined by two parts: weights and the parameters in each cluster. 
\paragraph{Outer layer.} We need the probability that the stick–breaking weights
$\bigl(w_1,\ldots,w_{H_\theta}\bigr)$ fall into a small neighborhood of the outer layer weights
$\bigl(w_1^*,\ldots,w_{H_\theta}^*\bigr)$.
By properties of the Dirichlet law, the log–probability is on the order
\[
  \log \Pi\bigl(\text{weights near } (w_1^*,\ldots,w_{H_\theta}^*)\bigr)
  \gtrsim -\, C_1\cdot H_\theta\cdot \log H_\theta ,
\]
where the $H_\theta\cdot\log H_\theta$ term mainly comes from combinatorial factors and Stirling approximations to Gamma functions.\\
At the same time, each “atom” in the outer layer $\theta_k$ is a \emph{vector of parameters}
that specifies the $k$-th outer cluster’s Gaussian regressions for $Y$ and $M$, where $\theta_k=\{\beta_k^{y},\,\Sigma_k^{y},\,\beta_k^{m},\,\Sigma_k^{m}\}.$ We need the probability that, when drawing $H_\theta$ such parameter vectors from the
base measure, all of them fall inside a small $\delta_\theta$-neighborhood of the
target vectors $\theta_k^*$ that define $P_n^*$.
The total log-probability admits the lower bound
\[
  \log \Pi(\theta_1,\cdots, \theta_{H_\theta}) \;\ge\; -\, C_2\cdot H_\theta \log(1/\delta_\theta).
\]
\paragraph{Inner layer.} Within each outer layer, we do the same trick again. We need the probability that the stick–breaking weights
$\bigl(w_1,\ldots,w_{H_\psi}\bigr)$ fall into a small neighborhood of the inner layer weights
$\bigl(w_1^*,\ldots,w_{H_\psi}^*\bigr)$.
By properties of the Dirichlet law, the total log–probability is on the order
\[
  \log \Pi\bigl(\text{weights near } (w_1^*,\ldots,w_{H_\psi}^*)\bigr)
  \gtrsim -\, C_3\cdot H_\theta \cdot H_\psi\cdot \log H_\psi ,
\]
\\
Within each inner layer, $\psi_{j|k}$ is a \emph{vector of parameters}
of variables $V,Z, \text{binary }C^1$ and $\text{continuous }C^2$ in the $j$-th inner cluster and $k$th outer cluster, where $$\psi_{j|k}=\left\{\beta_{j|k}^{v}, \Sigma_{j|k}^{v}, \pi_{j|k}^z, \pi_{j|k, q_1}^c,\mu^c_{j|k,q_2},\sigma^c_{j|k,q_2},  \text{ for }q_1 = 1, \dots, d^c_1,\ q_2 = 1,\dots, d^c_2\right\}.$$ We need the probability that, when drawing $H_\theta$ and $H_\psi$ such parameter vectors from the
base measure, all of them fall inside a small $\delta_\psi$-neighborhood of the
target vectors $\psi_{j|k}^*$ that define $P_n^*$.
The total log-probability admits the lower bound by the result in \textcite{hannah2011dirichlet} and Lemma~\ref{discrete}, 
\[
  \log \Pi(\psi_{1|k},\cdots, \psi_{H_\psi|k}) \;\ge\; -\, C_4\cdot H_\theta\cdot H_\psi\log(1/\delta_\psi), \quad \text{for }k = 1,\dots,H_{\theta}.
\]

Therefore, we obtain the classical bound
\begin{align*}
    \Pi\!\left(\mathcal B_n(P_n^*,\delta_n)\right)
  \;\ge\; \exp\bigl\{-C_1 H_\theta\log (H_\theta) - C_2 H_\theta \log(1/\delta_\theta) - C_3H_\theta H_\psi\log (H_\psi)-C_4H_\theta H_\psi\log(1/\delta_\psi) \bigr\},
\end{align*}
where $\mathcal B_n(P_n^*, \delta_n) = \{P\in \mathcal P\mid K(p_{P_n^*}, p_P)\leq \delta_n^2\text{ and }V_2(p_{P_n^*,p_P})\leq \delta_n^2\}$.
\par
A rigorous derivation of this formula can be found, for example, in \textcite{shen2013adaptive} Theorem 4.\par
It is important to note that our chosen base measures for all parameters, including the Beta priors for the Bernoulli probabilities of the discrete variables   $Z$ and $C_1$, have full support on their respective parameter spaces. This ensures a positive prior probability for any small neighborhood around the parameters of $P_n^*$, satisfying the fundamental requirement for this bound to hold.\par
We now verify that this lower bound is greater than or equal to the required $e^{-Cn\epsilon_n^2}$ with the standard choice for the sieve parameters. The dominant term in the exponent of the prior probability lower bound involves the total number of components $H=H_\theta H_\psi$.  Let us consider the term $-C_3H_\theta H_\psi\log (H_\psi)$.\\
We adopt the optimal scaling for the number of components, $H = H_\theta H_\psi \propto n\epsilon_n^2/\log n $, for instance by setting $H_\theta \propto \sqrt{n\epsilon_n^2}/\log n$ and $H_\psi\propto \sqrt{n\epsilon_n^2}$. With this choice, we have $\log (H_\psi)\approx \log(\sqrt{n\epsilon_n^2}) = 1/2\log(n^{1/3}(\log n)^{2/3})\propto \log n $.\\
Substituting these into the dominant term of the exponent yields:
\begin{align*}
    -C_3H_\theta H_\psi \log(H_\psi) \approx -C'\left(\frac{n \epsilon_n^2}{\log n }\right)(\text{const}\cdot \log n) = -C'' n \epsilon_n^2 
\end{align*}
The other terms in the exponent can be shown to be of a smaller or similar order. Therefore, the total prior probability is bounded below by $\exp(-C''n\epsilon _n^2)$ for some sufficiently large constant $C''$. This completes the proof of the prior support condition (\ref{condition2}).
\end{proof}

\subsection{Imputation Function Derivation}
To compute $\dot\chi_{P_{O'}}$, we must derive the explicit form of the imputation function $\mathbb E_{P_{O'}}[\dot \chi_{P_f}\mid W, S=1]$.  This function is defined as the conditional expectation of the full-data EIF $\dot{\chi}_{P_f}$ given $W=(M,Z,C)$.
We derive this by taking the conditional expectation of each of the five components $\phi_{j, P_f}, \ j=1, \dots, 5$. Denote $b_{j}(W) = \mathbb E_{P_{O'}}[\dot\chi_{P_f}\mid W,S=1] = \mathbb \sum_{j=1}^5E_{P_{O'}}[\phi_{j, P_f}\mid W,S=1 ]$. Then,
\begin{enumerate}
     \item[1.] $b_1(W) = \mathbb E_{P_{O'}}[\phi_{1,P_f}\mid W,S=1]$. Recall, 
 $$\phi_{1,P_f} = \frac{Z}{e(C)}\{Y-\mu_1(M;V_z,C)\}.$$ This term is non-zero only if $Z=1$. 
 \begin{align*}
     \mathbb E[\phi_{1,P_f}\mid M,Z=1, C,S=1] = \frac{1}{e(C)}\{\mathbb E_{P_{O'}}[Y\mid M,Z=1, C, S=1] - \mathbb E_{P_{O'}}[\mu_1\mid M,Z=1,C,S=1]\}
 \end{align*}
By the Law of Total Expectation and the definition of $\mu_1$, 
\begin{align*}
    \mathbb E_{P_{O'}}[Y\mid M,Z=1,C,S=1] =  \mathbb E_{P_{O'}}[\mu_1(M;V_1,C)\mid M,Z=1, C, S=1]
\end{align*}
As such, $b_1(W) = \mathbb E_{P_{O'}}[\phi_{1,P_f}\mid M,Z=1,C,S=1] = 0$.
\item[2.] $b_2(W) = \mathbb E_{P_{O'}}[\phi_{2,P_f}\mid W,S=1]$. Recall, \\
\begin{align*}
    \phi_{2,P_f}=\frac{1-Z}{1-e(C)}\!\int\![\mu_1(M;v,C)-\mu_2(V_{z'};v,C)]\,c(u,w;\rho)\,g_{V_z}(v\!\mid C)\,dv
\end{align*}
This term would be non-zero only if $Z=0$. Since $\phi_{2,P_f}$ depends on $V_0$, we must take its expectation conditional on $W_0=(M,Z=0, C)$.
\begin{align*}
    \mathbb  E_{P_{O'}}[\phi_{2,P_f}| W_0, S=1] = \frac{1-Z}{1-e(C)}\mathbb E_{V_0|M,Z=0, C}\left[\int \left[\mu_1(M;v,C)-\mu_2(V_0;v,C)\right]c(u,w;\rho)g_{V_1}(v| C)dv\right]
\end{align*}
To compute $b_2(W)$, we integral $\phi_{2,P_f}$  over the posterior distribution of $V_0\mid M,Z=0,C$.
\item[3.] $b_3(W) = \mathbb E_{P_{O'}}[\phi_{3,P_f}\mid W,S=1]$. Recall, 
\begin{align*}
    \phi_{3,P_f}=\frac{Z}{e(C)}\{\mu_2(V_{z'};V_z,C)-\mu_3(V_z;C)\}.
\end{align*}
This term is non-zero only if $Z=1$. By the definition of $\mu_3$ and the Law of Total Expectation, 
\begin{align*}
    \mathbb E_{P_{O'}}[\mu_2\mid M,Z=1, C, S=1] &= \mathbb E[\mu_2\mid Z=1, C, S=1] = \mathbb E_{V_1}[\mathbb  E_{V_0}[\mu_2\mid V_1, Z=1, C]\mid Z=1, C,S=1] \\&= \mathbb E_{V_1}[\mu_3\mid Z=1, C, S=1],
\end{align*}
which implies that $b_3(W) = \mathbb E_{P_{O'}}[\phi_{3,P_f}] = 0$.
\item[4.] $b_4(W) = \mathbb E_{P_{O'}}[\phi_{4,P_f}\mid W, S=1]$. Recall, 
\begin{align*}
\phi_{4,P}=\frac{Z}{e(C)}\{\mu_3(V_z;C)-\mu_4(C)\}.
\end{align*}
This term is non-zero only if $Z=1$. We take the expectation over the posterior distribution of $V_1\mid M, Z=1,C$: 
\begin{align*}
    \mathbb E_{P_{O'}}[\mu_3\mid Z=1, C, S=1] = \mathbb E_{V_1\mid M,Z=1,C}[\mu_3(V_1;C)].
\end{align*}
\item[5.] $b_5(W) = \mathbb E_{P_{O'}}[\phi_{5,P_f}\mid W, S=1]$. Recall, 
$$\phi_{5,P_f} = \mu_4-\chi(P_f).$$ This term depends only on $C\in W$ and is a constant. Thus, $b_5(W) = \mathbb E_{P_{O'}}[\phi_{5,P_f}\mid W, S=1] = \mu_4(C) - \chi(P_f)$.
\end{enumerate}

Finally, we have $b(W) = b_2(W) + b_4(W) + b_5(W)$. We write the new EIF as follows,
\begin{align*}
\dot \chi_{P_{O'}} = \frac{S}{\pi(W)}\dot\chi_{P_f} -\frac{S-\pi(W)}{\pi(W)}b(W).
\end{align*}

\subsection{Semiparametric BvM Assumptions and Main Theorem}
\subsubsection{Assumptions and Main Theorem}
\begin{assumption}\label{assumption_abc}
    There exists a sequence of measurable subsets $\left(\tilde H_n\right)_n$ of $\mathcal P$ satisfying $\Pi\left(P_{O'}\in \tilde H_n\mid \mathcal O'_{1:n}\right) \rightarrow 1$ with:\\
    \indent (a)(No second-order bias)
    \begin{align*}
        \sup _{P_{O'}\in \tilde H_n} \left|\sqrt{n}r_2(P_{O',0}, P_{O'})\right| = \sup_{P_{O'}\in \tilde H_n}\left|\sqrt n \{\chi(P_{O',0}) - \chi(P_{O'}) - P_0[\dot{\chi_{P_{O'}}}]\}\right| \rightarrow 0
    \end{align*}
    \indent (b)($L_2$-convergence) 
    \begin{align*}
        \sup_{P\in \tilde H_n}\|\dot\chi_{P_{O'}} - \dot \chi_{P_{O', 0}}\|_{P_{0}} \rightarrow0
    \end{align*}
    \indent (c*)(i)(Convergence of $\dot \chi_{P_{O'}}$ under the empirical process)\begin{align*}
        \sup_{P\in \tilde H_n}\left|\mathbb G_n[\dot \chi_{P_{O'}} - \dot \chi _{P_{O',0}}]\right| \rightarrow 0\quad\text{in }P_0\text{-probability}
    \end{align*}
    \indent (ii)(Bounding of envelope functions) The sets $\{\dot \chi_{P_{O'}}:P_{O'}\in \tilde H_n\}$ have envelope functions $G_n$ (i.e. $|\dot \chi_{P_{O'}} (o')| \leq G_n(z)$ for all $P \in \tilde H_n$ and all $o'\in \mathcal O'$) satisfying \begin{align*}
        \lim_{C\rightarrow\infty} \limsup_{n \rightarrow\infty} P_0G_n^2\mathbf{1}_{G_n^2>C} = 0, \quad P_0G_n^4 = o(n).
    \end{align*}
\end{assumption}

\begin{theorem}
    Under Assumption~\ref{assumption_abc}, the one-step posterior satisfies the semiparametric BvM theorem, which is 
    \begin{align*}
        d_{BL}\left(\mathcal L _{\Pi\times \Pi _{BB}}(\sqrt n(\tilde \chi - \hat \chi_n)\mid O_{1:n}),\mathcal N(0, ||\dot \chi_{P_0}||_{P_0}^2)\right) \rightarrow 0 \quad\text{in }P_0\text{-probability},
    \end{align*}
    where $\mathcal L_{\Pi\times\Pi_{BB}}\left(\sqrt n (\tilde \chi - \hat \chi_n)\mid O_{1:n}\right)$ denotes the posterior law of $\sqrt n (\tilde \chi - \hat \chi_n)$ and $\hat \chi _n$ is the asymptotically efficient sequence $\hat \chi_n = \chi(P_{O', 0})+\mathbb P_n[\dot \chi_{P_{O', 0}}]$. As discussed in Section 2.1 of  \textcite{yiu2025semiparametric}, this justifies the use of central
credible sets as confidence regions.
\end{theorem}

\subsubsection{Verification of Assumptions}
\begin{proposition}[Verification of Assumption~\ref{assumption_abc}(a)]Given Lemma~\ref{lem:r2} and Definition~\ref{defHnrate}, the second-order is negligible. Mathematically, we can prove 
\begin{align*}
\sup _{P\in \tilde H_n} \left|\sqrt{n}r_2(P_{O'0}, P_{O'})\right| \rightarrow 0
\end{align*}
\end{proposition}

\begin{proof}
\textbf{Step 1 (Apply the Quadratic Bound from the Lemma)}. First, we use Lemma~\ref{r2:MAR}, which provides an upper bound for the absolute value of the remainder term, $|r_2(P_{O',0}, P_{O'})|$. This lemma is crucial because it shows the remainder is made up of products of nuisance function errors.
\begin{align*}
   |r_{2}(P_{O',0}, P_{O'})| \le K'' \Big(
\sum_{j,k=1}^8 |\Delta \eta_j| \, |\Delta \eta_k| + |\Delta \pi| \sum_{j=1}^8 |\Delta \eta_j| \Big),
\end{align*}

where $\Delta t = t - t_0$ and $K''$ is a constant.

\textbf{Step 2 (Use the Properties of the Sieve $\tilde H_n$)}. Next, we use the fact that our supremum is over $P_{O'} \in \tilde H_n = H_n\cap S_n$. Thus $P_{O'}$ must satisfy the conditions of he rate-sieve $H_n$. \\
Definition~\ref{defHnrate} gives us the explicit bounds on $L_2(P_0)$ norms of the nuisance errors for any $P\in H_n$. Then we can substitute these bounds into the right hand side of the inequality from \textbf{Step 1}. Each product term is bounded,  \[\|\Delta\eta_j\|\ \|\Delta \eta_k\|\leq L_{j}L_{k}\varepsilon_n^2\]
Applying his to all terms in Lemma~\ref{lem:r2}'s bound gives:
\begin{align*}
    |r_2(P_{O'0}, P_{O'})|&\leq K\left(\sum_{j=1}^8L_jL_k + L_\pi\sum_{h=1}^8 L_h\right)\varepsilon_n^2 \\ &= K'\varepsilon^2_n = K' n^{-2/3}(\log n)^{2/3}
\end{align*}
Thus 
\begin{align*}
    \sup_{P_{O'}\in \tilde H_n}\left|\sqrt n r_2(P_{O'0}, P_{O'})\right| \leq K' n^{-1/6}(\log n)^{2/3} \rightarrow 0  
\end{align*}
\end{proof}
\begin{proposition}[Uniform $L_2$-Convergence of the EIF on the Sieve]
Let $\tilde H_n$ be the sieve defined in Definition~\ref{deftildeH}. Then, the efficient influence function (EIF)  $\dot \chi_{P_{O'}}$ converges uniformly to $\dot \chi_{P_{O', 0}}$ in the $L_2(P_0)$norm over $\tilde H_n$ as $n\rightarrow\infty$. That is:  
\begin{align*}
\sup_{P\in \tilde H_n}  \|\dot \chi_{P_{O'}} - \dot \chi_{P_{O',0}}\| \rightarrow0
\end{align*} 
\end{proposition}
\begin{proof}
It is trivial to verify Assumption~\ref{assumption_abc}(b) given Lemma~\ref{lem:l2_mar}. For any $P_{O'}\in \tilde H_n$, we have $\|\dot \chi_{P_{O'}} - \dot \chi_{P_{O',0}}\|_{P_0}\rightarrow0$ given Lemma~\ref{lem:l2_mar} and Definition~\ref{defHnrate}.
\end{proof}

\begin{proposition}[Verification of Assumption~\ref{assumption_abc}(c*)(ii)]
The class of functions $\{\dot \chi_{P_{O'}}\mid P\in \tilde H_n\}$ admits an envelope function $G_n(Z)$. This envelope function satisfies the required moment and uniform integrability conditions:
\begin{align*}
        \lim_{C\rightarrow\infty} \limsup_{n \rightarrow\infty} P_0G_n^2\mathbf{1}_{G_n^2>C} = 0, \quad P_0G_n^4 = o(n).
\end{align*}
\end{proposition}
\begin{proof}
   We have $\dot \chi_{P_{O'}} = \frac{S}{\pi(W)}\sum_{j = 1}^{4}\phi_{j,P} + \phi_{5,P_f} - \frac{S-\pi(W)}{\pi(W)}b(W)$.
\begin{align*}
\left|\dot \chi_{P_{O'}}\right| &\leq \left|\frac{S}{\pi(W)}\right|\cdot \sum_{j=1}^4\left|\phi_{j,P}\right|  + |\phi_{5,P_f}|+\left|\frac{S-\pi(W)}{\pi(W)}\right|\cdot |b(W)|\\
&\leq \frac{1}{\delta_\pi}\cdot \sum_{j=1}^4|\phi_{j,P_f}| + |\phi_{5,P_f}| + \left(1+\frac{1}{\delta_\pi}\right)\cdot |b(W)|.
\end{align*}
A direct consequence of the structural sieve $S_n$ is the uniform boundedness of the nuisance functions. The definition of $S_n$ follows the sieve construction of \textcite{shen2013adaptive}, which ensures this property by imposing constraints on the model parameters, particularly the covariance matrix.\\
Formally, as argued in Section S8.2 of the supplementary material of \textcite{yiu2025semiparametric} for this class of sieves, there exists a set of uniform constants $\{C_{e},C_{g_{V_1}}, C_{g_{0|1}}, C_{g_M}, C_{\mu_1}, C_\pi,C_{\text{copula}}, C_b\}$ such that for all $P\in \tilde H_n$ and almost all $x = (y,m,v,v_0,z,c)$, it holds that 
\begin{align*}
    &\left|e(c)\right| \leq C_e, \ \left|g_{V_1}(v|c)\right| \leq C_{g_{V_1}},\ \left|g_{0|1}(v_0|v,c)\right|\leq C_{g_{0|1}}, \ \left|g_M(m|v_0,c)\right|\leq C_{g_M},\ \left|\mu_1(m;v,c)\right|\leq C_{\mu_1}\\&|\pi(m,z,c)|\leq C_\pi, \ |c(u,w;\rho)|\leq C_{\text{copula}},\ |b(m,z,c)|\leq C_b.
\end{align*}
Based on the uniform boundedness of nuisance parameters shown above,   $\mu_2,\ \mu_3, \ \mu_4$ are also uniformly bounded. Noting $\left|\mathbb E[X]\right| \leq \mathbb E[\left|X\right|]$, we have
\begin{align*}
    &\left|\mu_2(v_0;v,c)\right| =\left|\mathbb E\left[\mu_1(M;v,c)\mid V_0 = v_0, Z=0, C=c\right]\right| \leq \mathbb E\left[\left|\mu_1(M;v,c)\right|\mid V_0 = v_0, Z= 0, C=c\right] = C_{\mu_1}\\
    &\left|\mu_3(v;c) \right|=\left| \mathbb E\left[\mu_2(V_0; v,c)\mid V_1 = v, Z=1,C=c\right]\right| \leq \mathbb E\left[\left|\mu_2(V_0; v,c)\right|\mid V_1 = v, Z=1,C=c\right] \leq C_{\mu_1} \\
    & \left|\mu_4(c)\right| \leq \left|\mathbb E\left[\mu_3(V;c)\mid Z = 1, C= c\right]\right| \leq \mathbb E\left[\left|\mu_3(V;c)\right|\mid Z = 1, C= c\right] \leq C_{\mu_1}\\
    &\left|\chi(P_{f})\right|  = \left|\mathbb E[\mu_4(C)]\right| \leq \mathbb E\left[|\mu_4(C)|\right] \leq C_{\mu_1}
\end{align*}
Based on $\delta < e(c) <1-\delta$ for $\delta \in \left(0,\frac{1}{2}\right)$ and $\pi(m,z,c) > \delta_\pi$ for $\delta_\pi > 0$
\begin{align*}
    |\phi_{1,P_f}| &= \left|\frac{Z}{e(C)}\{Y - \mu_1(M;V_1,C)\}\right| \leq \frac{|Z|}{|e(C)|}\cdot \left(|Y| + |\mu_1(M;V_1,C)|\right)  \leq \frac{1}{\delta}\left(|Y| + C_{\mu_1}\right)
   \\ |\phi_{2,P_f}| &= \left|\frac{1-Z}{1-e(C)}\int [\mu_1(M;v,C) - \mu_2(V_0;v,C)]c(u,w;\rho)g_{V_1}(v\mid C)\ dv\right| \\ &\leq \frac{1}{\delta} \int [\left|\mu_1(M;v,C)\right| +\left| \mu_2(V_0;v,C)\right|]\cdot |c(u,w;\rho)|\cdot g_{V_1}(v\mid C)\ dv \\ &\leq \frac{1}{\delta} \int 2C_{\mu_1}\cdot C_{\text{copula}}\cdot g_{V_1}(v\mid C)\ dv= \frac 2 \delta C_{\mu_1}C_{\text{copula}} = C_2'\\
    |\phi_{3,P_f}| &= \left|\frac{Z}{e(C)}\{\mu_2(V_0;V_1,C) - \mu_3(V_1;C)\}\right|\leq \frac{1}{\delta}\left(|\mu_2(V_0;V_1,C)| + |\mu_3(V_1;C)|\right) \leq \frac{2}{\delta}C_{\mu_1} = C'_3\\
    |\phi_{4,P_f}| &= \left|\frac{Z}{e(C)}\{\mu_3(V_1;C) - \mu_4(C)\}\right|\leq \frac{1}{\delta}\left(|\mu_3(V_1;C)|+ |\mu_4(C)|\right)\leq \frac{2}{\delta} C_{\mu_1} = C'_4\\
    |\phi_{5,P_f}|& = \left|\mu_4(C) - \chi(P)\right|\leq |\mu_4(C)| + |\chi(P)| \leq 2 C_{\mu_1} = C'_5 \\
    |b_2| &= \left|\mathbb E [\phi_{2,P_f}\mid M,Z,C,S=1]\right| \leq \mathbb E[|\phi_{2,P_f}|M,Z,C, S=1]\leq C_2' \\ 
    |b_4| & = \left|\mathbb E[\phi_{4,P_f}\mid M,Z,C,S=1]\right|\leq \mathbb E [|\phi_{4,P_f}|M,Z,C,S=1] \leq C_4'\\
    |b_5| &= |\phi_{5,P_f}| \leq C_{5}'.
\end{align*}
Thus,
\begin{align*}
    \left|\dot \chi _P\right|& \leq \frac{1}{\delta_\pi}\sum_{j=1}^5|\phi_{j,P}|  + \left(1+\frac{1}{\delta_\pi}\right)|b|\\&\leq \frac{1}{\delta_\pi} \left[\frac{1}{\delta}(|Y|+|C_{\mu_1}|) + C_2'+C_3'+C_4'+C_5' \right]+ \left(1+\frac{1}{\delta_\pi}\right)(2\ C_2' + C_5')\\ &= K'_1\cdot |Y| + K'_2.
\end{align*}
Then we can set the envelope function $G_n(X) = K_1'|Y| + K'_2$, where $K_1',K_2'$ are the constants do not depend on $n$ and $P\in \tilde H_n$.
Notice that we have the moment condition on $Y$ where $\mathbb E_{P_0}[|Y|^4]<\infty$.\\
Now we can first verify $P_0G_n^4 = o(n)$.
$$P_0G_n^4 = \mathbb{E}_{P_0}[G_n^4] = \mathbb{E}_{P_0}[(K'_1|Y| + K'_2)^4]$$
Apply the inequality $(a+b)^p \le 2^{p-1}(a^p+b^p)$ for $a,b \ge 0, p \ge 1$. Setting $p=4$, we have $$\mathbb{E}_{P_0}[(K'_1|Y| + K'_2)^4] \le \mathbb{E}_{P_0}[2^3( (K'_1|Y|)^4 + K'_2{}^4 )] = 8K'_1{}^4 \mathbb{E}_{P_0}[|Y|^4] + 8K'_2{}^4<\infty, $$
which indicates $P_0G_n^4 = o(n)$.\\
Secondly, we verify $\lim_{C\rightarrow\infty}\limsup_{n\rightarrow\infty}P_0[G_n^2\mathbf 1_{G_n^2>C}]$ = 0. Since $G_n = K'_1|Y| + K_2'$ does not depend on $n$, so the limit can be simplified as $\lim_{C\rightarrow\infty }P_0[G_n^2\mathbf 1 _{G_n^2>C}]$. Along with $P_0G_n^4<\infty$, we have
\begin{align*}
    P_0[G_n^2\mathbf 1_{G_n^2>C}] &= \mathbb E_{P_0}[G_n^2\mathbf 1_{G_n^2>C}] \leq \sqrt{\mathbb E_{P_0}[G_n^4]\mathbb E_{P_0}[\mathbf 1_{G_n^2 > C}^2] } \qquad \qquad\qquad  \text{(Cauchy-Schwarz)}
    \\&=\sqrt{\mathbb E_{P_0}[G_n^4]\mathbb E_{P_0}[\mathbf 1_{G_n^2>C}] } =\sqrt{\mathbb E_{P_0}[G_n^4]P_0(G_n^2 > C)}  \rightarrow 0
\end{align*}
Then Assumption~\ref{assumption_abc}
(c*)(ii) has been verified.
\end{proof}

\begin{proposition}[Verification of Assumption~\ref{assumption_abc}(c*)(i)]
Let $\mathcal F_n = \{\dot \chi_{P_{O'}}-\dot \chi_{P_{O', 0}}\mid P\in \tilde H_n\}$. Under the conditions defining the augmented sieve $\tilde H_n$ and the modeling assumptions, the empirical process indexed by the class $\mathcal F_n$ converges uniformly to zero in probability. That is:
\begin{align*}
    \sup_{P\in \tilde H_n}\left|\mathbb G_n[\dot \chi_{P_{O'}} - \dot\chi_{P_{O',0}}]\right| \rightarrow0
\end{align*}
\end{proposition}
\begin{proof}
Define $\mathcal F_n = \{\dot \chi_{P_{O'}} - \dot\chi_{P_{O',0}}\mid P_{O'}\in \tilde H_n\}$. Based on Lemma~\ref{Snentropycontrol}, we already derived the upper bounds of the log bracketing entropy for underlying density function classes, $\mathcal F_{C,n}, \mathcal F_{Z|C,n}, \mathcal F_{V|Z,C,n}, \mathcal F_{M|V,Z,C,n}, \mathcal F_{Y|M,V,Z,C,n}, \mathcal F_{S\mid W,n}$.\\ 
Notice that the nuisance parameters \(\eta=\{e,g_{V_1},g_{0|1},g_M,\mu_1, \pi, b\}\) are not modeled directly but can be derived as functionals of the density functions that constitute the full model.\\The transformations that map the base density functions to these nuisance parameters—involving operations such as integration and division—are smooth. The overlap condition in the rate sieve $H_n$ ensures the stability of division, while the boundedness from the structural sieve $S_n$ ensures the stability of the integrals.\\
Crucially, in empirical process theory, a smooth mapping (specifically, a Lipschitz continuous mapping) does not significantly increase the complexity of a function class. Therefore, the complexity (i.e., the order of the log-entropy) of the nuisance parameter function classes is controlled by and is of the same order as the complexity of the underlying density function classes, which we established in the preceding analysis as being dominated by polynomials of $\log n$.\\
With the properties of the nuisance parameter classes established, we consider the mapping from these parameters to the Efficient Influence Function (EIF) itself. Lemma~\ref{lem:L2} provides the key result, demonstrating that the map from the vector of nuisance parameters $\eta$ to the EIF $\dot \chi_P$ is Lipschitz continuous in the $L_2(P_0)$ norm. This Lipschitz continuity is a crucial property. Specifically, if a map is Lipschitz continuous, the bracketing entropy of the resulting function class can be bounded by the sum of the bracketing entropies of the input function classes. \\
Therefore, we can conclude that the log bracketing entropy of our target class, $\mathcal F_n = \{\dot \chi_P-\dot\chi_{P_0}\mid P\in \tilde H_n\}$, is controlled by the sum of the log bracketing entropies of the individual nuisance parameter classes and thus is controlled by the sum of the log bracketing entropies of the underlying density function classes $\{\mathcal F_{C,n}, \mathcal F_{Z|C,n}, \mathcal F_{V|Z,C,n}, \mathcal F_{M|V,Z,C,n}, \mathcal F_{Y|M,V,Z,C,n}, \mathcal F_{S\mid W,n}\}$. \\
Therefore, we can directly conclude:
\begin{align*}
    \log N_{[]}(\epsilon, \mathcal F_n, \rho ) \leq K'(&\log N_{[]}(\epsilon,\mathcal F_{C,n}, \rho) + \log N_{[]}(\epsilon,\mathcal F_{Z|C,n}, \rho) + \log N_{[]}(\epsilon,\mathcal F_{V|Z,C,n}, \rho)  \\ +&\log N_{[]}(\epsilon,\mathcal F_{M|V,Z,C,n}, \rho) + \log N_{[]}(\epsilon,\mathcal F_{Y|M,V,Z,C,n}, \rho) + \log N_{[]}(\epsilon,\mathcal F_{S\mid W,n}, \rho)).
\end{align*}

Given that we have already established that the entropy of each underlying density class is dominated by a polynomial of $\log n$ , it follows directly that the entropy of the final EIF class $\mathcal F_n$ is also dominated by a polynomial of $\log n$.

A key result from empirical process theory, Theorem 2.5.2 \parencite{van1996weak}, bounds the expected supremum of the empirical process. A simplified form of such a bound is:
\begin{align*}
    \mathbb E_{P_0}^*\left[\sup_{f\in \mathcal F_n}\left|\mathbb G_n[f]\right |\right] \leq J\left(\int_0^{\sigma_n}\sqrt{\log N_{[]} (\epsilon,\mathcal{F}_n, L_2(P_0))}d\epsilon, \sigma_n,\|G_n\|_{P_0 ,2} \right),
\end{align*}
where $\mathbb E^*$ is the outer expectation and $J(\cdot )$ is a monotonically increasing function with $J(0) = 0$. We verify that he arguments of this function converge to zero:\\
\textbf{Entropy Integral}: Define the variance upper bound as $\sigma^2_n := \sup_{f\in \mathcal F_n} P_0[f^2]$. We have established that the log bracketing entropy, $\log N_{[]}$, is bounded by an upper bound that is dominated by polynomials of $\log n$. Because this growth is slow, the Dudley integral $\int _0^{\sigma_n}\sqrt{\log N_{[]}} d\epsilon$  is a well-behaved term that converges to zero as its upper limit of integration, $\sigma_n$, converges to zero. That is, 
\begin{align}\label{limsqrtentropy=0}
    \lim_{n\rightarrow\infty}\int ^{\sigma_n}_0 \sqrt{\log N_{[]} (\epsilon,\mathcal{F}_n, L_2(P_0))}d\epsilon = 0.
\end{align}
Then,
\begin{align}\label{ineqentropy}
    \int ^{\sigma_n}_0 \sqrt{\log N_{[]} (\epsilon,\mathcal{F}_n, L_2(P_0))}d\epsilon \leq \sigma_n\cdot \sup_{\epsilon \in (0,\sigma_n]}\sqrt{\log N_{[]} (\epsilon,\mathcal{F}_n, L_2(P_0))}.
\end{align}
We already proved that the log bracketing entropy $\log N_{[]}(\epsilon, \mathcal F_n,L_2(P_0))$ is dominated by a polynomial of $\log n$, thus the supremum term on the right hand side of (\ref{ineqentropy}) grows at a sub-polynomial rate in $n$. To complete the proof of (\ref{limsqrtentropy=0}), it suffices to show that $\sigma_n\rightarrow0$ as $n\rightarrow\infty$.
\\
\textbf{Variance}: Now we are going to prove that the variance of all functions in the class  $\mathcal F_n = \{\dot \chi_{P_{O'}} - \dot\chi_{P_{O',0}}\mid P_{O'}\in \tilde H_n\}$ converges uniformly to 0. We have $$\sigma_n = \sup_{f \in \mathcal{F}_n} \|f\|_{P_{0}} = \sup_{P \in \tilde{H}_n} \|\dot \chi_{P_{O'}} - \dot\chi_{P_{O',0}}\|_{P_0}$$
By applying the result of Lemma~\ref{lem:l2_mar}, we can bound this supremum by the sum of the $L_2$ errors of the individual nuisance parameters. Crucially, by the Bridge regularity condition in Definition~\ref{defHnrate}, the error of $g_{0|1}$ is absorbed into the errors of the marginal CDFs and the copula correlation. Thus, we have:
\begin{align*}
    \sup_{P_{O'} \in \tilde{H}_n} \|\dot{\chi}_{P_{O'}} - \dot{\chi}_{P_{O'},0}\|_{P_0} \le \sup_{P \in \tilde{H}_n} K \Big( &\|e-e_0\|_{P_0} + \|g_{V_z} - g_{V_z,0}\|_{P_0} + \|g_M - g_{M,0}\|_{P_0} \\
    +& \|\mu_1 - \mu_{1,0}\|_{P_0} + \|\pi-\pi_0\|_{P_0} + \|F_{V_{z'}} - F_{V_{z'},0}\|_{P_0} \\
    +& \|F_{V_z} - F_{V_z,0}\|_{P_0} + |\rho - \rho_0| \Big).
\end{align*}
Furthermore, by the Definition~\ref{defHnrate} of the rate sieve $H_n$, we know that for every $P \in \tilde{H}_n$, each nuisance parameter's $L_2$ error is bounded by a term proportional to $\varepsilon_n$: $\|\eta_j - \eta_{j,0}\|_{P_0} \le L_j \varepsilon_n$. Substituting these bounds yields:$$
\sup_{P_{O'} \in \tilde{H}_n} \|\dot{\chi}_{P_{O'}} - \dot{\chi}_{P_{O'},0}\|_{P_0} \le K \cdot \left(L_e + L_{V_1} + L_{M} + L_{(0)} + L_{\pi} + L_{F0} + L_{F1} + L_\rho\right) \varepsilon_n = K' \varepsilon_n
$$
for a constant $K'$ that does not depend on $n$.\\
Therefore, we obtain the bound on the variance: $\sigma^2_n\leq \left(K'\varepsilon_n\right)^2$. Since the rate sieve $H_n$ requires $\varepsilon_n = n^{-1/3}(\log n)^{1/3} $, it follows that $\sigma_n^2 \rightarrow0$ as $n \rightarrow\infty$.\\
\textbf{Envelope}: The validity of Assumption (c*)(ii) ensures the existence of a suitable envelope function $G_n$ with finite moments.\\
To conclude, the entire upper bound from the maximal inequality converges to zero. Therefore, we have:
\begin{align*}
    \lim_{n\rightarrow\infty}\mathbb E^*_{P_0}\left[\sup_{f\in \mathcal F_n}\left|\mathbb G_n[f]\right|\right] = 0.
\end{align*}
By Markov's inequality, convergence in expectation implies convergence in $P_0$-probability. This completes the proof that $\sup_{f\in \mathcal F}\left|\mathbb G_n[f]\right|\rightarrow0$ in $P_0$-probability, thus verifying Assumption~\ref{assumption_abc}(c*)(i).
\end{proof}

\subsection{Supplementary Lemmas}
\begin{lemma}[Second–order remainder for the full data case]\label{lem:r2}
We discuss the full data case first. Let $P$ be any distribution in a (dominated) neighborhood of $P_0$.
Write the nuisance collection
\[
\eta=\{e,g_{V_z},g_{z'|z},g_M,F_{V_{z'}}, F_{V_z},\mu_1\}.
\]
Assume \emph{overlap} $\delta<e_0(C)<1-\delta,\ $ a.s.\ for some $\delta\in(0,\tfrac12)$ and \emph{moment/envelope} bounds ensuring that
$\mu_1,g_M,g_{z'|z},g_{V_z},F_{V_{z'}},F_{V_z}$ have square–integrable envelopes under $P_0$.
Denote $\Delta t:=t-t_0$ and the $L_2(P_0)$ norm by $\|\cdot\|$.
Then the second–order remainder
\[
r_2(P_0,P):=\chi(P)-\chi(P_0)-(P-P_0)\!\big[\dot\chi_{P_0}\big]
\]
contains \emph{no linear terms} in the nuisance errors and admits the quadratic bound
\begin{align*}
    |r_2(P_0,P)| \leq K \Big( &\|\Delta\mu_1\|\|\Delta g_M\| + \|\Delta\mu_1\|\|\Delta g_{z'|z}\| + \|\Delta\mu_1\|\|\Delta g_{V_z}\| \\
    &+ \|\Delta g_M\|\|\Delta g_{z'|z}\| + \|\Delta e\|\|\Delta\mu_1\| + \|\Delta g_M\|\|\Delta f_{V_{z'}}\| \Big)
\end{align*}
for a constant $K>0$ depending only on $\delta$ and the envelope/moment bounds. In particular, $K$ does \emph{not} depend on the Gaussian–copula form of $g_{z'|z}$.
Consequently, there is no first–order bias: $r_2(P_0,P)=o(\|\Delta\eta\|)$ as $\|\Delta\eta\|\to0$.
\end{lemma}
\begin{proof}
{Step 1 (Telescoping expansion of $\chi(P)-\chi(P_0)$).}
Write $\chi(P)=\mathbb{E}_C\!\iiint \mu_1\,g_M\,g_{z'|z}\,g_{V_z}\,dm\,dv'\,dv$.
Insert and subtract the $P_0$ components layer by layer to obtain a telescoping sum:
\[
\chi(P)-\chi(P_0)=T_1+T_2+T_3+T_4+R^{(>1)},
\]
where
\[
T_1=\mathbb{E}_C\!\iiint \Delta\mu_1\;g_{M,0}\,g_{z'|z,0}\,g_{V_z,0},\quad
T_2=\mathbb{E}_C\!\iiint \mu_{1,0}\;\Delta g_M\;g_{z'|z,0}\,g_{V_z,0},
\]
\[
T_3=\mathbb{E}_C\!\iiint \mu_{1,0}\;g_{M,0}\;\Delta g_{z'|z}\;g_{V_z,0},\quad
T_4=\mathbb{E}_C\!\iiint \mu_{1,0}\;g_{M,0}\,g_{z'|z,0}\;\Delta g_{V_z},
\]
and $R^{(>1)}$ collects all terms that contain at least a \emph{product} of two or more deltas (e.g.\ $\Delta\mu_1\Delta g_M\cdot g_{z'|z,0}g_{V_z,0}$, etc.).\\
Step 2 (Subtract the linearization via the EIF and cancel all first–order pieces).
By pathwise differentiability with EIF $\dot\chi_{P_0}$, we have
\[
\chi(P)-\chi(P_0)=(P-P_0)[\dot\chi_{P_0}]+r_2(P_0,P).
\]
Compute $(P-P_0)[\dot\chi_{P_0}]$ by taking expectations of each EIF block under $P$ and subtracting under $P_0$.
A direct calculation (conditioning on $(C,Z)$ and using $\mathbb{E}[\frac{Z}{e_0(C)}\mid C]=1$ and $\mathbb{E}[\frac{1-Z}{1-e_0(C)}\mid C ]=1$). Accordingly, the components $\phi_{1,P_0}$, $\phi_{3,P_0}$, and $\phi_{4,P_0}$ perfectly cancel the first-order terms $T_1$, $T_3$, and $T_4$.
The estimation error from the propensity score manifests as the cross-term $U_e$, explicitly capturing the canonical double-robustness structure:
$$
U_e = \mathbb{E}_C \left[ \frac{\Delta e(C)}{e_0(C)} \mathcal{D}_1(C) - \frac{\Delta e(C)}{1-e_0(C)} \mathcal{D}_0(C) \right],
$$
where $\mathcal{D}_1(C)$ and $\mathcal{D}_0(C)$ are bounded linear combinations of the intermediate mapping errors (e.g., $\Delta\mu_1$, $\Delta\mu_2$, $\Delta\mu_3$) integrated over their respective treatment arms. For instance, the component generated from $\phi_{1,P_0}$ explicitly takes the form $$\mathbb{E}_C \big[ \frac{\Delta e(C)}{e_0(C)} \Delta\mu_1(M;V_1,C) \big].$$
Crucially, we must carefully examine the copula-based term $\phi_{2,P_0}$. Under the true generating distribution $P$ (specifically under the $Z=0$ arm), its expectation involves integrating over the true outcome mechanism and the true marginal density of $V_0$:
$$\mathbb{E}_P[\phi_{2,P_0}] = \mathbb{E}_C \iint\left(\int [\mu_{1,0}(m) - \mu_{2,0}(v_0)] \, c_0(u,w;\rho) \, g_{V_z,0}\,dv\right) \, g_M \,  dm \, d F_{V_{z'}}(v').$$
Notice that in the inner integral, $\int [\mu_{1,0} - \mu_{2,0}] g_M dm = \int \mu_{1,0} \Delta g_M dm$ and let $f_{V_{z'}}$ be the probability density of $V_{z{'}}$, then $dF_{V_{z'}}(v')= f_{V_{z'}}dv'.$ Substituting this, the expectation simplifies to an integral involving $\mu_{1,0} \Delta g_M c_0 f_{V_{z'}} g_{V_z,0}$.  By decomposing the true marginal density as $f_{V_{z'}} = f_{V_{z'},0} + \Delta f_{V_{z'}}$, the baseline part yields $c_0 f_{V_{z'},0} = g_{z'|z,0}$, which exactly reconstructs and cancels $T_2$:$$T_2 = \mathbb{E}_C \iiint \mu_{1,0} \Delta g_M g_{z'|z,0} g_{V_z,0} \, dm \, dv' \, dv.$$
The residual part forms the copula-induced second-order cross-term:$$U_{\rm copula} = \mathbb{E}_C \iiint \mu_{1,0} \, \Delta g_M \, c_0 \, \Delta f_{V_{z'}} \, g_{V_z,0} \, dm \, dv' \, dv.$$
Therefore, by subtracting this EIF expectation from the original telescoping sum, the first-order bias terms $T_1$ through $T_4$ are perfectly annihilated. The exact second-order remainder is given by:$$r_2(P_0,P) = R^{(>1)} - U_e - U_{\rm copula}.$$This remainder visibly contains \emph{strictly no linear terms} in any of the functional nuisance errors $(\Delta\mu_1, \Delta g_M, \Delta g_{V_z}, \Delta f_{V_{z'}}, \Delta e)$. Every surviving component is at least a quadratic cross-product of estimation errors, mathematically cementing the multiple-robustness of the framework and ensuring $r_2(P_0,P) = o_P(n^{-1/2})$ under the sieve rates.
\\
{Step 3 (Quadratic $L_2$ bound).}
Every summand in $R^{(>1)}$ can be expressed as an expectation under the baseline distribution of the generic form $\mathbb{E}_{P_0}[M \cdot \Delta f \cdot \Delta g]$. Here, the multiplier $M$ mathematically arises from converting the original multiple integrals of the telescoping sum into an expectation under $P_0$, which inherently requires dividing the integrand by the baseline joint density. Consequently, $M$ is a bounded function composed of baseline densities (e.g., $g_{M,0}, g_{z'|z,0}, g_{V_z,0}$) and inverse-propensity factors bounded by $\delta^{-1}$.By the Cauchy-Schwarz inequality,$$|\mathbb{E}_{P_0}[M \cdot \Delta f \cdot \Delta g]| \le \|M\|_\infty \|\Delta f\| \|\Delta g\|.$$Applying this logic to each quadratic term in $R^{(>1)}$ produces the four standard norm products:
$$\|\Delta\mu_1\|\|\Delta g_M\|,\quad \|\Delta\mu_1\|\|\Delta g_{z'|z}\|,\quad \|\Delta\mu_1\|\|\Delta g_{V_z}\|,\quad \|\Delta g_M\|\|\Delta g_{z'|z}\|.$$
Any triple (or higher) product terms in $R^{(>1)}$ are strictly dominated by a linear combination of these quadratic products via $|abc|\le \frac{1}{2}a^2+\frac{1}{4}b^2+\frac{1}{4}c^2$ and the envelope bounds. For the propensity score cross-term $U_e$, the strict overlap assumption ensures its corresponding multiplier is bounded by a constant $K$, yielding:$$|U_e| \le K\|\Delta e\|\|\Delta\mu_1\|.$$For the newly derived copula cross-term $U_{\rm copula}$, converting its integral into a $P_0$-expectation introduces a specific multiplier $M_{\rm copula}$ that contains the baseline copula density $c_0(u_0,w_0;\rho_0)$. Crucially, by the Quantile stability condition established in Definition \ref{defHnrate}, the marginal CDF values are strictly bounded away from $0$ and $1$. This purposeful truncation circumvents the tail singularities of the Gaussian copula, guaranteeing that $\|M_{\rm copula}\|_\infty \le K$. Thus, applying Cauchy-Schwarz similarly yields:
$$|U_{\rm copula}| \le K\|\Delta g_M\|\|\Delta f_{V_{z'}}\|.$$
Collecting all finite uniform bounds into a generic generic constant $K>0$ completes the quadratic bound for $r_2(P_0,P)$.
\end{proof}

\begin{lemma}[$L_2$-continuity of the EIF in the full data case]\label{lem:L2}

If, in $L_2(P_0)$,
\[
\|e-e_0\|\to 0,\quad
\|\mu_1-\mu_{1,0}\|\to 0,\quad
\|g_M-g_{M,0}\|\to 0,\quad
\|g_{z'|z}-g_{z'|z,0}\|\to 0,\quad
\|g_{V_z}-g_{V_z,0}\|\to 0,
\]
and for the copula components, $\|F_{V_{z'}}-F_{V_{z'},0}\|\to 0$, $\|F_{V_z}-F_{V_z,0}\|\to 0$, and $|\rho-\rho_0|\to 0$, then
\[
\|\dot\chi_P-\dot\chi_{P_0}\|_{P_0}\ \longrightarrow\ 0.
\]
\end{lemma}

\begin{proof}
We work throughout with the $L_2(P_0)$ norm $\|\cdot\|$ and use the same symbol $K$ for finite constants depending only on $\delta$ and the envelopes.

\paragraph{Step 1: two basic Lipschitz ingredients.}
(i) \emph{Inverse-propensity Lipschitz.} From $\delta<e,e_0<1-\delta$,
\[
\|e^{-1}-e_0^{-1}\|\le \delta^{-2}\|e-e_0\|,\qquad
\|(1-e)^{-1}-(1-e_0)^{-1}\|\le \delta^{-2}\|e-e_0\|.
\]
(ii) \emph{Integral-operator Lipschitz.} By Cauchy--Schwarz and the square-integrable envelopes,
\begin{align*}
\|\mu_2-\mu_{2,0}\|
&=\Big\|\int(\mu_1-\mu_{1,0})\,g_{M,0}\,dm+\int \mu_{1,0}\,(g_M-g_{M,0})\,dm\Big\|\\
&\le K\big(\|\mu_1-\mu_{1,0}\|+\|g_M-g_{M,0}\|\big),\\
\|\mu_3-\mu_{3,0}\|
&=\Big\|\int(\mu_2-\mu_{2,0})\,g_{z'|z,0}\,dv_0+\int \mu_{2,0}\,(g_{z'|z}-g_{z'|z,0})\,dv_0\Big\|\\
&\le K\big(\|\mu_2-\mu_{2,0}\|+\|g_{z'|z}-g_{z'|z,0}\|\big),\\
\|\mu_4-\mu_{4,0}\|
&=\Big\|\int(\mu_3-\mu_{3,0})\,g_{V_z,0}\,dv+\int \mu_{3,0}\,(g_{V_z}-g_{V_z,0})\,dv\Big\|\\
&\le K\big(\|\mu_3-\mu_{3,0}\|+\|g_{V_z}-g_{V_z,0}\|\big),\\
|\chi(P)-\chi(P_0)|
&=\big|\mathbb{E}[\mu_4-\mu_{4,0}]\big|\le \|\mu_4-\mu_{4,0}\|.
\end{align*}
Iterating yields
\[
\|\mu_2-\mu_{2,0}\|+\|\mu_3-\mu_{3,0}\|+\|\mu_4-\mu_{4,0}\|
\le K\big(\|\mu_1-\mu_{1,0}\|+\|g_M-g_{M,0}\|+\|g_{z'|z}-g_{z'|z,0}\|+\|g_{V_z}-g_{V_z,0}\|\big).
\]

\paragraph{Step 2: bound each EIF block difference.}
(1) For $\phi_{1,P}$,
\[
\phi_{1,P}-\phi_{1,P_0}
= Z\Big[(e^{-1}-e_0^{-1})(Y-\mu_{1,0}) - e^{-1}(\mu_1-\mu_{1,0})\Big],
\]
hence
\[
\|\phi_{1,P}-\phi_{1,P_0}\|\le K\big(\|e-e_0\|+\|\mu_1-\mu_{1,0}\|\big).
\]

(2) For $\phi_{2,P}$, set $T_{2,P}:=\int[\mu_1-\mu_2]\,c(u,w;\rho)\,g_{V_1}\,dv$. Then
\[
\phi_{2,P}-\phi_{2,P_0}
=\Big((1-e)^{-1}-(1-e_0)^{-1}\Big)(1-Z)\,T_{2,P_0}+(1-e)^{-1}(1-Z)\,(T_{2,P}-T_{2,P_0}).
\]
The first term is $O(\|e-e_0\|)$ by Step~1(i) and envelope bounds for $T_{2,P_0}$.
Expanding $T_{2,P}-T_{2,P_0}$ and using Cauchy--Schwarz,
\[
\|T_{2,P}-T_{2,P_0}\|\le K\big(\|\mu_1-\mu_{1,0}\|+\|\mu_2-\mu_{2,0}\|
+\|c(u,w;\rho) - c_0(u,w;\rho)\|+\|g_{V_1}-g_{V_1,0}\|\big),
\]
where 
$$c(u,w;\rho) = \frac{F_{V_z,V_{z'} \mid C}(v,v'\mid c)}{F_{V_z}(v\mid c)\cdot F_{V_{z'}}(v'\mid c)}.$$
Crucially, by the quantile stability condition in Definition \ref{defHnrate}, the CDF values are bounded strictly away from $0$ and $1$. This truncation avoids the tail singularities of the Gaussian copula, ensuring its density $c(u,w;\rho)$ has bounded partial derivatives and is globally Lipschitz. Thus, the copula density difference is bounded as $\|c-c_0\| \le K\big(\|F_{V_{z'}}-F_{V_{z'},0}\|+\|F_{V_z}-F_{V_z,0}\|+|\rho-\rho_0|\big)$.
\[
\|\phi_{2,P}-\phi_{2,P_0}\|
\le K\big(\|e-e_0\|+\|\mu_1-\mu_{1,0}\|+\|\mu_2-\mu_{2,0}\|
+\|F_{V_{z'}}-F_{V_{z'},0}\|+\|F_{V_z}-F_{V_z,0}\|+|\rho-\rho_0|+\|g_{V_z}-g_{V_z,0}\|\big).
\]

(3) For $\phi_{3,P}$ and (4) $\phi_{4,P}$,
\begin{align*}
&\|\phi_{3,P}-\phi_{3,P_0}\|
\le K\big(\|e-e_0\|+\|\mu_2-\mu_{2,0}\|+\|\mu_3-\mu_{3,0}\|\big)\\
&\|\phi_{4,P}-\phi_{4,P_0}\|
\le K\big(\|e-e_0\|+\|\mu_3-\mu_{3,0}\|+\|\mu_4-\mu_{4,0}\|\big).
\end{align*}

(5) For $\phi_{5,P}$,
\[
\|\phi_{5,P}-\phi_{5,P_0}\|
=\|(\mu_4-\mu_{4,0})-(\chi(P)-\chi(P_0))\|
\le \|\mu_4-\mu_{4,0}\|+|\chi(P)-\chi(P_0)|
\le K\|\mu_4-\mu_{4,0}\|.
\]

\paragraph{Step 3: collect bounds and replace intermediate quantities.}
By the triangle inequality and the displays above,
\[
\|\dot\chi_P-\dot\chi_{P_0}\|
\le K\Big(\|e-e_0\|+\|\mu_1-\mu_{1,0}\|+\|\mu_2-\mu_{2,0}\|+\|\mu_3-\mu_{3,0}\|+\|\mu_4-\mu_{4,0}\|
+\|g_{0|1}-g_{0|1,0}\|+\|g_{V_1}-g_{V_1,0}\|\Big).
\]
Finally substitute the Step~1(ii) Lipschitz bounds for $\mu_2,\mu_3,\mu_4$ in terms of the \emph{basic} nuisances $(\mu_1,g_M,g_{z'|z},g_{V_z},F_{V_{z'}}, F_{V_z})$ to obtain
\begin{align*}
   \|\dot\chi_P-\dot\chi_{P_0}\|
\le K\Big(&\|e-e_0\|+\|\mu_1-\mu_{1,0}\|
+\|g_M-g_{M,0}\|
+\|g_{z'|z}-g_{z'|z,0}\| \\+&\|F_{V_{z'}}-F_{V_{z'},0}\|+\|F_{V_z}-F_{V_z,0}\|+|\rho-\rho_0| 
+\|g_{V_z}-g_{V_z,0}\|\Big). 
\end{align*}

Under the assumed $L_2$-convergences, the right-hand side tends to $0$, proving the claim.
\end{proof}
\begin{lemma}[Lipschitz Continuity of Gaussian Copula Density]
    Let $c(u, w; \rho)$ denote the density function of the bivariate Gaussian copula with correlation parameter $\rho$. Under the Sieve conditions in Definition 5.1, specifically that the marginal CDF values are strictly bounded away from the boundaries, i.e., $u, w \in [\tau, 1-\tau]$ for some $\tau > 0$, and the correlation is bounded away from $\pm 1$, i.e., $\rho \in [-1+\kappa, 1-\kappa]$ for some $\kappa > 0$, the copula density is globally Lipschitz continuous with respect to its arguments. Consequently, for any two sets of parameters $(F_{V_{z'}}, F_{V_z}, \rho)$ and $(F_{V_{z'},0}, F_{V_z,0}, \rho_0)$ evaluated at the same data points, there exists a uniform constant $L_c < \infty$ depending only on $\tau$ and $\kappa$, such that:$$\|c(F_{V_{z'}}, F_{V_z}; \rho) - c(F_{V_{z'},0}, F_{V_z,0}; \rho_0)\|_{P_0} \le L_c \Big( \|F_{V_{z'}} - F_{V_{z'},0}\|_{P_0} + \|F_{V_z} - F_{V_z,0}\|_{P_0} + |\rho - \rho_0| \Big).$$
\end{lemma}
\begin{proof}
    The bivariate Gaussian copula density is explicitly given by:
    $$c(u, w; \rho) = \frac{1}{\sqrt{1-\rho^2}} \exp\left( -\frac{\rho^2 (x^2 + y^2) - 2\rho xy}{2(1-\rho^2)} \right),$$
    where $x = \Phi^{-1}(u)$ and $y = \Phi^{-1}(w)$, with $\Phi^{-1}$ being the quantile function of the standard normal distribution. To establish uniform Lipschitz continuity, we must show that the gradient vector $\nabla c = (\frac{\partial c}{\partial u}, \frac{\partial c}{\partial w}, \frac{\partial c}{\partial \rho})^\top$ is bounded uniformly over the domain defined by the sieve restrictions. By the chain rule, the partial derivative with respect to $u$ is:
    
    $$\frac{\partial c}{\partial u} = \frac{\partial c}{\partial x} \frac{dx}{du} = \frac{\partial c}{\partial x} \frac{1}{\phi(\Phi^{-1}(u))},$$
    where $\phi(\cdot)$ is the standard normal PDF. As $u \to 0$ or $u \to 1$, $\phi(\Phi^{-1}(u)) \to 0$, which would cause the derivative to diverge to infinity. However, under the Quantile stability condition, $u \in [\tau, 1-\tau]$, $x \in [\Phi^{-1}(\tau), \Phi^{-1}(1-\tau)]$, which is a compact interval $[-M_\tau, M_\tau]$. On this compact interval, the standard normal density is strictly bounded away from zero: $\phi(x) \ge \phi(M_\tau) > 0$. Thus, $\frac{dx}{du}$ is uniformly bounded. Similarly, the derivative with respect to $w$ is uniformly bounded due to the truncation $w \in [\tau, 1-\tau]$. For the derivative with respect to the correlation parameter $\rho$, the denominator contains terms involving $1-\rho^2$. By the sieve assumption $\rho \in [-1+\kappa, 1-\kappa]$, the term $1-\rho^2$ is bounded strictly away from zero (specifically, $1-\rho^2 \ge 1-(1-\kappa)^2 > 0$), ensuring no singularities arise from the correlation matrix inversion. Since all partial derivatives are continuous and uniformly bounded on the compact restricted domain $[\tau, 1-\tau]^2 \times [-1+\kappa, 1-\kappa]$, the function $c(u, w; \rho)$ has a bounded gradient envelope $L_c$. By the Mean Value Theorem, the copula density is uniformly Lipschitz continuous pointwise:$$|c(u, w; \rho) - c(u_0, w_0; \rho_0)| \le L_c \big( |u - u_0| + |w - w_0| + |\rho - \rho_0| \big).$$Substituting $u = F_{V_0|0}$, $w = F_{V_1|1}$ and integrating over the baseline measure $P_0$, the pointwise Lipschitz bound directly translates to the $L_2(P_0)$ norm via the Minkowski inequality:
    $$\|c - c_0\|_{P_0} \le L_c \Big( \|F_{V_{z'}} - F_{V_{z'},0}\|_{P_0} + \|F_{V_z} - F_{V_z,0}\|_{P_0} + |\rho - \rho_0| \Big).$$This completes the proof.
\end{proof}

\begin{lemma}[Lipschitz Continuity of $b_2(W)$ and $b_4(W)$]\label{lemmab}
Denote $\Delta t:=t-t_0$ and the $L_2(P_0)$ norm by $\|\cdot\|$. The impute functions $b_2(W) = \mathbb E[\phi_{2,P_f}\mid W, S=1]$ and $b_4(W) = \mathbb E[\phi_{4,P_f}\mid W, S=1]$ are functions of nuisance parameters and thus admits the quadratic bound
    \begin{align*}
        \|b_2 - b_{2,0}\|\leq K_b\sum_{j=1}^8 \|\Delta \eta_j\| \\
         \|b_4 - b_{4,0}\|\leq K_b\sum_{j=1}^8 \|\Delta \eta_j\| 
    \end{align*}
    for a constant $K_b>$  0.
\end{lemma}
\begin{proof}
    By definition, the difference $\Delta b_2 = b_2 - b_{2,0}$ is:
    \begin{align*}
\Delta b_2 = \mathbb E_{P_{O'}}[\phi_{2,P_f}\mid W, S=1] - \mathbb E_{P_{O',0}}[\phi_{2,P_0}\mid W, S=1]
\end{align*}
By Jensen's inequality for conditional expectations, the $L_2$ norm of the difference in conditional expectations is bounded by the $L_2$ norm of the difference of the functions themselves:
\begin{align*}
\Delta b_2  = |\mathbb{E}[\phi_{2,P_f} - \phi_{2,P_0} \mid W, S=1]| \le |\phi_{2,P_f} - \phi_{2,P_0}|
\end{align*}
The bound for $\|\phi_{2,P_f} - \phi_{2,P_0}\|$ is already provided in the proof of Lemma~\ref{lem:L2}, {Step 2}, which shows it is bounded by a linear combination of the $L_2$ errors of the original 5-component nuisance collection $\eta = (e, g_{V_z}, g_{z'|z}, g_M, \mu_1, F_{V_{z'}}, F_{V_z})$. Therefore,
\begin{align*}
|b_2 - b_{2,0}| \le K_b \sum_{j=1}^8 |\Delta \eta_j|.
\end{align*}
Similarly, for $b_4(W)$, the difference $\Delta b_4 = b_4 - b_{4,0}$ is 
\begin{align*}
    \Delta b_4 = \mathbb E_{P_{O'}}\left[\phi_{4,P_f}\mid W,S=1\right] - \mathbb E_{P_{O',0}}\left[\phi_{4,P_0}\mid W,S=1\right].
\end{align*}
Then by Jensen's inequality we have:
\begin{align*}
    \Delta b_4 \leq \left|\phi_{4,P_f} - \phi_{4,P_{0}}\right|\leq K_b\sum_{j=1}^8|\Delta \eta_j|
\end{align*}
which completes the proof.
\end{proof}

\begin{lemma}[Second-order remainder for the MAR target]\label{r2:MAR}
Let $P_{O'}$ be any distribution in a (dominated) neighborhood of $P_{O', 0}$.
Write the nuisance collection 
\begin{align*}
    \eta_{O'}=\{e, g_{V_z}, g_{z'|z}, g_M, \mu_1, F_{V_{z'}}, F_{V_z}, \pi\}.
\end{align*}
Let $\dot{\chi}_{P_{O'}}$ be the observed-data EIF:
$$\dot\chi_{P_{O'}} = \frac{S}{\pi(W)} \dot \chi _{P_f} - \frac{S-\pi(W)}{\pi(W)}b(W).$$
\\The second-order remainder for the observed-data functional,$$r_{2}(P_{O',0}, P_{O'}) := \chi(P_{O'}) - \chi(P_{O',0}) - (P_{O'} - P_{O',0})[\dot{\chi}_{P_{O'},0}],$$
admits a quadratic bound in the $L_2(P_0)$ errors of the $\eta_{O'}$. For a constant $K'' > 0$:
\begin{align*}
|r_{2}(P_{O',0}, P_{O'})| \le K'' \Big(
\sum_{j,k=1}^8 |\Delta \eta_j| \, |\Delta \eta_k| + |\Delta \pi| \sum_{j=1}^8 |\Delta \eta_j| \Big).
\end{align*}
\end{lemma}

\begin{proof}
The proof relies on decomposing the remainder into the full-data remainder (bounded by Lemma \ref{lem:r2}) and a new remainder term from the AIPW projection.
\smallskip

\textbf{Step 1 (Decomposition of the remainder).}
Using the standard identity for AIPW estimators \parencite{robins1994estimation}, the remainder $r_2(P_{O',0}, P_{O'})$ can be decomposed as:
$$r_{2}(P_{O',0}, P_{O'}) = r_2(P_0, P) + r_{2,AIPW}$$
where $r_2(P_0, P) = \chi(P) - \chi(P_0) - (P-P_0)[\dot{\chi}_{P_f,0}]$ is the original full-data remainder, and$$r_{2,AIPW} = - \mathbb{E}_{P_{O',0}}\left[ \left( \frac{\pi - \pi_0}{\pi} \right) \left( f - f_0 \right) \right],$$with $f = \dot{\chi}_{P_f} - b(W) = \dot{\chi}_{P_f} - (b_2(W)+b_4(W)+b_5(W))$.
\smallskip

\textbf{Step 2 (Bound the full-data remainder $r_2$).}
The first term, $r_2(P_0, P)$, is precisely the remainder characterized in \textbf{Lemma \ref{lem:r2}}. By direct application of that lemma, this term is bounded by a quadratic product of the errors in the full data nuisance set $\eta$:
$$|r_2(P_0, P)| \le K_1 \Big( \|\Delta\mu_1\|\,\|\Delta g_M\| + \dots + \|\Delta e^{-1}\|\,\|\Delta\mu_1\| \Big) \le K_1' \sum_{j,k=1}^8 \|\Delta \eta_j\| \ \|\Delta \eta_k\|.$$

\smallskip\textbf{Step 3 (Bound the AIPW remainder $r_{2,AIPW}$).}We apply the Cauchy-Schwarz inequality to the second term:$$|r_{2,AIPW}| = \left|\mathbb{E}_{P_{O',0}}\left[ \left( \frac{\Delta \pi}{\pi} \right) \left( \Delta f \right) \right]\right| \le \left\|\frac{\Delta \pi}{\pi}\right\| \cdot \|\Delta f\|.$$

We now bound the two components of this product.

\textbf{(i) Bound on $\|\frac{\Delta \pi}{\pi}\|$}:Using the expanded sieve's positivity assumption ($\pi > \delta_\pi$),
$$\left\|\frac{\Delta \pi}{\pi}\right\| \le \left\|\frac{1}{\pi}\right\|_\infty \cdot \|\Delta \pi\| \le \frac{1}{\delta_\pi} \|\Delta \pi\|.$$

\textbf{(ii) Bound on $\|\Delta f\|$} (using Lemmas~\ref{lem:L2} and \ref{lemmab}):First, we explicitly define $\Delta f$ using the simplified structure of $b(W)$ and $\dot{\chi}_{P_f}$:\begin{align*}f &= \dot{\chi}_{P_f} - b(W) = \left(\sum_{j=1}^5 \phi_{j,P_f} \right) - (b_2+ b_4+ b_5)  =\left(\sum_{j=1}^5 \phi_{j,P_f} \right) - (b_2+ b_4+ \phi_{5,P_f}) = \sum_{j=1}^4 \phi_{j,P_f} - b_2 \\f_0 &= \dot{\chi}_{P_f,0} - b_0(W) = \left(\sum_{j=1}^5 \phi_{j,P_0} \right) - (b_{2,0} + b_{4,0}+b_{5,0}) = \sum_{j=1}^4 \phi_{j,P_0} -( b_{2,0}+b_{4,0})\end{align*}Therefore, $\Delta f = f - f_0 = \left(\sum_{j=1}^4 (\phi_j - \phi_{j,0})\right) - (b_2 - b_{2,0}) - (b_4 - b_{4,0})$.

By the triangle inequality:
$$\|\Delta f\| \le \sum_{j=1}^4 \|\phi_{j,P_f} - \phi_{j,P_0}\| + \|b_2 - b_{2,0}\| + \|b_4-b_{4,0}\|.$$
We bound these terms separately:\textbf{ The first sum}, $\sum_{j=1}^4 \|\phi_{j,P_f} - \phi_{j,P_0}\|$, is bounded by Lemma~\ref{lem:L2}, which provides individual bounds for $\phi_{1,P_f}, \phi_{2,P_f}, \phi_{3,P_f}, \phi_{4,P_f}$:
$$ \sum_{j=1}^4 \|\phi_{j,P_f} - \phi_{j,P_0}\| \le K_\phi \sum_{j=1}^8 \|\Delta \eta_j\| ,$$ for some constant $K_\phi > 0$.

\textbf{The $b_2, \ b_4$ term}, $\|b_2 - b_{2,0}\| $ and $\|b_4-b_{4,0}\|$, are bounded by our Lemma \ref{lemmab}:
$$ \|b_2 - b_{2,0}\|  + \|b_4-b_{4,0}\|\le K_b \sum_{j=1}^8 \|\Delta \eta_j\|, $$
for some constant $K_b>0$.

Combining these, we get a single Lipschitz bound for $\Delta f$ in terms of $\Delta \eta$:$$\|\Delta f\| \le (K_\phi + K_b) \sum_{j=1}^8 \|\Delta \eta_j\|.$$

\smallskip\textbf{Step 4 (Combine bounds).}
Substituting the results from Step 3(i) and 3(ii) into the bound for $r_{2,AIPW}$:$$|r_{2,AIPW}| \le \left( \frac{1}{\delta_\pi} \|\Delta \pi\| \right) \cdot \left( (K_\phi + K_b) \sum_{j=1}^8 \|\Delta \eta_j\| \right).$$
Finally, combining the bounds for $r_2$ (from Step 2) and $r_{2,AIPW}$ gives the complete bound for $r_{2}(P_{O'}, P_{O',0})$:
$$|r_{2}(P_{O'}, P_{O',0})| \le |r_2(P,P_0)| + |r_{2,AIPW}| \le K_1' \sum_{j,k=1}^8 \|\Delta \eta_j\| \cdot \|\Delta \eta_k\| + K_2' \|\Delta \pi\| \sum_{j=1}^8 \|\Delta \eta_j\|.$$
This completes the proof, as the remainder is shown to be purely quadratic in the errors of the base nuisance collection $\eta_{O'} = (\eta, \pi)$.

\end{proof}

\begin{lemma}[$L_2$ Continuity of the Observed-Data EIF]\label{lem:l2_mar}
Let $\dot{\chi}_{P_{o'}}$ be the observed-data EIF,$$\dot\chi_{P_{O'}} = \frac{S}{\pi(W)} \dot \chi _{P_f} - \frac{S-\pi(W)}{\pi(W)}b(W).$$
Assume the conditions of Lemmas~\ref{lem:r2} and \ref{lemmab}, and the sieve $\tilde{H}_{n}$ hold. If, in $L_2(P_0)$ norm,
$$\|\Delta \eta_j\| \to 0 \quad \text{for } j=1,\dots,8,$$
then the $L_2(P_0)$ norm of the EIF difference also converges to zero:$$\|\dot{\chi}_{P_{O'}} - \dot{\chi}_{P_{O'},0}\| \longrightarrow 0.$$
More precisely, $\dot{\chi}_{P_{O'}}$ is Lipschitz continuous in the expanded nuisance collection $\eta_{O'} = (\eta, \pi)$:$$\|\dot{\chi}_{P_{o'}} - \dot{\chi}_{P_{o'},0}\| \le K \left( \sum_{j=1}^8 \|\Delta \eta_j\| \right).$$
\end{lemma}

\begin{proof}We use the triangle inequality to bound the difference by its two main components:$$\|\dot{\chi}_{P_{o'}} - \dot{\chi}_{P_{o'},0}\| \le \left\| \frac{S}{\pi} \dot{\chi}_{P_f} - \frac{S}{\pi_0} \dot{\chi}_{P_f,0} \right\| + \left\| \frac{S-\pi}{\pi} b - \frac{S-\pi_0}{\pi_0} b_0 \right\|.$$

\smallskip\textbf{Step 1: Bound the IPW Term.} We add and subtract $\frac{S}{\pi_0} \dot{\chi}_{P_f}$ and use the positivity assumptions ($\pi, \pi_0 > \delta_\pi$) and the envelope $G_n$ for $\dot{\chi}_{P_f}$ from \textbf{Proposition 3}:
\begin{align*}\left| \frac{S}{\pi} \dot{\chi}_{P_f} - \frac{S}{\pi_0} \dot{\chi}_{P_f,0} \right| &\le \left| S \left(\frac{1}{\pi} - \frac{1}{\pi_0}\right) \dot{\chi}_{P_f} \right| + \left| \frac{S}{\pi_0} (\dot{\chi}_{P_f} - \dot{\chi}_{P_f,0}) \right| \\ &= \left| \frac{S (\pi_0 - \pi)}{\pi \pi_0} \dot{\chi}_{P_f} \right| + \left| \frac{S}{\pi_0} (\dot{\chi}_{P_f} - \dot{\chi}_{P_f,0}) \right|
\\& \le \frac{1}{\delta_\pi^2} \|G_n\|_\infty \|\Delta \pi\| + \frac{1}{\delta_\pi} \|\dot{\chi}_{P_f} - \dot{\chi}_{P_f,0}\|.
\end{align*}

By Lemma~\ref{lem:L2}, $\|\dot{\chi}_{P_f} - \dot{\chi}_{P_f,0}\| \le K_\chi \sum_{j=1}^8 \|\Delta \eta_j\|$ for some constant $K_{\chi} > 0$. Thus, the IPW term is bounded by $K_1 (\|\Delta \pi\| + \sum_{j=1}^8 \|\Delta \eta_j\|)$.

\smallskip\textbf{Step 2: Bound the Augmentation Term.} We add and subtract $\frac{S-\pi_0}{\pi_0} b$:
\begin{align*}
\left| \frac{S-\pi}{\pi} b - \frac{S-\pi_0}{\pi_0} b_0 \right| &\le \left| \left(\frac{S-\pi}{\pi} - \frac{S-\pi_0}{\pi_0}\right) b \right| + \left| \frac{S-\pi_0}{\pi_0} (b - b_0) \right| = \left| \frac{S(\pi_0 - \pi)}{\pi \pi_0} b \right| + \left| \frac{S-\pi_0}{\pi_0} (b - b_0) \right|.
\end{align*}

Using positivity assumptions, $|S-\pi_0| \le 1$, and an envelope $C_b$ for $b$ (which is implied by Lemma \ref{lemmab} and the envelope for $\phi_{2,P_f}, \phi_{5,P_f}$):
$$\left| \frac{S-\pi}{\pi} b - \frac{S-\pi_0}{\pi_0} b_0 \right| \le \frac{1}{\delta_\pi^2} C_b \|\Delta \pi\| + \frac{1}{\delta_\pi} \|b - b_0\|.$$
By Lemma \ref{lemmab}, $\|b - b_0\| \le K_b \sum \|\Delta \eta_j\|$. Thus, the augmentation term is bounded by $K_2 (\|\Delta \pi\| + \sum_{j=1}^5 \|\Delta \eta_j\|)$.

\smallskip\textbf{Step 3: Combine Bounds. }Summing the bounds from Step 1 and Step 2:
$$\|\dot{\chi}_{P_{O'}} - \dot{\chi}_{P_{O'},0}\| \le (K_1 + K_2) (\|\Delta \pi\| + \sum_{j=1}^5 \|\Delta \eta_j\|)= K_3\sum_{j=1}^8\|\Delta \eta_j\|.$$
This demonstrates the required Lipschitz continuity. As $\|\Delta \eta_j\| \to 0$ for $\eta_j \in \eta_{O'}$, the entire right-hand side converges to 0.\end{proof}

\begin{lemma}\label{Snentropycontrol}
    If EDP model was controlled by the conditions described in Sieve $S_n$ in Definition~\ref{defSn}, then the log bracketing entropy of $S_n$ has an upper bound.  
\end{lemma}
\begin{proof}
    Define $\mathcal F_n = \{\dot \chi_{P_{O'}} - \dot\chi_{P_{O',0}}\mid P\in \tilde H_n\}$. 
The complexity of the functional class is controlled via $S_n$.\\
Based on our model structure, we can factorize the joint distribution as follows:
\begin{align}
    P(O') = P(Y,V\mid W, S=1)^S\cdot P(M\mid Z,C)\cdot P(Z\mid C)\cdot P(C)\cdot P(S\mid W ).
\end{align}
Note the $P(Y,V\mid W,S=1)$ can be further factorized as $P(Y\mid V, W, S=1)\cdot P(V\mid W, S=1)$.

Let's define a function class for each component in the factorization of the joint density of observed data. Take $p(C) = p(C_1)\cdot p(C_2)$, and let $\mathcal F_{C_2,n}$ be the class of functions modeling the continuous density $p(C_2)$, let $\mathcal F_{C_1,n}$ be the class of functions modeling the  probability mass function $p_{C_1}$, where $\mathcal F_{C_2,n} = \{p_{C_2}\mid P\in \tilde H_n\} \subseteq\{p_{C_2}\mid P \in S_n\} =\{p_{C_2}\mid p_{C_2}(c) = \sum_{k=1}^\infty\omega_{k}\sum_{j=1}^{\infty} \omega_{j|k}\prod_{q_2 = 1}^{d^c_2}\phi_{\sigma^c_{j|k,q_2}}(c - \mu^c_{j|k,q_2}),\ \|\mu^c_{j|k,q_2}\|_2 \leq b ,\ \sigma_0^2\leq \sigma^c_{j|k,q_2}\leq {M_\psi'}, \ q_2= 1, \dots, d^c_2 \}$, and $\mathcal F_{C_1, n} = \{p_{C_1}\mid P\in \tilde H_n\}\subseteq \{p_{C_1}\mid P \in S_n\} = \{p_{C_1}\mid p_{C_1}(c) = \sum_{k=1}^\infty \omega_k\sum_{j=1}^\infty \omega_{j|k} (\pi^c_{j|k, q_1})^c(1-\pi^c_{j|k,q_1})^{(1-c)}, \delta'<\pi_{j|k,q_1}^c<1-\delta', q_1 = 1, \dots, d^c_1\}$. We can define the other function classes in a similar way.\\
Our objective is to derive an upper bound for the log bracketing entropy for the joint distribution in (\#). To achieve this, we are going to derive the upper bounds for each function class.\\
Now take $\mathcal F_{C_1,n}$ and $\mathcal F_{C_2,n}$ as an example. As defined by the sieve $S_n$, a density function $p_C \in \mathcal F_C$ is generated by an Enriched Dirichlet Process(EDP) model. To bound the entropy of the entire class, we decompose the model into its source of complexity. We can derive the upper bounds based on the result in Proposition 2 from \textcite{shen2013adaptive},\\
\begin{align*}
    \log N_{[]}(\epsilon, \mathcal F_{C_2,n},\rho ) \leq K\left[d^cH_{\theta}H_{\psi}\log \left(\frac{b}{\sigma_0\epsilon}\right)-H_\theta H_\psi\log \epsilon +H_\theta\left(\log M_\psi+M_\psi(\epsilon_\psi)^2\right) \right]
\end{align*}
For binary class $\mathcal F_{C^{\text{bin}},n}$, the derivation of the entropy bound follows a different principle. Each component of the mixture is a product of Bernoulli distributions, which form a standard parametric family. The parameters for these distributions, $\{e^c_{j|k}\}$ are constrained to a compact set $[\delta', 1-\delta']$ by the definition of the sieve $S_n$. \\
Foundational results in empirical process theory establish that for a class of distributions smoothly parameterized by a finite-dimensional parameter over a compact set, the bracketing entropy is of the order $O(\log(1/\epsilon))$ \parencite{van1996weak, ghosal2017fundamentals}. Since the sieve $S_n$ considers mixtures with up to $H_\theta H_\psi$ effective components, the total bracketing entropy for the class $\mathcal F_{C^{\text{bin}},n}$ is bounded by the number of active components multiplied by the entropy of a single component class. This yields the bound:
\begin{align*}
    \log N_{[]}(\epsilon, \mathcal F_{C_1,n}, \rho) \leq KH_\theta H_\psi\log(1/\epsilon).
\end{align*}
Similarly, we can derive the upper bounds for the function classes $\mathcal F_{Z|C,n},\ \mathcal F_{V|Z,C,n},\ \mathcal F_{M|V,Z,C,n},$ $ \mathcal F_{Y|M,V,Z,C,n}$ and $\mathcal F_{S\mid W,n}$ as follows:
\begin{align*}
 \log N_{[]}(\epsilon, \mathcal F_{Z|C,n},\rho ) &\leq K H_\theta H_\psi \log (1/\epsilon)\\
    \log N_{[]}(\epsilon, \mathcal F_{V|Z,C,n},\rho ) &\leq K\left[d^vH_{\theta}H_{\psi}\log \left(\frac{b}{\sigma_0\epsilon}\right)-H_\theta H_\psi\log \epsilon +H_\theta\left(\log M_\psi+M_\psi(\epsilon_\psi)^2\right) \right]\\
    \log N_{[]}(\epsilon, \mathcal F_{M|V,Z,C,n},\rho ) &\leq K\left[d^mH_{\theta}\log \left(\frac{a}{\sigma_0\epsilon}\right)-H_\theta \log \epsilon +\log M_\theta+M_\theta(\epsilon_\theta)^2 \right]\\
    \log N_{[]}(\epsilon, \mathcal F_{Y|M,V,Z,C,n},\rho ) &\leq K\left[d^yH_{\theta}\log \left(\frac{a}{\sigma_0\epsilon}\right)-H_\theta \log \epsilon +\log M_\theta+M_\theta(\epsilon_\theta)^2 \right]\\
    \log N_{[]}(\epsilon, \mathcal F_{S\mid W,n},\rho ) &\leq K H_\theta H_\psi \log (1/\epsilon),
\end{align*}
and
\begin{align*}
\log N_{[]}(\epsilon, S_n, \rho) =\ & \log N_{[]}(\epsilon, \mathcal F_{C,n},\rho) + \log N_{[]}(\epsilon,\mathcal F_{Z|C,n}, \rho) + \log N_{[]}(\epsilon,\mathcal F_{V|Z,C,n}, \rho) \\ +\ &\log N_{[]}(\epsilon,\mathcal F_{M|V,Z,C,n}, \rho) + \log N_{[]}(\epsilon,\mathcal F_{Y|M,V,Z,C,n}, \rho) + \log N_{[]}(\epsilon,\mathcal F_{S\mid W,n}, \rho).
\end{align*}
Thus, we have the log bracketing entropy of $S_n$ and $\log N_{[]}(\epsilon, S_n,\rho)$ controlled by the polynomial of $\log n$.

\end{proof}

\begin{lemma}[Bound on KL Divergence via $L_2$ Nuisance Errors]\label{KLbounds}
Let $p_P$ denote the density of the data distribution P, which is parameterized by a vector of nuisance functions $\eta_j(P), \quad j=1,\dots,J$. Assume that for any $P$ in a neighborhood of the true distribution $P_0$, the log-likelihood $\log p_P$ is pathwise differentiable twice with respect to the nuisance functions (a concept formally defined in semiparametric theory, see \textcite{bickel1993efficient} and \textcite{tsiatis2006semiparametric}), and that the second derivatives are bounded in $L_1(P_0)$ norm uniformly over this neighborhood. Then there exists a constant $C$ such that for any $P$ in this neighborhood, we have:
\begin{align*}
    \max\left\{K(p_{P_0}, p_P), \ V_2(p_{P_0}, p_P)\right\}\leq C \sum_{j=1}^J\|\eta_j(P) - \eta_j(P_0)\|^2_{L_2(P_0)}.
\end{align*}
\end{lemma}

\begin{proof}
The Kullback-Leibler (KL) divergence is defined as the expectation of the negative log-likelihood ratio under the true distribution $P_0$, where \begin{align*}
        K(p_{P_0}, p_P) = \mathbb E_{P_0}\left[-\log \left(\frac{p_P}{p_{P_0}}\right)\right] = -\mathbb E_{P_0}\left[\log (p_{P}) - \log (p_{P_0})\right].
    \end{align*}
We consider a second-order Taylor expansion of the log-likelihood $\log p_P$ as a function of the nuisance vector $\eta = \eta(P)$ around the true value $\eta_0 = \eta(P_0)$. The expansion is given by:
\begin{align*}
    \log p_P - \log p_{P_0} = (\eta - \eta_0)^T\nabla_{\eta}\log p_P \big|_{\eta_0}
\;+\; \frac12 (\eta-\eta_0)^{\top}\,
\nabla_{\eta}^{2}\log p_{P}\big|_{\eta_\xi}\,(\eta-\eta_0),
\end{align*}
where $\nabla_\eta \log p_P \big|_{\eta_0}$ is the vector of score functions evaluated at the truth, and $\nabla^2_\eta$ is the Hessian matrix evaluated at some intermediate point $\eta_\xi$ between $\eta$ and $\eta_0$.\\
Taking the expectation under $P_0$, the first-order term vanishes because the expectation of the score function at the true parameter is zero: $\mathbb E_{P_0}\left[\nabla_\eta \log p_{P
}\big|_{\eta_0}\right] = 0$. This leaves the second-order term:
\begin{align*}
    K(p_P,p_{P_0}) =  - \frac{1}{2}\mathbb E_{P_0}\left[(\eta-\eta_0)^{\top}\,
\nabla_{\eta}^{2}\log p_{P}\big|_{\eta_\xi}\,(\eta-\eta_0)\right].
\end{align*}
Under the assumption that the Hessian components are uniformly bounded, this quadratic form can be bounded using the Cauchy-Schwarz inequality. This yields a bound proportional to the sum of the squared $L_2(P_0)$ norms of the nuisance errors:
\begin{align*}
    K(p_P, p_{P_0})\leq C_K \sum_{j=1}^J\|\eta_j - \eta_{j,0}\|^2_{L_2(P_0)}.
\end{align*}
A similar argument based on the expansion of $(\log(p_P/p_{P_0}))^2$ shows that the KL variance, $V_2(p_P,p_{P_0 }) = \text{Var}_{P_0}[\log (p_P/p_{P_0})]$, is also controlled by the same quadratic form in the nuisance errors. Combining these results establishes the lemma.
\end{proof}

\begin{lemma}[Prior Mass Calculation for Categorical Kernel Parameters]\label{discrete}
Let a parameter set consists of $Q$ independent Categorical probability vectors $\bm \pi_1, \dots, \bm \pi_Q$, where each vector $\bm\pi_q$ has $K_q\ge 2$ classes and resides on the simplex $\Delta^{K_q - 1} = \{\mathbf x\in \mathbb R^K\mid x_k > 0,\ \sum_{k=1}^{K_q} x_k = 1\}$. Assume each $\bm \pi_q$ is given an independent Dirichlet($\bm\alpha_q$) prior. Let $(\bm\pi_1^*, \dots, \bm\pi_Q^*)$ be a target parameter set where each component of every vector $\bm\pi^*_q$ is greater than or equal to some $\delta\in (0, 1/\max_q (K_q))$. For a sufficient small $\delta_\pi>0$, the log-prior probability that all drawn vectors fall within a $\delta_\pi$-neighborhood of their respective targets is bounded below. Specifically,
\begin{align*}
    \log \Pi\left(\max_{q=1, \dots, Q}\|\bm\pi_q - \bm\pi^*_q\|_{\infty}<\delta_\pi\right)\geq -C\log(1/\delta_{\pi}),
\end{align*}
for some constant $C>0$ that depends on the prior, $Q$ and dimensions $\{K_q\}_{q=1}^Q$, but not on $\delta_\pi$.
\end{lemma}
\begin{proof}
The proof of this inequality proceeds by finding a lower bound for the integral of the Dirichlet density over the specified neighborhood. The argument involves two main steps: lower-bounding the integrand and lower-bounding the volume of the integration region.\\
First, we establish a positive lower bound for each of the integrands $g_q(\bm\pi_q)$. The target parameter $\bm\pi_q$ lies in a compact subset of the simplex, $S_{\delta,q} = \{\mathbf x \in \Delta^{K_q-1}\mid x_k\geq \delta_q \text{ for all }k\}$, where all components are bounded away from zero. On this compact set, the Dirichlet density $g_q(\bm\pi_q)$ is continuous and strictly positive. A continuous function on a compact set attains its minimum; therefore, there exists a constant $c_q>0$ such that $g_q(\bm\pi_q)\geq c_q$ for all $\bm\pi_q\in S_{\delta, q}$. For a sufficient small $\delta_\pi$, the entire integration region is contained within this set. By the properties of integration, we can bound the integral by this minimum density multiplied by the volume of the integration region:
\begin{align*}
    \Pi\left(\|\bm\pi_q - \bm\pi_q^*\|_\infty<\delta_\pi\right) &= \int _{B_\infty(\bm\pi^*_q,\delta_\pi)\cap \Delta^{K_q-1}}g_q(\bm\pi_q) \ d\bm\pi_q\\ &\geq  \int _{B_\infty(\bm\pi^*_q,\delta_\pi)\cap \Delta^{K_q-1}}c_q \ d\bm\pi_q
     \\ & = c_q\cdot \text{Volume}\left(B_\infty (\bm\pi_q^*,\delta_\pi)\cap \Delta^{K_q-1}\right).
\end{align*}
Second, we lower-bound the volume. The integration region is the intersection of a hypercube of side length $2\delta_\pi$ with the $(K_q-1)$-dimension simplex. The $(K_q-1)$-volume of such a region is proportional to $\delta_\pi$ raised to the power of the dimension. Thus, there exists a geometric constant $C_q'>0$, which depends on the geometry of the simplex but not on $\delta_\pi$, such that:
\begin{align*}
    \text{Volume}\left(B_\infty (\bm\pi_q^*,\delta_\pi)\cap \Delta^{K_q-1}\right)\geq C_q'\cdot \delta_\pi^{K_q-1}.
\end{align*}
Combining these two bounds, we obtain a lower bound for the probability:
\begin{align*}
    \Pi\left(\|\bm\pi_q-\bm\pi_q^*\|<\delta_\pi\right)\geq c_qC_q'\delta^{K_q-1}.
\end{align*}
Taking the natural logarithm of both sides and applying the properties of logarithms yields the  log-prior probability of falling within a $\delta_\pi$-neighborhood of a target $\bm\pi^*_q$ is bounded below:
    \begin{align*}
        \log\Pi\left(\|\bm \pi_q - \bm\pi_q^*\|_\infty<\delta_\pi\right)\geq \log (c_qC_q') + (K_q-1)\log(\delta_\pi) ,
    \end{align*}
    where $c_q>0$ is he lower bound of the Dirichlet density on the compact subset of the simplex, and $C_q'\delta_\pi^{K_q-1}$ is the lower bound on the volume of the integration region. For a small $\delta_\pi$, this can be simplified to:
    \begin{align*}
        \log\Pi(\|\bm\pi_q-\bm\pi_q^*\|_\infty < \delta_\pi)\geq -C_q\log(1/\delta_\pi),
    \end{align*}
    where the constant $C_q$ is approximately $K_q-1$.\\
    Since the $Q$ parameter vectors are assumed to have independent priors, the joint prior probability of all vectors simultaneously falling within their respective infinity-norm balls is the product of the marginal probabilities:
    \begin{align*}
        \Pi\left(\max_{q=1, \dots, Q}\|\bm\pi_q - \bm\pi_q^*\|_\infty<\delta_\pi\right) = \prod_{q=1}^Q\Pi(\|\bm\pi_q - \bm\pi_q^*\|_\infty<\delta_\pi).
    \end{align*}
    In log-space, this product becomes a sum:
    \begin{align*}
        \log \Pi\left(\max_{q=1,\dots, Q}\|\bm\pi_q - \bm\pi_q^*\|_\infty<\delta_\pi\right) = \sum_{q=1}^Q\log\Pi(\|\bm\pi_q - \bm\pi_q^*\|_\infty<\delta_\pi).
    \end{align*}
    Substituting the lower bound for each term, we get:
    \begin{align*}
    \sum_{q=1}^Q\log\Pi\left(\max_{q=1,\dots, Q}\|\bm\pi_q - \bm\pi_q^*\|_\infty<\delta_\pi\right)\geq \sum_{q=1}^Q(-C_q\log(1/\delta_\pi)) = -\left(\sum_{q=1}^QC_q\right)\log(1/\delta_\pi).
    \end{align*}
    By letting $C = \sum_{q=1}^QC_q$, we obtain the final result:
    \begin{align*}
        \log\Pi\left(\max_{q=1\dots, Q}\|\bm\pi_q-\bm\pi_q^*\|_\infty<\delta_\pi\right)\geq -C\log(1/\delta_\pi).
    \end{align*}
    This demonstrates that the required prior thickness bound holds for a set of multiple independent categorical distributions, preserving the same functional form as the single-parameter case. 
\end{proof}
\begin{remark}
    Notice that Bernoulli parameters with Beta prior is a special case when $K_q = 2$ for $q=1,\dots, Q$.
\end{remark}

\section{One-Step Posterior Correction Algorithm}
\subsection{Algorithm Skeleton}
At the $b$-th MCMC iteration, given the drawn full parameter set $\Theta^{(b)}$ and the global plug-in g-computation estimate $\hat{\chi}^{(b)}$, we evaluate the observed-data efficient influence function (EIF) sequentially for each subject $i = 1, \dots, n$. To accommodate both single-world (ATE) and cross-world (NIE/NDE) causal estimands within a unified computational framework, the algorithm proceeds as follows. Notice that the capital notations represent the observed data.

\paragraph{Initialization.} For subject $i$, extract $(Y_i, M_i, V_i, Z_i, C_i)$ and missingness indicator $S_i = R_{Y,i} R_{V,i}$. Compute the missingness probability $\pi_i = P(S_i=1 \mid W_i)$ which is evaluated using a separate posterior sample drawn at iteration $b$ from an independently fitted Bayesian Additive Regression Trees (BART) model. Compute the exact propensity scores $\pi^z_i = P(Z=z \mid C_i)$, which is derived from the EDPM's conditional mixture distribution
\begin{align*}
    \pi^z_i = P(Z_i=z \mid C_i; \Theta^{(b)}) = \frac{\sum_{k=1}^{M} \omega_k^{\theta, (b)} \sum_{j=1}^{N_k} \omega_{j|k}^{\psi, (b)} P(Z_i=z \mid \psi_{j|k}^{(b)}) f(C_i \mid \psi_{j|k}^{(b)})}{\sum_{k=1}^{M} \omega_k^{\theta, (b)} \sum_{j=1}^{N_k} \omega_{j|k}^{\psi, (b)} f(C_i \mid \psi_{j|k}^{(b)})}.
\end{align*}
\paragraph{Scenario I: Single-world Case ($z = z'$). } 
\begin{enumerate}
    \item {Compute Baseline Terms:} Evaluate $\mu_3^z(C_i)$
    \item Compute the baseline residual projection: $b_5 = \hat\mu_3^z(C_i) - \hat{\chi}^{(b)}$.
    \item Evaluate Conditional Projection ($b_W$):
    \begin{enumerate}
        \item If $Z_i = z$: Compute the single-world projection $b_1 = \mathbb E[Y\mid M,Z,C,S=1]$. 
       
        \item Else: $b_2 = 0$.
    \end{enumerate}
    \item Assemble: $b_W = b_1 + b_5$.
\end{enumerate}
\paragraph{Scenario II: Cross-World Case($z\neq z'$). }
\begin{enumerate}
    \item Compute Baseline Terms: Evaluate the full cross-world marginal integral $\hat\mu_4^{(b)}(C_i)$ and compute the baseline residual projection: $b_5 = \hat\mu_4^{(b)}(C_i) - \hat{\chi}^{(b)}$.
    \item Evaluate Conditional Projections ($b_W$):
    \begin{enumerate}
        \item If $Z_i = z'$: Compute the projection $b_2$ via copula mapping and Gaussian conjugate updates.
        \item Else if $Z_i = z$: Compute the projection $b_4$ via copula mapping and Gaussian conjugate updates.
        \item Assemble: $b_W = b_2 + b_4 + b_5$.
        \item Evaluate Full-Data EIF $\dot\chi_{P_f}$ (only if $S_i = 1$): \begin{enumerate}
            \item If $Z_i = z$:
            \begin{enumerate}
                \item Evaluate analytical $\mu_1^{(b)}$, map to the counterfactual $v'$ via the copula CDF, and compute $\hat\mu_2^{(b)}$.
                \item Compute the full-data EIF \begin{align*}
                    \hat{\dot \chi}_{P_{f}^{(b)}} =\frac{1}{e(C_i)}\{Y_i - \mu_1(M_i;V_{z,i},C_i) + \hat\mu_2(V_{z',i};V_{z,i},C_i) - \hat\mu_4(C_i)\} + \hat\mu_4(C_i) - \hat\chi^{(b)}.
                \end{align*}
            \end{enumerate}  
            \item Else if $Z_i = z'$:
            \begin{enumerate}
                \item  Estimate $\hat\phi_{2,P_f^{(b)}}$ by Monte Carlo Integration. 
                \item Compute the full-data  EIF
                \begin{align*}
                \hat{\dot\chi}^{(b)}_{P_f} = \hat\phi_{2,P_f}^{(b)} + \hat\mu^{(b)}_4(C_i) - \hat\chi^{(b)}.
                \end{align*}
            \end{enumerate}
        \end{enumerate}
    \end{enumerate}
\end{enumerate}
\paragraph{Final AIPW Construction.} The observed-data influence function is constructed as 
$$\hat {\dot \chi}_{P_{O'}^{(b)}}(\mathcal O_i) = \frac{S_i}{\pi_i} \hat{\dot\chi}_{P_f^{(b)}} - \frac{S_i - \pi_i}{\pi_i} b_W.$$
The final one-step corrected estimator for iteration $b$, denoted $\tilde{\chi}^{(b)}$, is computed as:
\begin{equation}
    \tilde{\chi}^{(b)} = \hat{\chi}^{(b)} + \sum_{i=1}^{n} \omega_{BB, i} \cdot \hat {\dot \chi}_{P_{O'}^{(b)}}(\mathcal O_i),
\end{equation}
where $\mathcal{O}_i = (Y_i, M_i, V_{z,i}, V_{z',i}, Z_i, C_i)$ represents the observed data. $\omega_{BB, i}$ are weights sampled from a Bayesian Bootstrap distribution $\text{Dirichlet}(n; 1, \dots, 1)$, independent of the posterior $P^{(b)}_{O'}$. The term $\hat{\dot{\chi}}_{P^{(b)}_{O'}}(\mathcal{O}_i)$ is the EIF evaluated at the $i$-th data point, wherein the nuisance functions are instantiated using the posterior samples of the parameters at the $b$-th MCMC iteration.

\subsection{Details of Computations}

To compute the components of the efficient influence function (EIF) and the targeted causal estimands, we must evaluate a sequence of nested conditional expectations. While the innermost expectation $\mu_1$ admits a closed-form analytical expression under our EDPM framework, the subsequent outer integrals defining $\mu_2, \mu_3, \mu_4$, and the cross-world components (e.g., $\phi_{2}$) are intractable and require Monte Carlo integration. For each MCMC iteration $b$, the computation proceeds as follows:

\textbf{1. Estimation of $\mu_1(m; v, z, c)$:}
This is the conditional expectation of $Y$ under the mixture model:
\begin{align*}
     \mu_{1}^{(b)}(M; V, z, C) = \frac{\sum_{k=1}^{M}\sum_{j=1}^{N_k} W_{kj}^{y,(b)} \cdot \mathbb{M}\beta_{k}^{y,(b)}}{\sum_{k=1}^{M}\sum_{j=1}^{N_k} W_{kj}^{y,(b)}},
\end{align*}
where the weights \begin{equation*}
        W_{kj, i}^{(b)} = \omega_k^{\theta,(b)}\omega_{j|k}^{\psi,(b)} \cdot \mathcal{N}\left(m_i'^{(b)} \mid \mathbb{V}\beta_k^{m,(b)}, \sigma_k^{2,m,(b)}\right) \cdot \mathcal{N}\left(v_i^{(b)} \mid \mathbb{X}\beta_{kj}^{v,(b)}, \sigma_{kj}^{2,v,(b)}\right) \cdot P(z,  {c}_i^{(b)} \mid \Psi_{kj}^{(b)}).
    \end{equation*}

\textbf{2. Estimation of $\mu_2(V_{z'}; V, {C})$:}
Draw $L$ samples $m'_{l,1} \sim P^{(b)}_{M \mid V_{z'}, Z=z', C}(\cdot)$. Then:
\begin{align*}
    \hat{\mu}_{2}^{(b)}(V_{z'}; V_z, C) = \frac{1}{L} \sum_{l=1}^{L} \mu_{1}^{(b)}(m'_{l,1}; V_z, z, C).
\end{align*}

\textbf{3. Estimation of $\mu_3(V_z; C)$:}
Draw $v'_{l,2} \sim P^{(b)}_{V_{z'} \mid V_z, Z=z, C}(\cdot)$ (using the copula sampling step described in Section \ref{app:g_comp_details}), then draw $m'_{l,2} \sim P^{(b)}_{M \mid V_{z'}=v'_{l,2}, Z=z', C}(\cdot)$.
\begin{align*}
\hat{\mu}_{3}^{(b)}(V_z; C) = \frac{1}{L} \sum_{l=1}^{L} \mu_{1}^{(b)}(m'_{l,2}; V_z, z, C).    
\end{align*}

\textbf{4. Estimation of $\mu_4({c})$ and $\chi(P)$:}
Draw $v_{l,3} \sim P^{(b)}_{V_z\mid Z=z, C}(\cdot)$, then proceed sequentially to draw $v'_{l,3}\sim P^{(b)}_{V_{z'}\mid V_z = v_{l,3},Z=z,C}(\cdot)$ and $m'_{l,3}\sim P^{(b)}_{M\mid V_
{z'}= v_{l,3}',Z=z',C}(\cdot)$ as above.
\begin{align}\label{estmu4}
     \hat{\mu}_{4}^{(b)}(C) = \frac{1}{L} \sum_{l=1}^{L} \mu_{1}^{(b)}(m'_{l,3}; v_{l,3}, z, C).
\end{align}
The plug-in estimator $\hat{\chi}^{(b)}$ is obtained by integrating $\hat{\mu}_4^{(b)}(c)$ over $P^{(b)}_{C}(c)$, which is exactly the same algorithm introduced in Section \ref{app:g_comp_details}. \\
\textbf{5. Estimation of $\phi_{2,P_f}^{(b)}$}
\begin{enumerate}
    \item[5.1] Draw $v_{l,4}\sim P^{(b)}_{V_z\mid Z=z,C}(\cdot)$,
    \item[5.2] Calculate $u = F^{(b)}_{V_{z'}\mid Z=z',C}(V_{z'})$ and $w = F^{(b)}_{V_z\mid Z=z,C=c}(v_{l,4})$,
    \item[5.3] Calculate calculate the copula function value $c^{(b)}(u,w;\rho^{(b)})$,
    \item[5.4] Draw $m_{l,4}'$ from $P^{(b)}_{M \mid V_{z'}, Z=z', c}(\cdot)$,
    \item[5.5] Compute  $\phi_{2,P_f}^{(b)}$ $$\hat\phi_{2,P_f}^{(b)} = \frac{1}{1-e^{(b)}(C)}\cdot \frac{1}{L}\sum^{L}_{l=1}\left[\mu_1^{(b)}(M;v_{l,4},z,C) - \mu_1^{(b)}(m_{l,4}';v_{l,4},z,C)\right]c^{(b)}(u,w;\rho^{(b)}),$$
    where $$e^{(b)}(C) = \frac{P^{(b)}(Z=z,C)}{P^{(b)}(C)}.
    $$
\end{enumerate}
\textbf{6. Estimation of $b_2$}
\begin{enumerate}
    \item[6.1] Draw $v_{l,5}'\sim P^{(b)}_{V_{z'}\mid M,Z=z',C}(\cdot)$ and $v_{l,5}\sim P^{(b)}_{V_z\mid Z=z,C}(\cdot)$,
    \item[6.2] Calculate $u' = F^{(b)}_{V_{z'}\mid Z=z',C}(v_{l,5}')$ and $w' = F^{(b)}_{V_z\mid Z=z,C=c}(v_{l,5})$,
    \item[6.3] Calculate calculate the copula function value $c^{(b)}(u',w';\rho^{(b)})$,
    \item[6.4] Draw $m_{l,5}'$ from $P^{(b)}_{M \mid v_{l,5}', Z=z', c}(\cdot)$
    \item[6.5] Compute  $\phi_{2,P_f}^{(b)}$ $$\hat\phi_{2,P_f}^{(b)} = \frac{1}{1-e^{(b)}(C)}\cdot \frac{1}{L}\sum^{L}_{l=1}\left[\mu_1^{(b)}(M;v_{l,5},z,C) - \mu_1^{(b)}(m_{l,5}';v_{l,5},z,C)\right]c^{(b)}(u',w';\rho^{(b)}).$$
    
\end{enumerate}
\textbf{7. Estimation of $b_4$}
\begin{enumerate}
    \item[7.1] Draw $v_{l,6}\sim P^{(b)}_{V_{z}\mid M,Z=z,C}(\cdot)$ 
    \item[7.2] Draw $v_{l,6}'\sim P^{(b)}_{V_{z'}\mid V_z = v_{l,6}, Z=z',C}(\cdot)$.
    \item[7.3] Draw $m'_{l,6}\sim P^{(b)}_{M\mid V_{z'} = v'_{l,6},Z=z,C}(\cdot)$.
    \item[7.4] Compute $b_4$
    \begin{align*}
        \hat b_4 = \frac{Z}{e^{(b)}(C)}\cdot\left[\frac{1}{L}\sum_{l=1}^L \mu^{(b)}_1(m_{l,6}';v_{l,6},z,C) -\hat\mu_{4}(C)^{(b)} \right],
    \end{align*}
    where $\hat{\mu}_4(C)^{(b)}$ is the quantity from (\ref{estmu4}).
\end{enumerate}

\begin{remark}
    For single-world case where $z=z'$, we skip the steps of sampling $v'$ from the conditional distribution $P_{V_{z'}\mid V_z=v, Z=z,C=c}(v')$. Instead, we set $v' = v$. The rest of the steps exactly align with those in the cross-world case.
\end{remark}

\subsection{The mixture form of the distributions in the Truncated EDPM}

\begin{align*}
P^{(b)}_{{C}}({c}) &=  \sum_{k=1}^{M}\omega_k^{(b)}\sum_{j=1}^{N_k}\omega_{j|k}^{(b)}P({c}|\psi_{j|k}^{(b)})\\
P^{(b)}_{V_z|Z=z, {C=c}}(v)& = \frac{\sum_{k=1}^{M}\omega_k^{(b)}\sum_{j=1}^{N_k}\omega_{j|k}^{(b)} P(V_z = v|Z=z, {C=c},\psi_{j|k}^{(b)})\times P(Z=z, {C=c};\psi_{j|k}^{(b)})}{\sum_{k=1}^{M}\omega_k^{(b)}\sum_{j=1}^{N_k}\omega_{j|k}^{(b)} P(Z=z, {C=c};\psi_{j|k}^{(b)})}\\
P^{(b)}_{M|V_z,Z=z,{C=c}}(m) & = \frac{\sum_{k=1}^{M}\omega_k^{(b)}\sum_{j=1}^{N_k}\omega_{j|k}^{(b)} P(M = m|V_z, Z=z, {C=c},\theta_{k}^{(b)},\psi_{j|k}^{(b)})\times P(V_z,Z=z,{C=c};\psi_{j|k}^{(b)})}{\sum_{k=1}^{M}\omega_k^{(b)}\sum_{j=1}^{N_k}\omega_{j|k}^{(b)} P(V_z, Z=z, {C=c};\psi_{j|k}^{(b)})}
\end{align*}
where $P(V_z, Z=z, {C=c};\psi_{j|k}^{(b)}) = P(V_1|Z=1,{C=c},\psi_{j|k}^{(b)})\times P(Z=1,{C=c};\psi_{j|k}^{(b)})$ and $P(Z=1,{C=c};\psi_{j|k}^{(b)}) = P(Z=1;\psi_{j|k}^{(b)})\times P({C};\psi_{j|k}^{(b)})$. $P(Z=1;\psi_{j|k}^{(b)})$ and $P({C};\psi_{j|k}^{(b)})$ are easily accessible.
\begin{align*}
P^{(b)}(Z=z, {C})& =\sum_{k=1}^{M}\omega_k^{(b)}\sum_{j=1}^{N_k}\omega_{j|k}^{(b)}P(Z=z;\psi_{j|k}^{(b)})\times P({C};\psi_{j|k}^{(b)})\\
    P^{(b)}(V_z,Z=z,{C}) & = \sum_{k=1}^{M}\omega_k^{(b)}\sum_{j=1}^{N_k}\omega_{j|k}^{(b)}P(V_z|Z=z,{C},\psi^{(b)}_{j|k})\times P(Z=z; \psi^{(b)}_{j|k})\times P({C};\psi_{j|k}^{(b)})\\
    P^{(b)}(M,V_z,Z=z,{C}) &= \sum_{k=1}^{M}\omega_k^{(b)}\sum_{j=1}^{N_k}\omega_{j|k}^{(b)}P(M|V_z,Z=z,{C},\theta_{j|k}^{(b)})\times P(V_z|Z=z,{C}\psi_{j|k}^{(b)})\times P(Z=z, {C};\psi_{j|k}^{(b)})
\end{align*}
\begin{align*}
    P^{(b)}(V_{z'}|V_z,Z=z, {C=c})  = \frac{P^{(b)}(V_{z'},V_z|Z=z,{C=c})}{P^{(b)}(V_z|Z=z,{C=c})}
\end{align*}
Notice that under the Gaussian Copula Assumption \ref{assump4} we have: $$F(V_{z'},V_z|Z=z,{C=c}) = \Phi_{2}\left[\Phi^{-1}\left(F(V_{z'}|Z=z',{C=c})\right),\Phi^{-1}\left(F(V_z|Z=z,{C=c})\right); \rho\right].$$
Then we can derive the density of $(V_{z'},V_z|Z=z,{C=c})$ on the basis of the copula:
\begin{align*}
f(V_{z'},V_z|Z=z,{C=c})=f(V_{z'}|Z=z',{C=c})\times f(V_z|Z=z,{C=c})\times c(u,v;\rho),
\end{align*}
where $u=F(V_{z'}|Z=z',{C=c})$, $v=F(V_z|Z=z,{C=c})$ are the CDFs of $V_{z'}$, $V_z$, and $c(u,v)$ is the copula density, given by
\begin{align*}
    c(u,v;\rho)=\frac{1}{\sqrt{1-\rho^2}}\mathrm{exp}\left[-\frac{1}{2(1-\rho^2)}\left(\Phi^{-1}(u)^2-2\rho\Phi^{-1}(u)\Phi^{-1}(v)+\Phi^{-1}(v)^2+\frac{\Phi^{-1}(u)^2+\Phi^{-1}(v)^2}{2}\right)\right].
\end{align*}
So for iteration b, 
\begin{align*}
    P^{(b)}(V_{z'},V_z|Z=z,{C=c})=P^{(b)}(V_{z'}|Z=z',{C=c})\times P^{(b)}(V_z|Z=z,{C=c})\times c(u^{(b)},v^{(b)};\rho^{(b)}).
\end{align*}
We sample $\rho^{(b)}$ from its prior. \\
To obtain $u^{(b)},v^{(b)}$, we directly integrate $V_{z},V_{z'}$ respectively on their mixture density derived as above and then we have:
\begin{align*}
    u^{(b)}=F^{(b)}_{V_{z'}|Z=z', {C=c}}(v')& = \frac{\sum_{k=1}^{M}\omega_k^{(b)}\sum_{j=1}^{N_k}\omega_{j|k}^{(b)} F(V_{z'} = v'|Z=z', {C=c},\psi_{j|k}^{(b)})\times P(Z=z', {C=c};\psi_{j|k}^{(b)})}{\sum_{k=1}^{M}\omega_k^{(b)}\sum_{j=1}^{N_k}\omega_{j|k}^{(b)} P(Z=z', {C=c}|\psi_{j|k}^{(b)})},\\
    v^{(b)} = F^{(b)}_{V_z|Z=z, {C=c}}(v)& = \frac{\sum_{k=1}^{M}\omega_k^{(b)}\sum_{j=1}^{N_k}\omega_{j|k}^{(b)} F(V_z = v|Z=z, {C=c},\psi_{j|k}^{(b)})\times P(Z=z, {C=c};\psi_{j|k}^{(b)})}{\sum_{k=1}^{M}\omega_k^{(b)}\sum_{j=1}^{N_k}\omega_{j|k}^{(b)} P(Z=z, {C=c}|\psi_{j|k}^{(b)})}.
\end{align*}

\end{document}